\documentclass[english]{article}
\usepackage{geometry}
\usepackage{float}
\usepackage{mathrsfs}
\usepackage{amsmath}
\usepackage{amssymb}
\usepackage{graphicx}
\usepackage{amsfonts}
\usepackage{bm}
\usepackage{subfigure}
\usepackage{amsthm}
\usepackage{epsfig}
\usepackage{algorithm}
\makeatletter

\@ifundefined{definecolor}{\@ifundefined{definecolor}
 {\@ifundefined{definecolor}
 {\usepackage{color}}{}
}{}
}{}

\usepackage{subfig}\usepackage[all]{xy}

\newtheorem{theorem}{Theorem}[section]
\newtheorem{lem}{Lemma}[section]
\newtheorem{rem}{Remark}[section]
\newtheorem{prop}{Proposition}[section]
\newtheorem{cor}{Corollary}[section]
\newcounter{hypA}

\usepackage{babel}\date{}

\usepackage{babel}

\makeatother

\usepackage{babel}

\begin{document}

\begin{center}

{\Large \textbf{The Alive Particle Filter}}

\vspace{0.5cm}

BY AJAY JASR$\textrm{A}^1$, ANTHONY LEE$^{2}$, CHRIS YAU$^{3}$ \& XIAOLE ZHANG$^{3}$

{\footnotesize $^{1}$Department of Statistics \& Applied Probability,
National University of Singapore, Singapore, 117546, SG.}\\
{\footnotesize E-Mail:\,}\texttt{\emph{\footnotesize staja@nus.edu.sg}}

{\footnotesize $^{2}$Department of Statistics,
University of Warwick, Coventry, CV4 7AL, UK.}\\
{\footnotesize E-Mail:\,}\texttt{\emph{\footnotesize anthony.lee@warwick.ac.uk}}

{\footnotesize $^{3}$Department of Mathematics,
Imperial College London, London, SW7 2AZ, UK.}\\
{\footnotesize E-Mail:\,}\texttt{\emph{\footnotesize c.yau@ic.ac.uk}}, \texttt{\emph{\footnotesize x.zhang11@ic.ac.uk}}

\end{center}

\begin{abstract}
In the following article we develop a particle filter for approximating Feynman-Kac models with indicator potentials.
Examples of such models include approximate Bayesian computation (ABC) posteriors associated with hidden Markov models (HMMs)
or rare-event problems. Such models require the use of advanced particle filter or Markov chain Monte Carlo (MCMC)
algorithms e.g.~Jasra et al.~(2012), to perform estimation.
One of the drawbacks of existing particle filters, is that they may `collapse', in that the algorithm may terminate early, due to the indicator potentials.
In this article, using a special case of the locally adaptive particle filter in Lee et al.~(2013), which is closely related to Le Gland \& Oudjane (2004), we use an algorithm which can deal with this latter problem,
whilst introducing a random cost per-time step. This algorithm is investigated from a theoretical perspective and several results are given which help 
to validate the algorithms and to provide guidelines for their implementation. In addition, we show how this algorithm can be used within MCMC,
using particle MCMC (Andrieu et al.~2010). Numerical examples are presented for ABC approximations of HMMs.
\\
\textbf{Key Words:} Particle Filters, Markov Chain Monte Carlo, Feynman-Kac Formulae.
\end{abstract}

\section{Introduction}\label{sec:intro}

Let $\{(\mathcal{E}_n,\mathscr{E}_n)\}_{n\geq 1}$ be a sequence of measurable spaces, $\{G_n(x) = \mathbb{I}_{\mathsf{B}_n}(x)\}_{n\geq 1}$, $(x,\mathsf{B}_n)\in \mathsf{E}_n
\times\mathscr{E}_n$, $\mathsf{B}_n\subset \mathsf{E}_n$, be a sequence of indicator potentials and $\{M_n:\mathsf{E}_{n-1}\times\mathscr{E}_n\rightarrow[0,1]\}_{n\geq 1}$, with $x_0\in\mathsf{E}_0$ a fixed point, be a sequence
of Markov kernels. Then for the collection of bounded and measurable functions $\varphi\in\mathcal{B}_b(\mathsf{E}_n)$ the $n-$time Feynman-Kac marginal is:
$$
\eta_n(\varphi) := \frac{\gamma_n(\varphi)}{\gamma_n(1)}, \quad n\geq 1
$$
assuming that $\gamma_n(\varphi) = \mathbb{E}_{x_0}[\prod_{p=1}^{n-1} G_{p}(X_p)]$ is well-defined, where $\mathbb{E}_{x_0}[\cdot]$ is the expectation w.r.t.~the law of an inhomogeneous Markov chain with transition kernels $\{M_n\}_{n\geq 1}$. Such models appear routinely in the statistics and applied probability literature including:
\begin{itemize}
\item{ABC approximations (as in, e.g.,~Del Moral et al.~(2012))}
\item{ABC approximations of HMMs (Dean et al.~2010;Jasra et al.~2012)}
\item{Rare-Events problems (as in, e.g.,~C\'erou et al.~(2012))}
\end{itemize}
In order to perform estimation for such models, one often has to resort to numerical methods such as particle filters or MCMC; see the aforementioned references. 

The basic particle filter, at time $n$ and given a collection of samples $N\geq 1$ with non-zero potential on $\mathsf{E}_{n-1}^N$, will generate samples on $\mathsf{E}_n$ using the Markov kernels $\{M_n\}_{n\geq 1}$ and then sample with replacement amongst $\{x_{n}^i\}_{1\leq i \leq N}$ according to the normalized weights $G_{n}(x^i_{n})/\sum_{j=1}^N G_{n}(x^j_{n})$. The key issue with this basic particle filter is that, at any given time, there is no guarantee that any sample $x_n^i$ lies in $\mathsf{B}_n$, and in some challenging scenarios,
the algorithm can `die-out' (or collapse), that is, that all of the samples have zero potentials. From an inference perspective, this is clearly an undesirable property and can lead to
some poor performances; for example, many particle filters display time-uniform convergence properties, which has yet to be shown for this class of particle filters.
For some classes of examples, e.g.~C\'erou et al.~(2012) or Del Moral et al.~(2012), there are some adaptive techniques which can reduce the possibility of the algorithm collapsing,
but these are not always guaranteed to work in practice.
In this article we develop a particle filter. This algorithm uses the same sampling mechanism, but the samples are generated until there
is a prespecified number that are alive. This removes the possibility that the algorithm can collapse, but introduces a random cost per time-step.
The algorithm turns out to be an important special case of the work in Lee et al.~(2013) and is closely related to Le Gland \& Oudjane (2004). 

The particle filter is analyzed from a theoretical perspective. In particular, under assumptions, we establish the following results:
\begin{enumerate}
\item{Time uniform $\mathbb{L}_p$ bounds for the particle filter estimates of $\eta_n(\varphi)$}
\item{A central limit theorem (CLT) for suitably normalized and centered particle filter estimates of $\eta_n(\varphi)$}
\item{An unbiased property of the particle filter estimates of $\gamma_n(\varphi)$}
\item{The relative variance of the particle filter estimates of $\gamma_n(\varphi)$, assuming $N=\mathcal{O}(n)$, is shown to grow linearly in $n$.}
\end{enumerate}
Whilst all of these results are classical in the literature on particle filters (C\'erou et al.~2011; Del Moral 2004), the proof in this new context requires some modifications.
In the main, these technical adjustments are associated to $\mathbb{L}_p-$bounds and CLTs for sums of random variables with a random number
of summands (in the context of 1.-2.). The technical results in 1.-2.~not only verify the correctness of the new algorithm, but suggest a substantial improvement over the standard particle filter,
at the cost of increased computational time. The results in 3.-4.~are of particular interest when using the new particle filter within MCMC methodology (a particle
MCMC (PMCMC) algorithm, (Andrieu et al.~2010)). There are variety of applications of such PMCMC algorithms, for example, when performing static parameter estimation for ABC
approximations of HMMs. The results in 3.-4.~not only allow one to construct new PMCMC algorithms, but also provide theoretical guidelines for their implementation.
It is remarked that some of these results can also be found in Le Gland \& Oudjane (2006) (with regards to 1.~2.), except for \emph{a different estimate}; this is a critical difference between the work in this article
and that in Le Gland \& Oudjane (2006). In particular, as mentioned in result 3.~our estimates of $\gamma_n(\varphi)$, and in particular $\gamma_n(1)$ are unbiased and this allows one to develop principled MCMC methodology.
We also note that,  in Le Gland \& Oudjane (2006) the authors do not give a time-uniform bound.

The structure of this article is as follows. In Section \ref{sec:algo} we provide a motivating example, ABC approximations of HMMs, for the construction of the particle filter, as well as the new particle filter
itself. In Section \ref{sec:theory} our theoretical results are provided along with some interpretation of their meaning. In Section \ref{sec:numerics} we implement the new particle
filter for the motivating example and then develop a basic PMCMC algorithm using the guidelines in Section \ref{sec:theory} for static parameter estimation associated to ABC
approximations of HMMs. In Section \ref{sec:summ} the article is concluded, with some discussion of future work. The appendix contains technical results for the theory in Section
\ref{sec:theory} and is split into three sections.

\section{Motivating Example and Algorithm}\label{sec:algo}

\subsection{Motivating Example}\label{sec:motivating_ex}

We are given a HMM with observations $\{Y_n\}_{n\geq 1}$, $Y_n\in\mathsf{Y}\subseteq\mathbb{R}^{d_y}$, hidden states $\{Z_n\}_{n\geq 0}$, $Z_n\in\mathsf{X}\subseteq\mathbb{R}^{d_x}$, $Z_0$ given.
We assume:
$$
\mathbb{P}(Y_n\in A|\{Z_n\}_{n\geq 0}) = \int_A g_{\theta}(y|z_n)dy \quad n\geq 1
$$
and
$$
\mathbb{P}(Z_n\in A|\{Z_n\}_{n\geq 0}) = \int_A f_{\theta}(z|z_{n-1})dy\quad n\geq 1
$$
with $\theta\in\Theta$ a static parameter and $dy$ Lebesgue measure.

We assume $g_{\theta}(y|x_n)$ is unknown (even up to an unbiased estimate), but one can sample from the associated distribution. In this scenario, one cannot apply a standard particle filter (or many other numerical approximation schemes).
Dean et al.~(2010) and Jasra et al.~(2012) introduce the following ABC approximation of the joint smoothing density, for $\epsilon>0$:
\begin{equation}
\pi_{\theta}(z_{1:n}|y_{1:n}) = \frac{\prod_{k=1}^n g_{\theta}^{\epsilon}(y_k|z_k)f_{\theta}(z_k|x_{k-1})}{\int_{\mathsf{X}^n}\prod_{k=1}^n g_{\theta}^{\epsilon}(y_k|z_k)f_{\theta}(z_k|z_{k-1}) dz_{1:n}}\label{eq:abc_smoothing}
\end{equation}
where
$$
g_{\theta}^{\epsilon}(y_k|x_k) = \frac{\int_{B_{\epsilon}(y_k)}g_{\theta}(u|x_k) du}{\int_{B_{\epsilon}(y_n)} du}
$$
and $B_{\epsilon}(y_k)$ is the open ball centered at $y_k$ with radius $\epsilon$. 

We let $\theta$ be fixed and omit it from our notations; it is reintroduced later on. We introduce a Feynman-Kac representation of the ABC approximation described above.
Let $\mathsf{E}_n=\mathsf{E}=\mathsf{X}\times\mathsf{Y}$ and define $G_n:\mathsf{E}\rightarrow\{0,1\}$:
$$
G_n(x) = \mathbb{I}_{\mathsf{X}\times B_{\epsilon}(y_n)}(x).
$$
Now introduce Markov kernels $\{M_n\}_{n\geq 1}$, $M_n:\mathsf{E}\times\mathbb{B}(\mathsf{X}\times\mathsf{Y})\rightarrow[0,1]$ ($\mathbb{B}(\cdot)$ are the Borel sets), with
$$
M_n(x,dx') = f(z'|z)g(u'|z') du'dz'
$$
with $x=(z,u)$. Then the ABC predictor is for $n\geq 1$:
\begin{equation}
\eta_n(\varphi) := \frac{\gamma_n(\varphi)}{\gamma_n(1)}, \label{eq:eta_n}
\end{equation}
where $\varphi\in\mathcal{B}_b(\mathsf{E})$ and 
\begin{equation}
\gamma_n(\varphi) =  \mathbb{E}_{x_0}\left[\prod_{p=1}^{n-1} G_p(X_p) \varphi(X_n)\right] = \int_{\mathsf{E}^{n}} \prod_{p=1}^{n-1} G_p(x_p)M_p(x_{p-1},dx_p) M_n(x_{n-1},dx_n)\varphi(x_n). \label{eq:gamma_n_phi}
\end{equation}
This provides a concrete example of the Feynman-Kac model in Section \ref{sec:intro}. In light of \eqref{eq:eta_n}, we henceforth refer to $\gamma_n(1)$ as the \emph{normalizing constant}. This quantity is of fundamental importance in a wide variety of statistical applications, notably in static parameter estimation, as it is equivalent to the marginal likelihood of the observed data $Y_1,\dots,Y_{n-1}$ in contexts such as the ABC approximation presented above, as can be determined from \eqref{eq:gamma_n_phi}.

\subsection{Old Filter}\label{sec:old_algo}

Now define, for $n\geq 2$:
$$
\Phi_n(\eta_{n-1})(\varphi) = \frac{\eta_{n-1}(G_{n-1}M_n(\varphi))}{\eta_{n-1}(G_{n-1})}.
$$
The standard particle filter works by sampling $x_1^1,\dots,x_1^N$ i.i.d from $M_1(x_0,\cdot)$ and
setting
$$
\eta_n^N(\varphi) = \frac{1}{N}\sum_{i=1}^N \varphi(x_n^i)\quad n\geq 1
$$
at times $n\geq 2$ sampling $x_n^1,\dots,x_n^N$ from $\Phi_n(\eta_{n-1}^N)(\cdot)$, \emph{assuming that the system has not died out}.

\subsection{Alive Filter}\label{sec:new_smc}

We now discuss an idea which will prevent the particle filter from dying out; see also Lee et al.~(2013) and Le Gland \& Oudjane (2006). Throughout we assume that $M_n(x,\mathsf{B}_n)$ is not known for each $x,n$;
if this is known, then one can develop alternative algorithms.
At time 1, we sample $x^1_1,\dots, x_1^{T_1}$ i.i.d.~from $M_1(x_0,\cdot)$, where
$$
T_1 = \inf\{n\geq N: \sum_{i=1}^n G_1(x_1^i) \geq N\}.
$$
Then, define
$$
\eta_1^{T_1}(\varphi) = \frac{1}{T_1-1}\sum_{i=1}^{T_1-1} \varphi(x_1^i).
$$
Now, at time $2$ sample $x^1_2,\dots, x_{2}^{T_2}$, conditionally i.i.d.~from $\Phi_2(\eta_1^{T_1})(\cdot)$, where
$$
T_2 = \inf\{n\geq N: \sum_{i=1}^n G_2(x_2^i) \geq N\}.
$$
This is continued until needed (i.e.~with an obvious definition of $T_3,T_4$ etc). The idea here is that, at every time step, we retain $N-1$ particles with non-zero weight, so that the algorithm never dies out, but with the additional issue that the computational cost per time-step is a random variable. The procedure is described in Algorithm~\ref{alg:alive_pf}.
We note that the approach in Le Gland \& Oudjane (2004) retains $N$ alive particles, i.e. it differs only in step $2(a)$ of Algorithm~\ref{alg:alive_pf} by sampling instead $a_{p-1}^j$ uniformly on $\{k \in \{1,\dots,T_{p-1}\} : G_{p-1}(x_{p-1}^k) = 1\}$. This seemingly innocuous difference is, however, crucial to the unbiasedness results we develop in the sequel.

\begin{algorithm}
\caption{Alive Particle Filter}
\label{alg:alive_pf}
\begin{enumerate}
\item At time 1. For $j=1,2,\dots$ until $j=:T_1$ is reached such that $G_1(x_1^j) = 1$ and $\sum_{i=1}^{j} G_1(x_1^i) = N$:
\begin{itemize}
\item Sample $x^j_1$ from $M_1(x_0,\cdot)$.
\end{itemize}
\item At time $1<p\leq n$. For $j=1,2,\dots$ until $j=:T_p$ is reached such that $G_p(x_p^j) = 1$ and $\sum_{i=1}^{j} G_p(x_p^i) = N$:
\begin{enumerate}
\item Sample $a_{p-1}^j$ uniformly from $\{k \in \{1,\dots,T_{p-1}-1\} : G_{p-1}(x_{p-1}^k) = 1\}$.
\item Sample $x^j_p$ from $M_p(x^{a_{p-1}^j}_{p-1},\cdot)$.
\end{enumerate}
\end{enumerate}
\end{algorithm}

\subsubsection{Some Remarks}

We remark that one can show (Del Moral, 2004) that  for $n\geq 2$, the normalizing constant is given by
$$
\gamma_n(1) = \prod_{p=1}^{n-1}\eta_p(G_p).
$$
Thus, a natural estimate of the normalizing constant is
$$
\gamma_n^{T_n}(1) = \prod_{p=1}^{n-1}\eta_p^{T_p}(G_p) = \prod_{p=1}^{n-1}\frac{N-1}{T_p-1}.
$$
We note that the estimates of $\eta_n$ and $\gamma_n$ are different from those considered in Le Gland \& Oudjane (2006). This is a critical point as in Proposition \ref{prop:unbiased} we show that this estimate of the normalizing constant is unbiased which is crucial for
using this idea inside MCMC algorithms. In this direction, one uses the particle filter to help propose values and there is an accept/reject step; we discuss this approach in Section \ref{sec:PMCMC}.
Again, it is clearly undesirable in an MCMC proposal, if the particle filter will collapse and so, our approach will prove to be very useful in this context.

Other than the fact that this filter will not die out, in the context of our motivating example, there is also a natural use of this idea. This is because, one can envisage the arrival of an outlier or unusual data; in such scenarios, the alive particle
filter will assign (most likely) more computational effort for dealing with this issue, which is not something that the standard filter is designed to do.

A final remark is as follows; in our example $\mathsf{B}_n = \mathsf{X} \times B_\epsilon(y_n)$ and so, as assumed in this article in general, $M_n(x,\mathsf{B}_n)$ is not known for each $x,n$.
This removes the possibility of changing measure to $\mathbb{Q}$ (in the formula for $\gamma_n(\cdot)$), (with finite dimensional marginal $\mathbb{Q}_n$)
$$
\mathbb{Q}_n(d(x_1,\dots,x_n)) \prod_{p=1}^n \frac{M_p(x_{p-1},dx_p)\mathbb{I}_{\mathsf{B}_p}(x_p)}{M_p(x_{p-1},\mathsf{B}_p)}
$$
call the Markov kernels in the product $\hat{M}_p$.
This is because the new potential at time $n$ is exactly:
$
M_n(x,\mathsf{B}_n).
$
However, one can simulate from $\hat{M}_n$ and use an unbiased estimate
of $M_n(x,\mathsf{B}_n)$ for each particle. That is, we obtain samples $z^{(1)}, z^{(2)}$ from $\hat{M}_p(x^{i}_{p-1},\cdot)$ using $R$ samples (total) from $M_p(x^{i}_{p-1},\cdot)$ and then we set $x^{i}_p=z^{(1)}$ (say) with associated weight $1/(R-1)$.
This particular procedure would then have a fixed number of particles with no possibility of collapsing. Other than the algorithm being convoluted, some
particles $x^{i}_{p-1}$ could be such that $\mathbb{E}[R]$ is prohibitively large, even though $\mathbb{E}[T_p]$ is not very large, which provides a reasonable argument against such a scheme.


\section{Theoretical Results}\label{sec:theory}

We will now present some theoretical results for the particle filter in Section \ref{sec:new_smc}. This Section could be skipped with little loss of continuity in the article; although we do provide numerical simulations to verify the behaviour that is predicted by the forthcoming theoretical results.

\subsection{Assumptions and Notations}\label{sec:assump}


Define the following sequence of Markov kernels, for $n\geq 1$:
$$
\hat{M}_n(x,dy) = \frac{M_n(x,dy) G_n(y)}{M_n(G_n)(x)}.
$$
We will make use of the following assumptions:
\begin{itemize}
\item{($M_1$): For each there exist a $\delta\in(0,1)$ such that for each  $n\geq 1$, $(x,y)\in\mathsf{E}_n^2$
$$
M_n(x,\cdot)\geq \delta M_n(y,\cdot).
$$
In addition there exist a $0<c<1$ such that for each $n\geq 1$, $x\in \mathsf{E}_{n-1}$, $M_n(\mathsf{B}_{n})(x)\wedge M_n(\mathsf{B}_{n}^c)(x)\geq c$.}
\item{($\hat{G}$): For each $n\geq 0$
$$
\sup_{(x,y)\in \mathsf{B}_{n}^2}\frac{M_{n+1}(G_n)(x)}{M_{n+1}(G_n)(y)} =\delta_n <\infty.
$$
}
\item{($\hat{M}_m$): There exist $m\geq 1$ and $\beta_p^{(m)}\in[1,\infty)$ such that for any $p\geq 1$ and $(x,y)\in \mathsf{B}_{p}$
$$
\hat{M}_{p,p+m}(x,dz) \leq \hat{\beta}_p^{(m)} \hat{M}_{p,p+m}(x,dz)  \quad \hat{M}_{p,p+m} = \hat{M}_{p+1}\hat{M}_{p+2}\dots \hat{M}_{p+m}.
$$
}
\end{itemize}
The final two conditions are $(\hat{H}_m)$ in C\'erou et al.~(2011); we also use the notation $\hat{\delta}_p^{(m)} = \prod_{q=p}^{p+m-1}\hat{\delta}_q$. 
These assumptions are exceptionally strong, but we remark that for the scenario of interest, weaker conditions have not been used in the literature. Note that in addition, in the context of ABC, the assumptions are essentially qualitative as verifying them is very difficult (even on compact state-spaces) as the likelihood density is typically
intractable. However, we still expect the phenomena reported in the below results to hold in some practical situations. We again remark that our results are relevant
for scenarios other than ABC.

In order to understand some of the subsequent results, we introduce some notations.
For a probability measure on $\mathsf{E}$ (denoted $\mathcal{P}(\mathsf{E})$) $\mu\in\mathcal{P}(\mathsf{E})$ and bounded measurable real-valued function (denoted $\mathcal{B}_b(\mathsf{E})$) $\varphi\in\mathcal{B}_b(\mathsf{E})$, we write $\mu(\varphi) := \int_{\mathsf{E}}\varphi(x)\mu(dx)$.
For $\varphi\in\mathcal{B}_b(\mathsf{E})$, $\|\varphi\|_{\infty}:=\sup_{x\in\mathsf{E}}|\varphi(x)|$.
For $\varphi\in\mathcal{B}_b(\mathsf{E})$, $\textrm{Osc}(\varphi)=\sup_{(x,y)\in\mathsf{E}^2}|\varphi(x)-\varphi(y)|$.
For $\mu,\nu\in\mathcal{P}(\mathsf{E})$, $\|\mu-\nu\|_{tv}$ denotes the total variation distance.
For a non-negative operator on $\mathcal{B}_b(\mathsf{E})$, $R(x,\cdot)$, and $\varphi\in\mathcal{B}_b(\mathsf{E})$, $R(\varphi)(x) = \int_{\mathsf{E}}\varphi(y) R(x,dy)$.
Iterates of $R$ are written $R^n(x_0,dx_n) = \int R(x_0,dx_1)\times\cdots\times R(x_{n-1},dx_n)$.
We will use the semi-group $Q_n(x,dy) = G_{n-1}(x) M_n(x,dy)$ with $n\geq 2$, with
the convention for $p<n$, $Q_{p,n}(\varphi)(x_p) = \int Q_{p+1}(x_p,dx_{p+1}) \dots$  $Q_n(x_{n-1},dx_n) \varphi(x_n)$, where $\varphi\in\mathcal{B}_b(\mathsf{E})$; when $p=n$, $Q_{n,n}$ is the identity operator. We also adopt the notation for $\mu\in\mathcal{P}(\mathsf{E})$ $\Phi_{p,n}(\mu)(\varphi) = \Phi_n\circ\cdots\circ\Phi_{p+1}(\mu)(\varphi)$, $\varphi\in\mathcal{B}_b(\mathsf{E})$, $p<n$; when $p=n$, $\Phi_{n,n}(\mu)$ is the identity operator. $\mathbb{E}$ denotes expectation w.r.t.~the stochastic process which generates the algorithm, with corresponding probability $\mathbb{P}$. It is assumed that $\prod_{\emptyset} = 1$. Note the important formula $\gamma_n(\varphi) = [\prod_{q=1}^{n-1} \eta_q(G_q)]\eta_n(\varphi)$, $\varphi\in\mathcal{B}_b(\mathsf{E})$. $\mathcal{N}(\mu,\sigma^2)$ denotes the normal distribution with mean $\mu$ and variance $\sigma^2$.
$\mathcal{G}eo(p)$ denotes a geometric random variable (with support $\{1,2,\dots\}$) with success probability $p$.

\subsection{Predictor}

In this Section, we consider the long-time behaviour of approximation of the prediction filter
$$
\eta_n^{T_n}(\varphi) = \frac{1}{T_n-1}\sum_{i=1}^{T_n-1} \varphi(z_n^i).
$$
In particular, the study of this latter behaviour w.r.t.~the algorithm in Section \ref{sec:old_algo}, is difficult due to the fact that the algorithm can collapse.
For example, in a slightly different context, it is shown in Del Moral \& Doucet (2004) that the algorithm which can die out has an upper-bound on the $\mathbb{L}_p-$error which increases with $n$
(and under strong hypotheses as in this article).
In general, we do not know of any time-uniform result for algorithms which can die out.
Below, we restrict $p\in[1,4]$ as this is all that is needed for a strong law of large numbers.
The additional technical results associated to Theorem \ref{theo:time_uniform} can be found in Appendix \ref{app:tech_pred}.

\begin{theorem}\label{theo:time_uniform}
Assume ($M_1$). Then  for any $p\in[1,4]$ there exist a $C_p<\infty$ such that for any $n\geq 1$, $N\geq 2$, $\varphi\in\mathcal{B}_b(\mathsf{E_n})$:
$$
\mathbb{E}[|\eta_n^{T_n}(\varphi)-\eta_n(\varphi)|^p]^{1/p} \leq \frac{C_p\|\varphi\|_{\infty}}{\sqrt{N-1}}.
$$
\end{theorem}

\begin{proof}
Throughout $C_p$ is a finite positive constant (that does not depend upon $n$) whose value may change from line to line.
The proof follows that of Theorem 7.4.4 of Del Moral (2004). We have, using eq.~7.24 of Del Moral (2004)
$$
\mathbb{E}[|\eta_n^{T_n}(\varphi)-\eta_n(\varphi)|^p]^{1/p} \leq \sum_{q=1}^n\mathbb{E}[|[\Phi_{q,n}(\eta_q^{T_q}) - \Phi_{q,n}(\Phi_q(\eta_{q-1}^{T_{q-1}}))](\varphi)|^{p}]^{1/p}.
$$
WLOG we suppose $\textrm{Osc}(\varphi)\leq 1$.
Now for $x\in\mathsf{B}_{q}$ we define the Markov kernel $P_{q,n}(x,\cdot) := Q_{q,n}(x,\cdot)/Q_{q,n}(1)(x)$ with associated
Dobrushin coefficient $\beta(P_{q,n})=\sup_{(x,y)\in \mathsf{B}_{q}^2}\|P_{q,n}(x,\cdot)-P_{q,n}(y,\cdot)\|_{tv}$ and also set
$r_{q,n}=$  \\$\sup_{(x,y)\in \mathsf{B}_{q}^2} Q_{q,n}(1)(x)/Q_{q,n}(1)(y)$. Then following the calculations of pp.245--246 of Del Moral (2004), we have
$$
\mathbb{E}[|\eta_n^{T_n}(\varphi)-\eta_n(\varphi)|^p]^{1/p} \leq \sum_{q=1}^n r_{q,n}\beta(P_{q,n})
\mathbb{E}[|[\eta_q^{T_q} - \Phi_q(\eta_{q-1}^{T_{q-1}}) ](\bar{Q}_{q,n}^N(\varphi))|^p]^{1/p}.
$$
where $\bar{Q}_{q,n}^N(\varphi)$ is defined in pp.~246 of Del Moral (2004) and note that $\|\bar{Q}_{q,n}^N(\varphi)\|_{\infty}\leq 1$.
Application of Corollary \ref{cor:cond_lp} gives:
$$
\mathbb{E}[|\eta_n^{T_n}(\varphi)-\eta_n(\varphi)|^p]^{1/p} \leq \frac{C_p}{\sqrt{N-1}} \sum_{q=1}^n r_{q,n}\beta(P_{q,n}).
$$
The sum on the R.H.S.~can be bounded uniformly in $n$ by using standard arguments in C\'erou et al.~(2011) or Del Moral (2004) and are hence omitted. This concludes the proof.
\end{proof}

\subsection{Central Limit Theorem}

In this Section we consider the asymptotic properties of a suitably normalized and centered estimate of the predictor; a central limit theorem. We note that such a result is not a direct corollary of existing CLTs for particle filters in the literature (e.g.~Del Moral (2004)). The additional technical results associated to  Theorem \ref{theo:clt} can be found in Appendix \ref{app:tech_clt}.

Here we write:
$$
\eta_{n}^{N-1}(\varphi) = \frac{1}{N-1} \sum_{i=1}^N \varphi(X_n^i)
$$
the empirical measure of the first $N-1$ sampled particles at time $n$. The convergence in probability (written $\rightarrow_{\mathbb{P}}$) weak convergence 
(written $\Rightarrow$)
results are as $N\rightarrow\infty$.

\begin{theorem}\label{theo:clt}
Assume ($M_1$). Then for any $n\geq 1$, $\varphi\in\mathcal{B}_b(\mathsf{E}_n)$ we have:
$$
\sqrt{T_n-1} [\eta_{n}^{T_n} - \eta_n](\varphi) \Rightarrow \mathcal{N}(0,\sigma^2_n(\varphi))
$$
where, setting $\varphi_n = \varphi-\eta_n(\varphi)$ 
\begin{eqnarray*}
\sigma^2_n(\varphi) & = & \eta_n(G_n) \eta_n(\varphi_n^2) + [\sigma^2_{n-1}(Q_n(\varphi_n))]/[\eta_n(G_n)\eta_{n-1}(G_{n-1})] \quad n\geq 2\\
\sigma^2_1(\varphi) & = & \eta_1(G_1) \eta_1((\varphi-\eta_1(\varphi))^2)
\end{eqnarray*}
or equivalently for any $n\geq 1$
\begin{equation}
\sigma^2_n(\varphi) = \eta_n(G_n)\sum_{q=1}^n \frac{\gamma_{q}(G_q)^2}{\gamma_{n}(G_n)^2}\eta_q([Q_{q,n}(\varphi_n)-\eta_q(Q_{q,n}(\varphi_n))]^2).
\label{eq:clt_asymp_var}
\end{equation}
\end{theorem}

\begin{proof}
Our proof proceeds via induction. For the case $n=1$, by Lemma \ref{lem:clt_main_random_to_deterministic} we
need only deal with the term
$$
\sqrt{(N-1)\eta_1(G_1)} [\eta_{n}^{N-1} -\eta_1](\varphi).
$$
Then one need only apply the CLT for i.i.d.~bounded random variables; this yields the result with
$$
\sigma_1^2(\varphi) = \eta_1(G_1) \eta_1((\varphi-\eta_1(\varphi))^2).
$$

We assume the result for $n-1$ and consider $n$. Let $\varphi_n = \varphi - \eta_n(\varphi)$ then we have
\begin{equation}
\sqrt{T_n-1} [\eta_{n}^{T_n} - \eta_n](\varphi) = 
\sqrt{T_n-1} [\eta_{n}^{T_n} -\Phi_n(\eta_{n-1}^{T_{n-1}}) ](\varphi_n) 
+ \sqrt{T_n-1}\Phi_n(\eta_{n-1}^{T_{n-1}}) ](\varphi_n).
\label{eq:clt_prf1}
\end{equation}
Now the first term on the R.H.S.~of \eqref{eq:clt_prf1} can be written as 
\begin{eqnarray}
\sqrt{T_n-1} [\eta_{n}^{T_n} -\Phi_n(\eta_{n-1}^{T_{n-1}}) ](\varphi_n) 
& = & \sqrt{T_n-1} [\eta_{n}^{T_n} -\Phi_n(\eta_{n-1}^{T_{n-1}}) ](\varphi_n) -  
\sqrt{(N-1)\eta_n(G_n)} [\eta_{n}^{N-1} -\Phi_n(\eta_{n-1}^{T_{n-1}}) ](\varphi_n) \nonumber \\ & &
+ \sqrt{(N-1)\eta_n(G_n)} [\eta_{n}^{N-1} -\Phi_n(\eta_{n-1}^{T_{n-1}}) ](\varphi_n).
\label{eq:clt_prf2}
\end{eqnarray}
In addition, the second term on the R.H.S.~of \eqref{eq:clt_prf1} can be written as 
\begin{eqnarray}
\sqrt{T_n-1}\Phi_n(\eta_{n-1}^{T_{n-1}}) ](\varphi_n) & = &
\bigg[\sqrt{\frac{T_n-1}{T_{n-1}-1}}- \sqrt{\frac{\eta_{n-1}(G_{n-1})}{\eta_{n}(G_{n})}}\bigg]\sqrt{T_{n-1}-1}\Phi_n(\eta_{n-1}^{T_{n-1}}) ](\varphi_n) \nonumber \\ & & 
+ \sqrt{T_{n-1}-1}\sqrt{\frac{\eta_{n-1}(G_{n-1})}{\eta_{n}(G_{n})}}\Phi_n(\eta_{n-1}^{T_{n-1}}) ](\varphi_n)
\label{eq:clt_prf3}
\end{eqnarray}
By Lemma \ref{lem:clt_main_random_to_deterministic} the first term on the R.H.S.~of \eqref{eq:clt_prf2} converges in probability to zero.
Also, by Lemma \ref{lem:conv_t_n}, Theorem \ref{theo:time_uniform} (which provides a strong law of large numbers) and the induction hypothesis ($\eta_{n-1}(Q_n(\varphi_n))=0$), the first term on the R.H.S.~of \eqref{eq:clt_prf3} converges in probability to zero. 
Thus, by a corollary to Slutsky's theorem, we can consider the weak convergence of 
$$
\sqrt{(N-1)\eta_n(G_n)} [\eta_{n}^{N-1} -\Phi_n(\eta_{n-1}^{T_{n-1}}) ](\varphi_n) + 
\sqrt{\frac{\eta_{n-1}(G_{n-1})}{\eta_{n}(G_{n})}}\sqrt{T_{n-1}-1}\Phi_n(\eta_{n-1}^{T_{n-1}}) ](\varphi_n) := A(N) + B(N).
$$
We now consider the characteristic function:
\begin{equation}
\mathbb{E}[\exp\{it(A(N) + B(N))\}] = \mathbb{E}[\{\mathbb{E}[\exp\{itA(N)\}|\mathscr{F}_{n-1}]-e^{-\tilde{\sigma}_n^2(\varphi_n) t^2/2}\}\exp\{itB(N)\}] + 
e^{-\tilde{\sigma}_n^2(\varphi_n) t^2/2} \mathbb{E}[\exp\{itB(N)\}]
\label{eq:clt_prf4}
\end{equation}
where
$$
\tilde{\sigma}_n^2(\varphi_n) = \eta_n(G_n) \eta_n(\varphi_n^2)
$$
and $\mathscr{F}_{n-1}$ is the filtration generated by the particle system up-to time $n-1$.
We deal with the limit of the expectations on the R.H.S.~of \eqref{eq:clt_prf4} independently. We will show that 
$\mathbb{E}[\exp\{itA(N)\}|\mathscr{F}_{n-1}]-e^{-\tilde{\sigma}_n^2(\varphi_n) t^2/2}$ will converge in probability to zero, by using
Theorem A3 of Douc \& Moulines (2008). To that end, we note that for any $\varphi\in\mathcal{B}_b(\mathsf{E}_n)$, we have that
$$
\frac{1}{N-1}\sum_{i=1}^{N-1}[\varphi(X_n^i) - \Phi(\eta_{n-1}^{T_{n-1}})(\varphi)]
$$
will converge in probability to zero (for example, by controlling the second moment with the Marcinkiewicz-Zygmund (M-Z) inequality). Then by Theorem \ref{theo:time_uniform} as
$\Phi(\eta_{n-1}^{T_{n-1}})(\varphi)$ converges almost surely to $\eta_n(\varphi)$ that $\eta_n^{N-1}(\varphi)$ will converge in probability
to $\eta_n(\varphi)$. Using this result it follows easily that
$$
\eta_n(G_n) \frac{1}{N-1}\sum_{i=1}^{N-1} [\varphi(X_n^i) - \Phi(\eta_{n-1}^{T_{n-1}})(\varphi)]^2
$$
converges in probability to $\tilde{\sigma}_n^2(\varphi_n)$. This verifies the first condition of Theorem A3 of Douc \& Moulines (2008) (eq.~(31) of that paper).
As $\varphi$ is bounded, it is straightforward to verify the second (Lindeberg-type) condition of Theorem 13 of Douc \& Moulines (2008) (eq.~(32) of that paper). Thus,
application of this latter theorem shows that $\mathbb{E}[\exp\{itA(N)\}|\mathscr{F}_{n-1}]-e^{-\tilde{\sigma}_n^2(\varphi_n) t^2/2}$
converges in probability to zero. Then by the induction hypothesis 
\begin{equation}
B(N) \Rightarrow \mathcal{N}(0,[\sigma^2_{n-1}(Q_n(\varphi_n))]/[\eta_n(G_n)\eta_{n-1}(G_{n-1})])\label{eq:clt_prf5}.
\end{equation}
Thus 
$$
\{\mathbb{E}[\exp\{itA(N)\}|\mathscr{F}_{n-1}]-e^{-\tilde{\sigma}_n^2(\varphi_n) t^2/2}\}\exp\{itB(N)\} \rightarrow_{\mathbb{P}} 0.
$$
Application of Theorem 25.12 of Billingsley (1995), shows that 
$$
\lim_{N\rightarrow \infty}  \mathbb{E}[\{\mathbb{E}[\exp\{itA(N)\}|\mathscr{F}_{n-1}]-e^{-\tilde{\sigma}_n^2(\varphi_n) t^2/2}\}\exp\{itB(N)\}] = 0.
$$
Thus, returning to \eqref{eq:clt_prf4}, we consider $e^{-\tilde{\sigma}_n^2(\varphi_n) t^2/2} \mathbb{E}[\exp\{itB(N)\}]$. Noting \eqref{eq:clt_prf5}
and again applying Theorem 25.12 of Billingsley (1995) we yield
$$
\lim_{N\rightarrow \infty} \mathbb{E}[\exp\{itB(N)\}] = e^{-[\sigma^2_{n-1}(Q_n(\varphi_n))]/[\eta_n(G_n)\eta_{n-1}(G_{n-1})] t^2/2}.
$$
Thus, we have proved that
$$
\sqrt{T_n-1} [\eta_{n}^{T_n} - \eta_n](\varphi) \Rightarrow \mathcal{N}(0,\sigma^2_n(\varphi))
$$
where
$$
\sigma^2_n(\varphi) = \eta_n(G_n) \eta_n(\varphi_n^2) + [\sigma^2_{n-1}(Q_n(\varphi_n))]/[\eta_n(G_n)\eta_{n-1}(G_{n-1})].
$$
The verification of the formula \eqref{eq:clt_asymp_var} for the asymptotic variance follows standard calculations and is omitted.
\end{proof}

\begin{rem}\label{rem:clt}
The formula for the asymptotic variance of the particle filter in Section \ref{sec:old_algo}, 
pp.~304 of Del Moral (2004) is
\begin{equation}
\sum_{q=1}^n \frac{\gamma_{q}(1)^2}{\gamma_n(1)^2}\eta_q([Q_{q,n}(\varphi_n)-\eta_q(Q_{q,n}(\varphi_n))]^2).
\label{eq:old_asymp_var}
\end{equation}
Comparing to the asymptotic variance formula \eqref{eq:clt_asymp_var}, this latter formula is certainly smaller if for each $1\leq q< n$
\begin{equation}
\eta_q(G_q)^2\leq \eta_n(G_n).
\label{eq:var_ineq}
\end{equation}
An alternative interpretation of \eqref{eq:var_ineq} is if $\{V_n\}_{n\geq 1}$ is a Markov chain with transition kernels $\{M_n\}_{n\geq 1}$ then \eqref{eq:var_ineq} is
$$
\mathbb{P}(V_q\in\mathsf{B}_{q}| v_1\in\mathsf{B}_{1},\dots, v_{q-1}\in\mathsf{B}_{q-1})^2 \leq
\mathbb{P}(V_n\in\mathsf{B}_{n}| v_1\in\mathsf{B}_{1},\dots, v_{n-1}\in\mathsf{B}_{n-1}).
$$
Whilst this can be difficult to verify in general, if the spaces are $\mathsf{E}_n=\mathsf{E}$, $n\geq 1$, potentials $G_n=G$ for each $n\geq 1$ and the Markov kernels are such that for each $n\geq 1$, $x\in\mathsf{E}$, $M_n(x,\cdot) = \nu(\cdot)$ for some $\nu\in\mathcal{P}(\mathsf{E})$, then the L.H.S.~of \eqref{eq:var_ineq} is $\nu(G)^2$ and the R.H.S.~is $\nu(G)$; so in this ideal scenario, the new algorithm asymptotically outperforms the old one with regards to variance. In general, one might believe that \eqref{eq:clt_asymp_var} is smaller that the formula \eqref{eq:old_asymp_var}, as its leading term (when $q=n$) is smaller and the condition \eqref{eq:var_ineq} can certainly hold in many examples.
\end{rem}

\begin{rem}
One can also prove a CLT for the estimate of the filter; a direct corollary is, under ($M_1$), using Theorem \ref{theo:clt}, we have
$$
\sqrt{T_n-1} \Big[\frac{\eta_{n}^{T_n}(G_n\varphi)}{\eta_{n}^{T_n}(G_n)} - \frac{\eta_n(G_n\varphi)}{\eta_n(G_n)}\Big]
\Rightarrow\mathcal{N}(0,\tilde{\sigma}^2_n(\varphi))
$$
where
$$
\tilde{\sigma}^2_n(\varphi) = \sum_{q=1}^n \frac{\gamma_{q}(G_q)^2}{\gamma_{n}(G_n)^2\eta_n(G_n)}\eta_q([Q_{q,n}(G_n[\varphi_n-\eta_n(G_n\varphi)])-\eta_q(Q_{q,n}(G_n[\varphi_n-\eta_n(G_n\varphi)]))]^2).
$$
In comparison, the asymptotic variance of the estimate in Theorem 4 of Le Gland \& Oudjane (2006)(which differs to the one in this article) has asymptotic variance 
$$
\hat{\sigma}^2_n(\varphi) = \sum_{q=1}^n \frac{\gamma_{q}(1)\gamma_{q}(G_q)}{\gamma_{n}(G_n)^2}\eta_q([Q_{q,n}(G_n[\varphi_n-\eta_n(G_n\varphi)])-\eta_q(Q_{q,n}(G_n[\varphi_n-\eta_n(G_n\varphi)]))]^2).
$$
Thus there is an asymptotic difference between the two procedures. In general, our approach is better with regards to asymptotic variance if
$$
\frac{\gamma_{q}(G_q)^2}{\gamma_{n}(G_n)^2\eta_n(G_n)} \leq \frac{\gamma_{q}(1)\gamma_{q}(G_q)}{\gamma_{n}(G_n)^2}
$$
or using the Markov chain interpretation in Remark \ref{rem:clt}:
$$
\frac{\mathbb{P}(V_1\in\mathsf{B}_{1},\dots, V_{q-1}\in\mathsf{B}_{q-1},V_q\in\mathsf{B}_{q})}
{\mathbb{P}(V_n\in\mathsf{B}_n|v_1\in\mathsf{B}_{1},\dots, v_{n-1}\in\mathsf{B}_{n-1})} \leq
\mathbb{P}(V_1\in\mathsf{B}_{1},\dots, V_{q-1}\in\mathsf{B}_{q-1}).
$$
In general, one cannot say which is preferable, but in the case: the spaces are $\mathsf{E}_n=\mathsf{E}$, $n\geq 1$, potentials $G_n=G$ for each $n\geq 1$ and the Markov kernels are such that for each $n\geq 1$, $x\in\mathsf{E}$, $M_n(x,\cdot) = \nu(\cdot)$ for some $\nu\in\mathcal{P}(\mathsf{E})$,
both the L.H.S.~and R.H.S.~of the inequality are equal.
\end{rem}

\subsection{Normalizing Constant}

Define the estimate of the normalizing constant:
$$
\gamma_n^{T_n}(\varphi) := \bigg\{\prod_{p=1}^{n-1}\frac{N-1}{T_{p}-1}\bigg\}\eta_n^{T_n}(\varphi).
$$
The technical results used in this  Section can be found in Appendix \ref{app:tech_nc}.

\subsubsection{Unbiasedness}

\begin{prop}\label{prop:unbiased}
We have for any $n\geq 1$, $N\geq 2$ and $\varphi\in\mathcal{B}_b(\mathsf{E}_n)$, that
$$
\mathbb{E}[\gamma_n^{T_{n}}(\varphi)] = \gamma_n(\varphi).
$$
\end{prop}

\begin{proof}
The proof uses the standard Martingale difference decomposition in Del Moral (2004), with some additional expectation properties that need to be proved.
The case $n=1$ follows from Lemma \ref{lem:tech_res}, so we assume $n\geq 2$.
We remark that for $p\in\{2,\dots,n\}$:
$$
\gamma_p^{T_p}(1)\Phi_p(\eta_{p-1}^{T_{p-1}})(Q_{p,n}(\varphi)) = 
\gamma_{p-1}^{T_{p-1}}(1)\eta_{p-1}^{T_{p-1}}(Q_{p-1,n}(\varphi))
$$
and hence that
$$
\gamma_n^{T_{n}}(\varphi) - \gamma_n(\varphi)
= \sum_{p=1}^n \gamma_p^{T_p}(1)[\eta_p^{T_p} - \Phi_p(\eta_{p-1}^{T_{p-1}})](Q_{p,n}(\varphi)).
$$
Then by Lemma \ref{lem:tech_res}, it follows that 
$$
\mathbb{E}[\gamma_p^{T_p}(1)[\eta_p^{T_p} - \Phi_p(\eta_{p-1}^{T_{p-1}})](Q_{p,n}(\varphi))|\mathscr{F}_{p-1}] = 0
$$
and hence that 
$$
\mathbb{E}[\gamma_n^{T_{n}}(\varphi) - \gamma_n(\varphi)] = 0
$$
from which we easily conclude the result.
\end{proof}


\subsubsection{Non-Asymptotic Variance Theorem}

Below the term $ \sum_{s=1}^n \frac{\hat{\delta}_s^{(m)}\hat{\beta}_s^{(m)}}{\eta_s(G_s)}$ is as in C\'erou et al.~(2011).
The expressions and interpretations  for $\hat{\delta}_s^{(m)}\hat{\beta}_s^{(m)}$ can be found in Section \ref{sec:assump}.
In addition, $(\eta_n^{T_n})^{\odot 2}$ is the $U-$statistic that is formed from our empirical measure $\eta_n^{T_n}$ and
$(\eta_n^{T_n})^{\otimes 2}$  is the corresponding $V-$statistic. In addition $(\gamma_n^{T_n})^{\otimes 2}(F) = \gamma_{n-1}^{T_{n-1}}(1)^2
(\eta_n^{T_n})^{\otimes 2}(F)$ for $F\in\mathcal{B}_b(\mathsf{E}^2)$.

\begin{prop}\label{prop:non_asymp}
Assume ($\hat{H}_m$). Then for any $n\geq 2$, $N\geq 3$
$$
N > \sum_{s=1}^n \frac{\hat{\delta}_s^{(m)}\hat{\beta}_s^{(m)}}{\eta_s(G_s)} \Rightarrow\quad\mathbb{E}\Big[\Big(\frac{\gamma_n^{T_n}(1)}{\gamma_n(1)} - 1\Big)^2\Big] \leq \frac{4}{N} \sum_{s=1}^n \frac{\tilde{\delta}_s\hat{\delta}_s^{(m)}\hat{\beta}_s^{(m)}}{\eta_s(G_s)}.
$$
\end{prop}

\begin{proof}
The result follows essentially from C\'erou et al.~(2011). To modify the proof to our set-up, we will prove that for $F:\mathsf{E}^2\rightarrow\mathbb{R}_+$ (where the expectation on the L.H.S.~is 
w.r.t.~the stochastic process that generates the SMC algorithm)
\begin{equation}
\mathbb{E}[(\gamma_n^{T_n})^{\otimes 2}(F)] \leq \Big(\frac{N-1}{N-2}\Big)^n \mathbb{E}_{\xi}[\eta_1^{\otimes 2} C_{\xi_1}Q_2^{\otimes 2} C_{\xi_2}\dots Q_n^{\otimes 2}C_{\xi_n}(F)]\label{eq:fund_ineq}
\end{equation}
where for each $n\geq 1$, independently
$$
\mathbb{P}_{\xi}(\xi_n = 1) = 1 - \mathbb{P}_{\xi}(\xi_n = 0) = \frac{1}{N-1}
$$
with corresponding joint expectation $\mathbb{E}_{\xi}$ and $C_1(F)(x,y) = F(x,x)$, $C_0(F)(x,y) = F(x,y)$. Once \eqref {eq:fund_ineq} is proved
this gives a verification of Lemma 3.2, eq.~(3.3)
of C\'erou et al.~(2011), given this, the rest of the argument then follows Proposition 3.4 of C\'erou et al.~(2011) and Theorem 5.1 and Corollary 5.2 in C\'erou et al.~(2011) (note that the fact that we have an upper-bound with $\alpha=0$ (as in C\'erou et al.~(2011)) does not modify the result). We will write expectations w.r.t.~the probability space associated to the particle system enlarged with the
(independent) $\{\xi_n\}_{n\geq 1}$ as $\overline{\mathbb{E}}_{\xi}$.

Thus, we consider the proof of \eqref{eq:fund_ineq}. We have 
$$
\mathbb{E}[(\gamma_n^{T_n})^{\otimes 2}(F)|\mathscr{F}_{n-1}]  =  \gamma_{n}^{T_n}(1)^2\mathbb{E}[(\eta_n^{T_n})^{\otimes 2}(F)|\mathscr{F}_{n-1}].
$$
Now
\begin{eqnarray*}
\mathbb{E}[(\eta_n^{T_n})^{\otimes 2}(F)|\mathscr{F}_{n-1}] & = & \mathbb{E}\Big[\frac{T_n-2}{T_n-1}(\eta_n^{T_n})^{\odot 2}(F) + \frac{1}{T_n-1} \eta_n^{T_n}(C(F))\Big|\mathscr{F}_{n-1}\Big]\\
& \leq & \Phi_n(\eta_{n-1}^{T_{n-1}})^{\otimes 2}(F) + \frac{1}{N-1} \Phi_n(\eta_{n-1}^{T_{n-1}})(C(F)) \\
& \leq & \Big(\frac{N-1}{N-2}\Big)\overline{\mathbb{E}}_{\xi}[\Phi_n(\eta_{n-1}^{T_{n-1}})^{\otimes 2}(C_{\xi_n}(F))|\mathscr{F}_{n-1}]
\end{eqnarray*}
where we have used $(T_n-2)/(T_n-1) \leq 1$, $1/(T_n-1) \leq 1/(N-1)$ and Lemmas \ref{lem:tech_lem} and \ref{lem:tech_res} to obtain the second line.
Thus we have that
\begin{eqnarray*}
\mathbb{E}[(\gamma_n^{T_n})^{\otimes 2}(F)|\mathscr{F}_{n-1}] & \leq &  \gamma_{n}^{T_n}(1)^2 \Big(\frac{N-1}{N-2}\Big) \overline{\mathbb{E}}_{\xi}[\Phi_n(\eta_{n-1}^{T_{n-1}})^{\otimes 2}(C_{\xi_n}(F))|\mathscr{F}_{n-1}] \\
&\leq &  \gamma_{n-1}^{T_{n-1}}(1)^2 \Big(\frac{N-1}{N-2}\Big) \overline{\mathbb{E}}_{\xi}[(\eta_{n-1}^{T_{n-1}})^{\otimes 2}(Q_n C_{\xi_n}(F))|\mathscr{F}_{n-1}].
\end{eqnarray*}
Using the above inequality, one can repeat the argument inductively to deduce \eqref{eq:fund_ineq}. This completes the proof of the Proposition.
\end{proof}

\begin{rem}
The significance of the result is simply that if 
$$
\sup_s \frac{\hat{\delta}_s^{(m)}\hat{\beta}_s^{(m)}}{\eta_s(B_{\epsilon}(y_s))} < c
$$
then if $N>cn$ the relative variance will be constant in $n$. This will be useful for the PMCMC algorithm in Section \ref{sec:PMCMC}.
\end{rem}

\section{Numerical Implementation}\label{sec:numerics}

\subsection{ABC Filtering}\label{sec:linear_gaussian}

\subsubsection{Linear Gaussian model}

To investigate the alive particle filter, we consider the following linear Gaussian state space model (with all quantities one-dimensional):
\begin{align*}
Z_n=&Z_{n-1}+V_n,\\
Y_n=&2Z_n+W_n,\;\;\;\;\;\;\; t\geq 1
\end{align*}
where $V_n\sim{}{\mathcal{N}}(0, \sigma_v^2)$, and independently $W_n\sim{}{\mathcal{N}}(0, \sigma_w^2)$. Our objective is to fit an ABC approximation of this HMM; this is simply to investigate the algorithm constructed in this article.

\subsubsection{Set up}

 Data are simulated from the (true) model for $T = 5000$ time steps and $\sigma_v^2\in\{0.1, 1, 5\}$ and $\sigma_w^2\in{}\{0.1, 1, 5\}$. For $n\in \{1,\dots,T\}$, if $p_n\geq \frac{1}{500}$, where $p_t\stackrel{\textrm{i.i.d.}}{\sim}\mathcal{U}_{[0,1]}$ (the uniform distribution on $[0,1]$), we have $Y_n = c$, where $c\in\{80, 90, \dots\, 140, 150\}$. Recall $B_{\epsilon}(y)=\{u:|u-y|<\epsilon\}$  and we consider a fixed sequence of $\epsilon$ which values belong to set \{5, 10, 15\}, i.e. $\epsilon\in\{5, 10, 15\}$. We compare the alive particle filter to the approach in Jasra et al.~(2012).

%

The proposal dynamics are as described in Section \ref{sec:new_smc}.
For the approach in Jasra et al.~(2012), $N=2000$ and we resample every time. For the alive particle filter, we used $N=1500$ particles; this is to keep the computation time approximately equal.
We also estimate the normalizing constant via the alive filter at each time step and compare it with `exact' values obtained via the Kalman filter in the limiting case $\epsilon = 0$. To assess the performance in normalizing constant estimation, the relative variance is estimated via independent runs of the procedure.

Our results are constructed as two parts. In the first part, we compare the performance of two particle filters under different scenarios. In the second part, we focus on examples where the approach in Jasra et al.~(2012) collapses.
All results were averaged over $50$ runs. We note that, with regards to the results in this Section and the approach in Le Gland \& Oudjane (2004); generally similar conclusions can be drawn with regards to comparison to the approach in Jasra et al.~(2012). 


\subsubsection{Part \uppercase\expandafter{\romannumeral1}}

In this part, the analyses of the alive particle filter were completed in approximately 115 seconds and approximately 103 seconds were taken for the approach in Jasra et al.~(2012) (which we just term the particle filter). Our results are shown in Figures \ref{fig:fig1}-\ref{fig:fig6}.

Figure \ref{fig:fig1} displays the log relative error for the alive filter to the particle filter. We present the time evolution of the $\mathbb{L}_1$ log relative error between the `exact' and estimated first moment. From our results, the mean log relative error for each panel is $\{0.06, 0.04, 0.07\}$. 
Figure \ref{fig:fig2} plots the absolute $\mathbb{L}_1$ error of the alive particle filter error across time. These results indicate, in the scenarios under study, that both filters are performing about the same time with regards to estimating the filter. This is unsurprising as both methods use essentially the same information,
and the outlying values do not lead to a collapse of the particle filter. In addition, the behaviour in Figure \ref{fig:fig2}, which is predicted in Theorem \ref{theo:time_uniform} under strong assumptions, appears to hold in a situation where the state-space is non-compact.

In Figure \ref{fig:fig3}, we show the time evolution of the log of the normalizing constant estimate for three approaches, i.e. Kalman filter (black `--' line), new ABC filter (red `-$\cdot$-' line) and SMC method (blue `$\cdot \cdot$' line). Figure \ref{fig:fig4} displays the (log) relative variance of
the estimate of the normalizing constant via the alive particle filter, when using the Kalman filter as the ground truth. In Figure \ref{fig:fig3}, there is unsurprisingly a bias in estimation of the normalizing constant, as the ABC approximation is not exact, i.e. $\epsilon \neq 0$. In Figure \ref{fig:fig4}
the linear decay in variance proven in Proposition \ref{prop:non_asymp} is demonstrated (although under a log transformation).


In Figure \ref{fig:fig5} and \ref{fig:fig6}, we show the number of particles used at each time step (that is to achieve $N$ alive particles) of the alive filter (Figure \ref{fig:fig5}) and the number of alive particles for the standard particle filter (Figure \ref{fig:fig6}).
Both Figures illustrate the effect of outlying data, where the alive filter has to work `harder' (i.e.~assigns more computational effort), whereas the standard filter just loses particles.

%

\vspace{1cm}
\begin{figure}[h]
\begin{center}
\scalebox{0.2}{\includegraphics{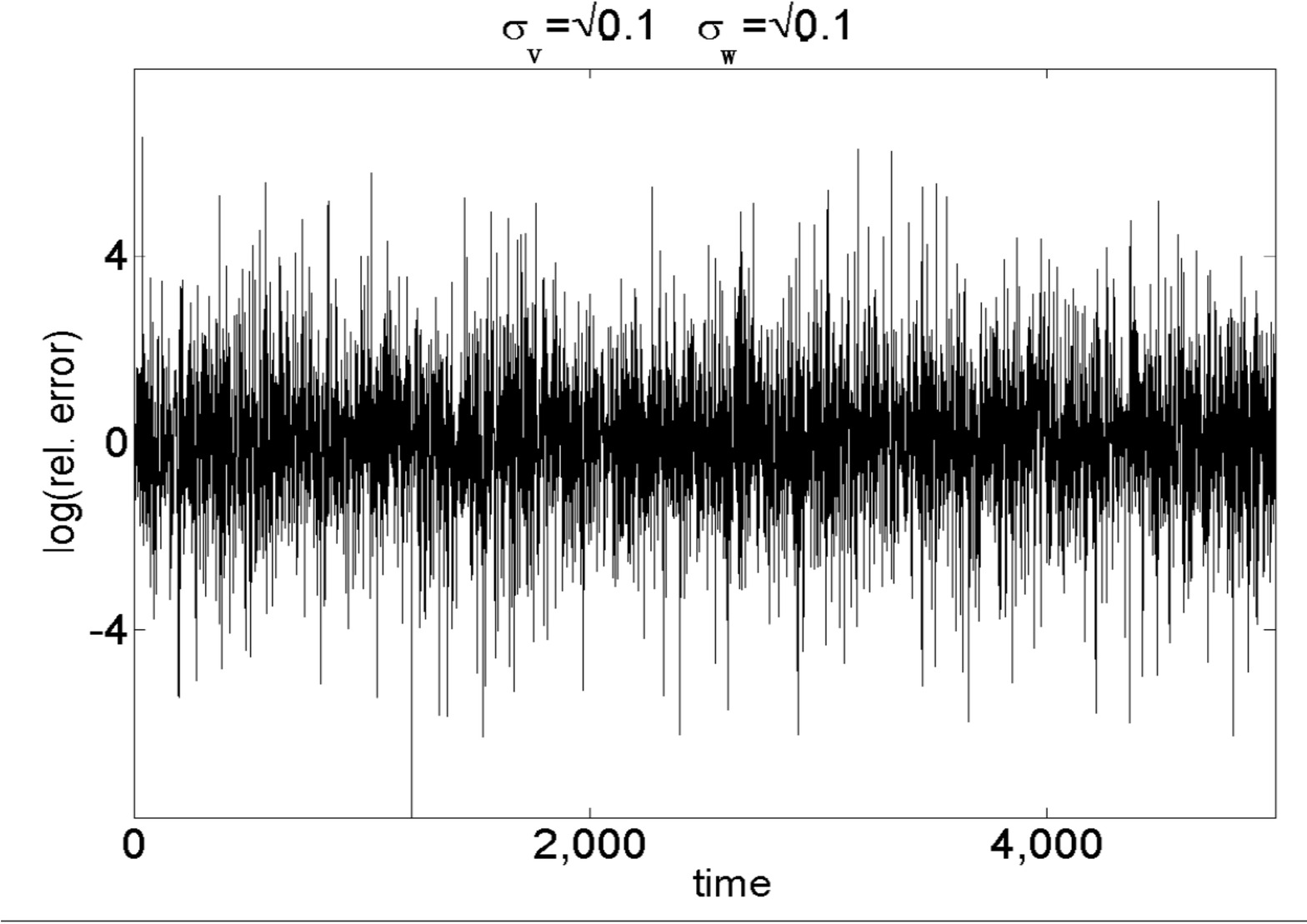}}\scalebox{0.2}{\includegraphics{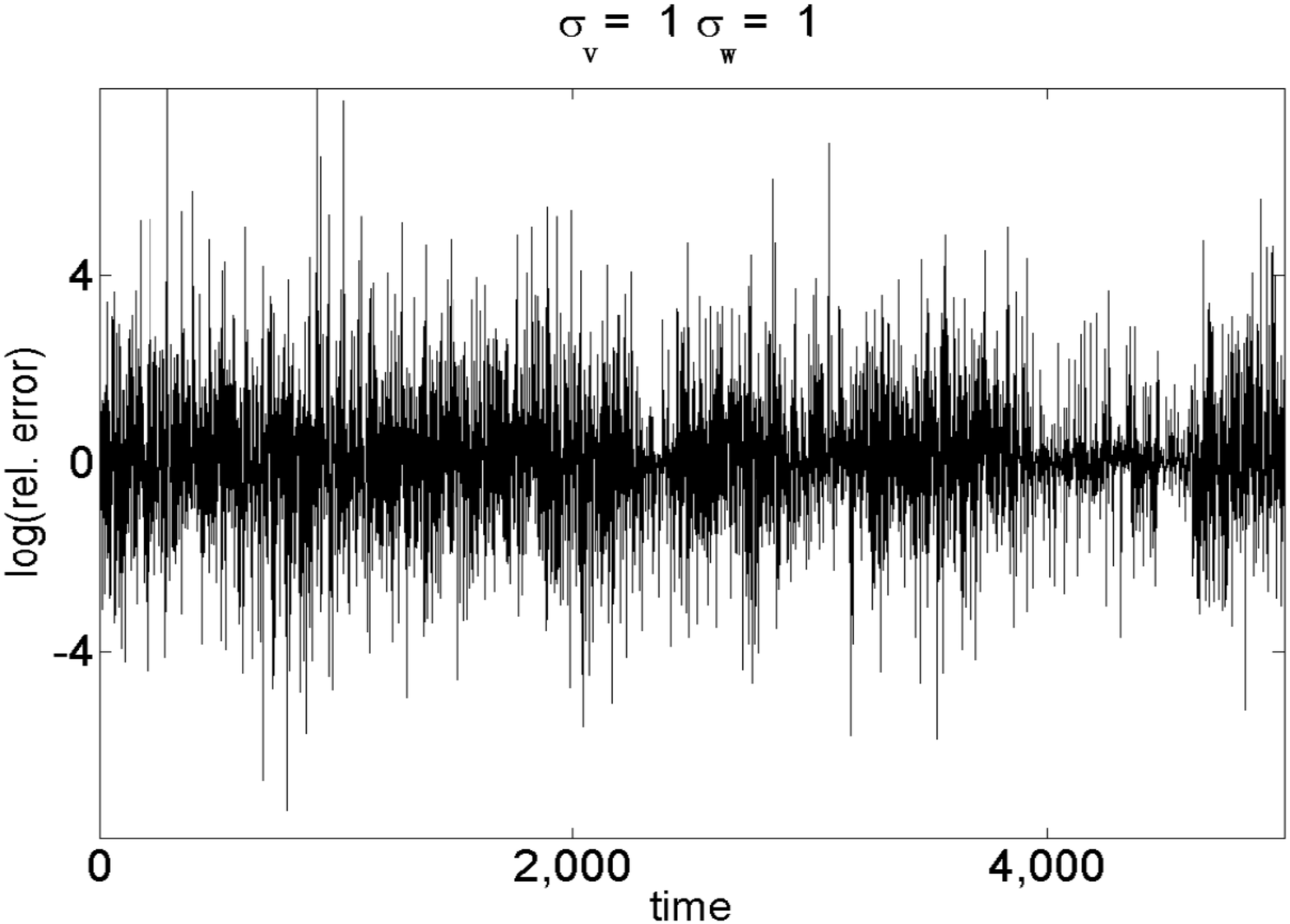}}\scalebox{0.2}{\includegraphics{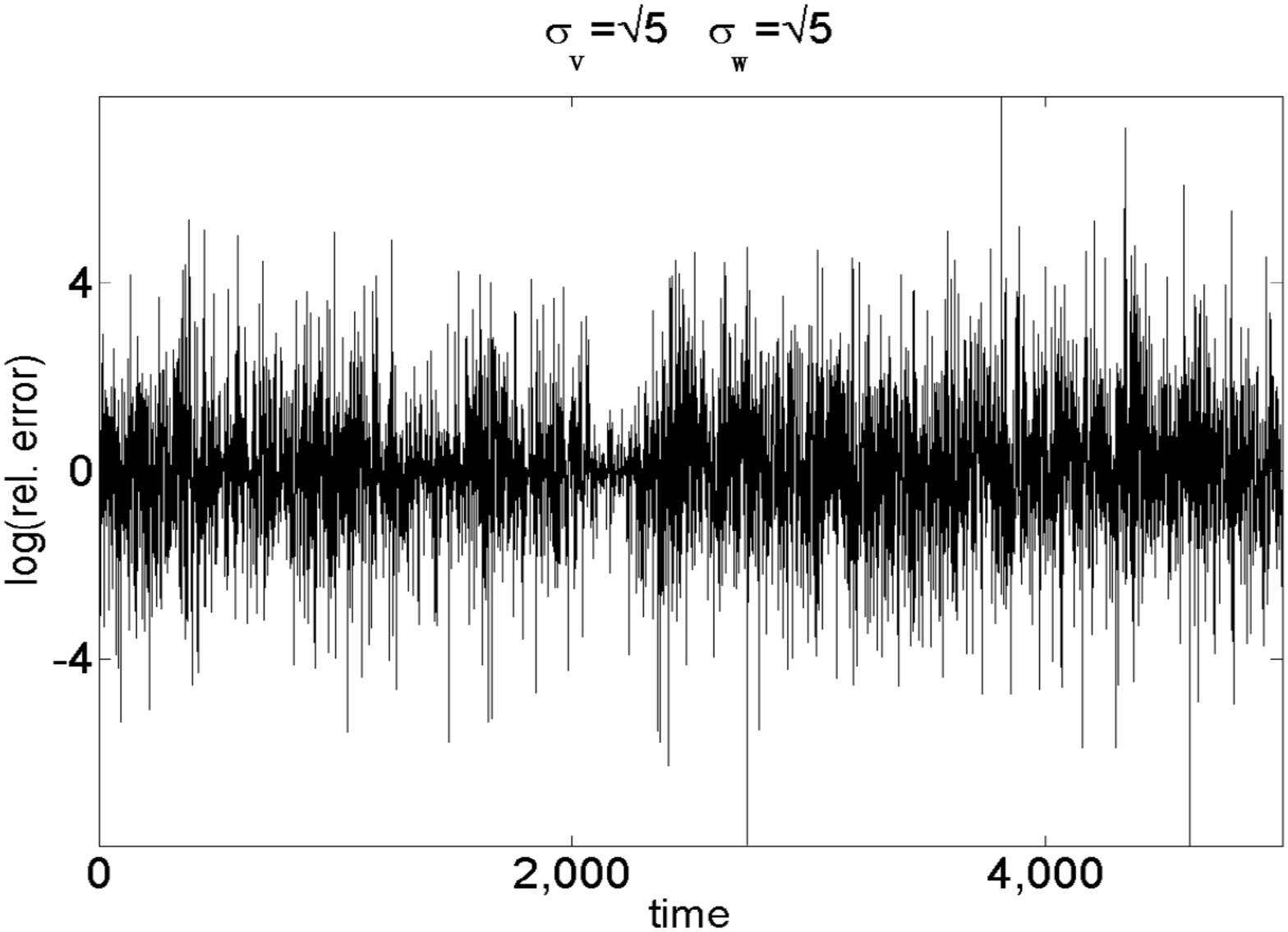}}
\caption{{\footnotesize Estimation error of the first moment for the linear state space model. Each panel displays (Log) the ratio of $\mathbb{L}_1$ error of the alive filter to old filter.}\label{fig:fig1}}
\end{center}
\end{figure}

\begin{figure}[h]
\begin{center}
\scalebox{0.2}{\includegraphics{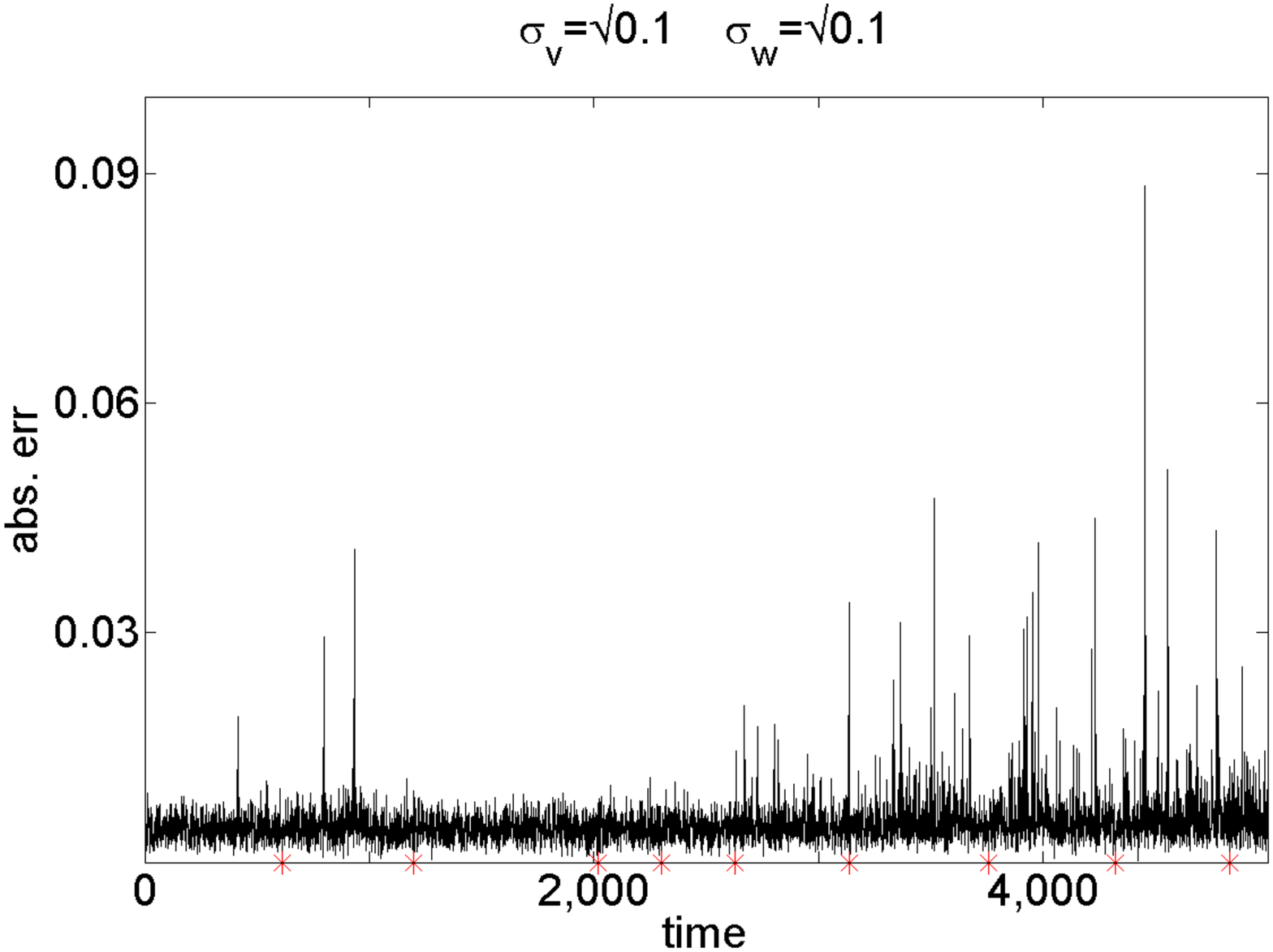}}\scalebox{0.2}{\includegraphics{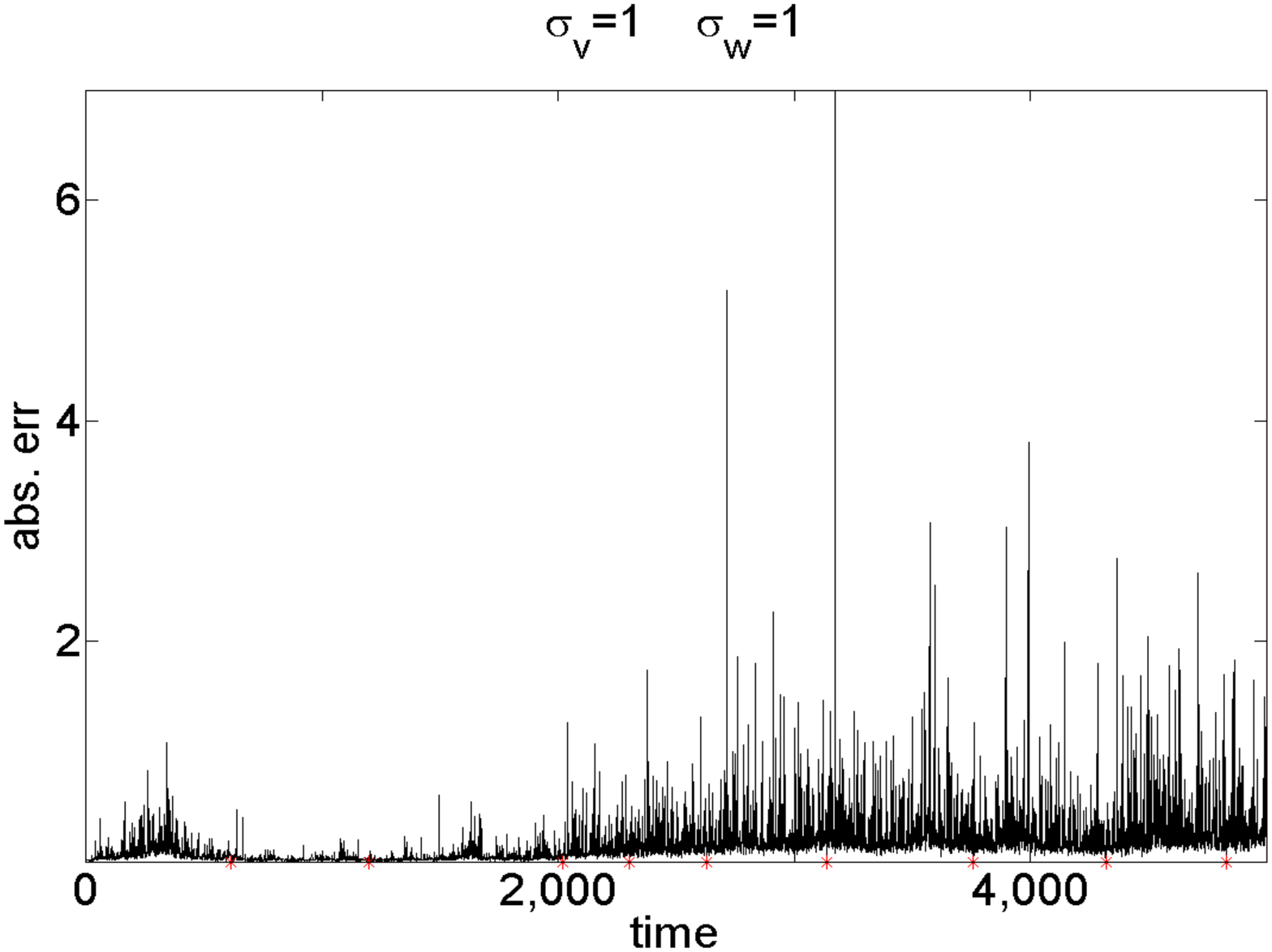}}\scalebox{0.2}{\includegraphics{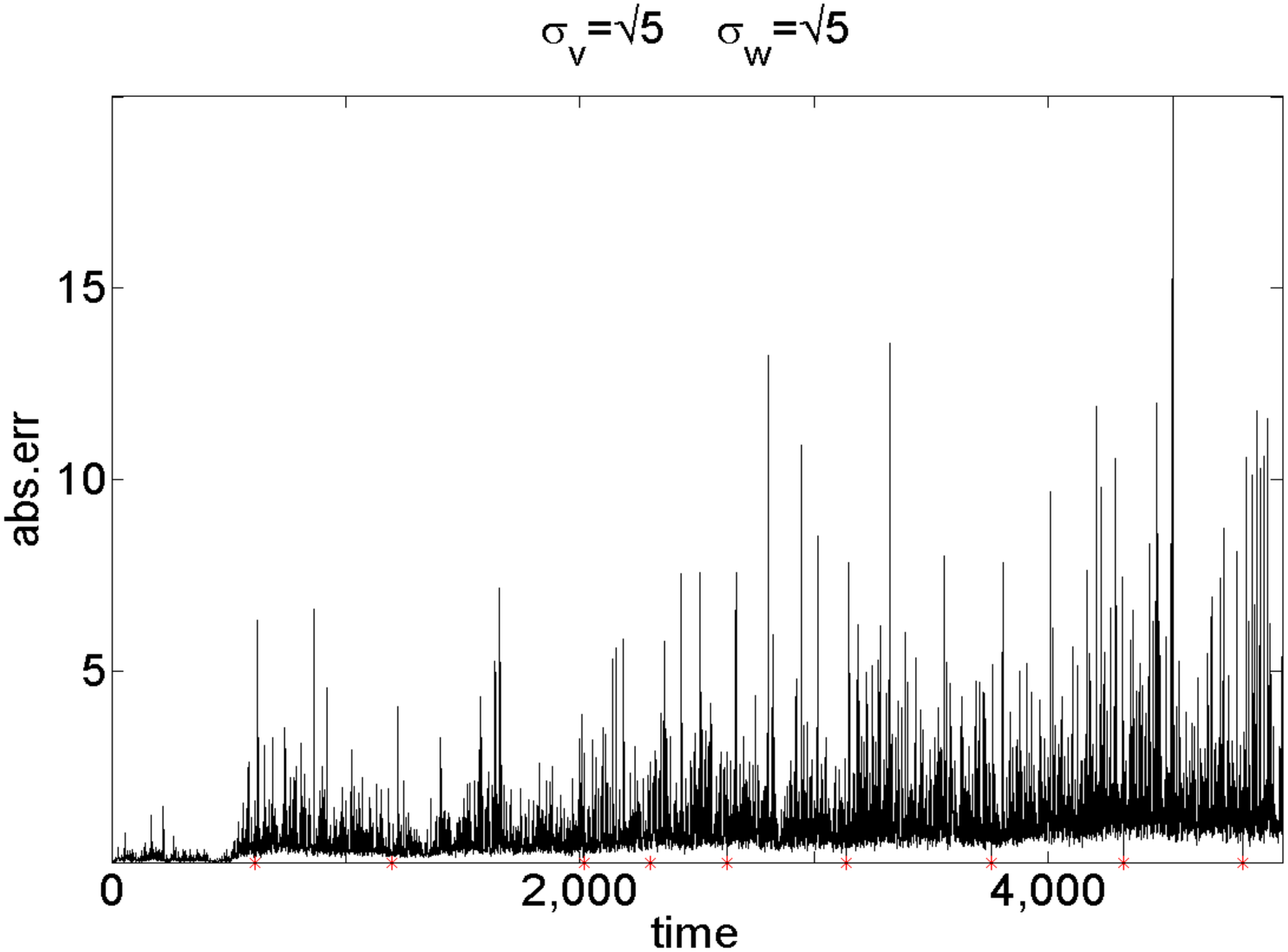}}
\caption{{\footnotesize Estimation error of the first moment for the linear state space model using the alive particle filter (red `$\star$' indicates the x-axis position of outlier)}\label{fig:fig2}}
\end{center}
\end{figure}

\begin{figure}[h]
\begin{center}
\scalebox{0.2}{\includegraphics{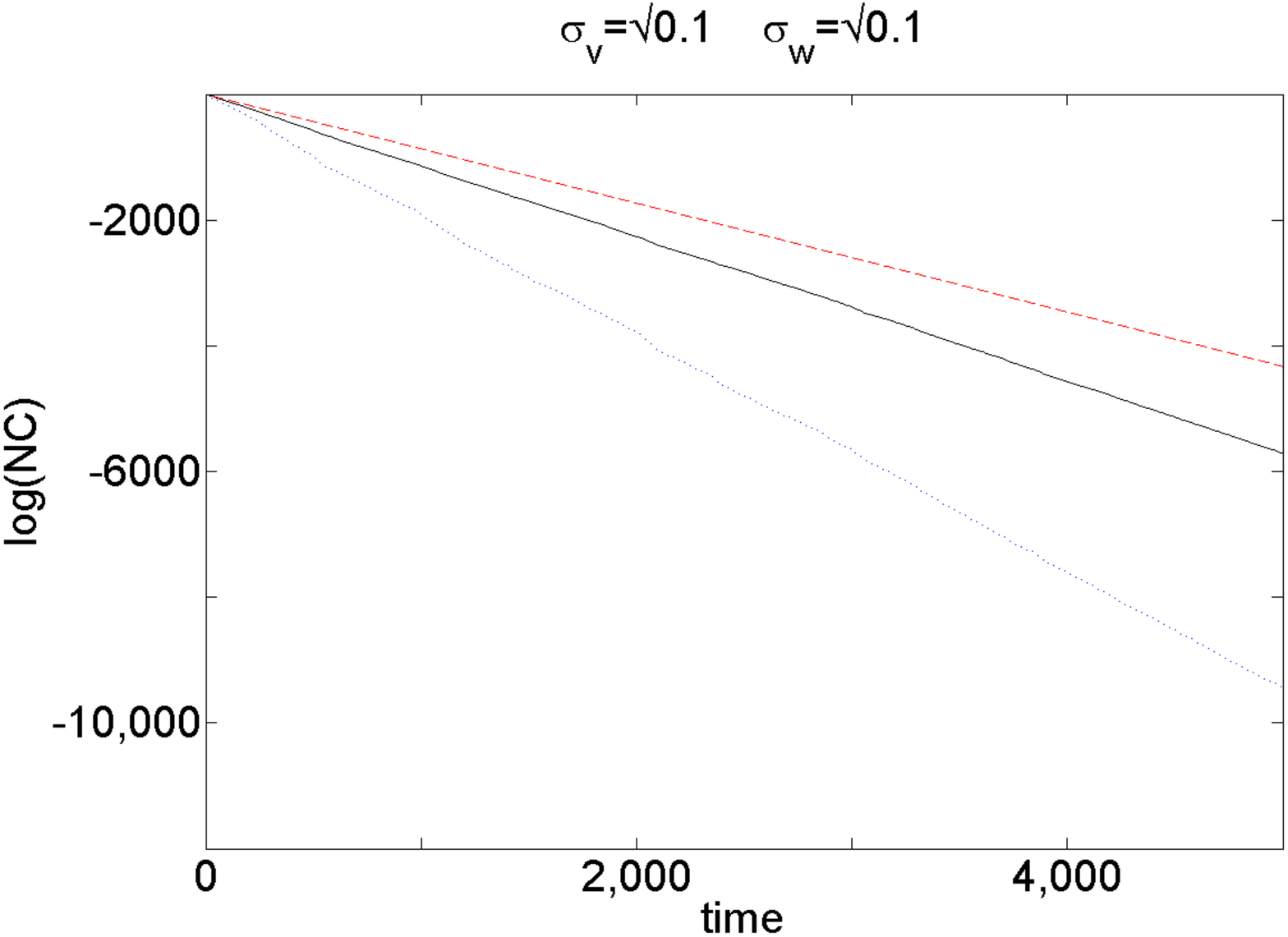}}\scalebox{0.2}{\includegraphics{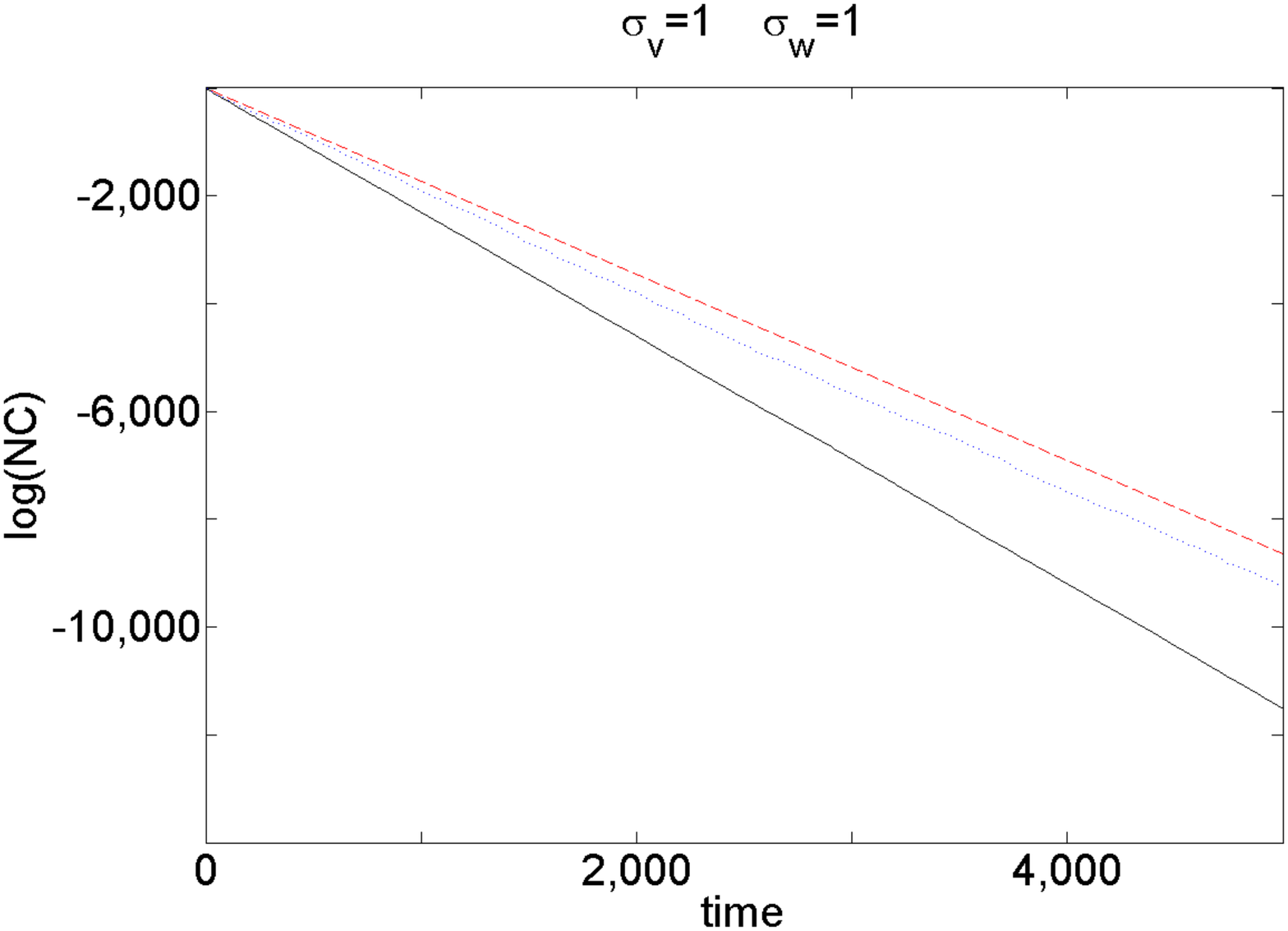}}\scalebox{0.2}{\includegraphics{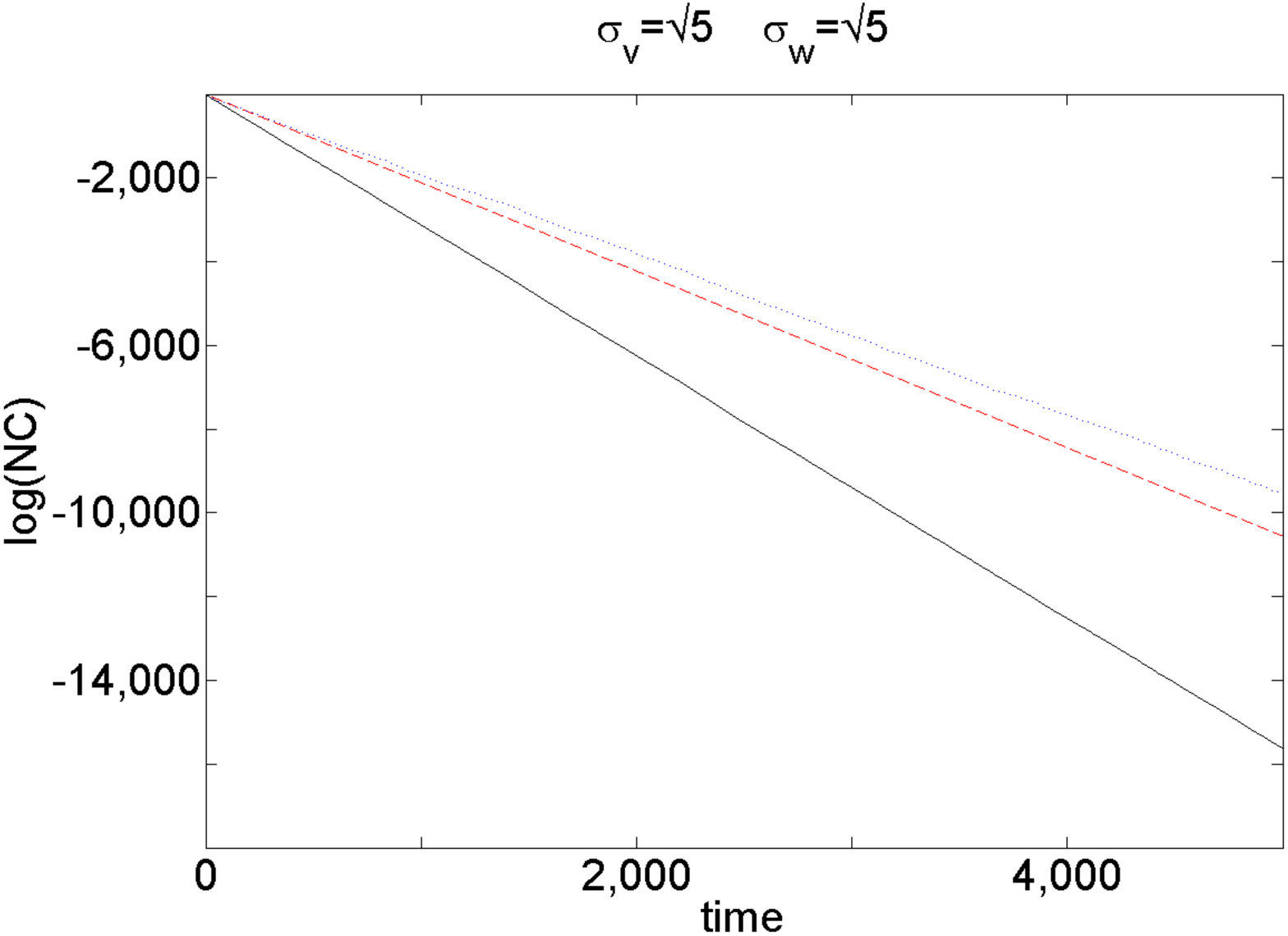}}
\caption{{\footnotesize Estimated normalizing constant for the linear state space model: Kalman filter (black '--'), alive filter (red '-$\cdot$-') and old filter (blue '$\cdot \cdot$'). Each panel displays the estimated normalizing constant across time.}\label{fig:fig3}}
\end{center}
\end{figure}

\begin{figure}[h]
\begin{center}
\scalebox{0.2}{\includegraphics{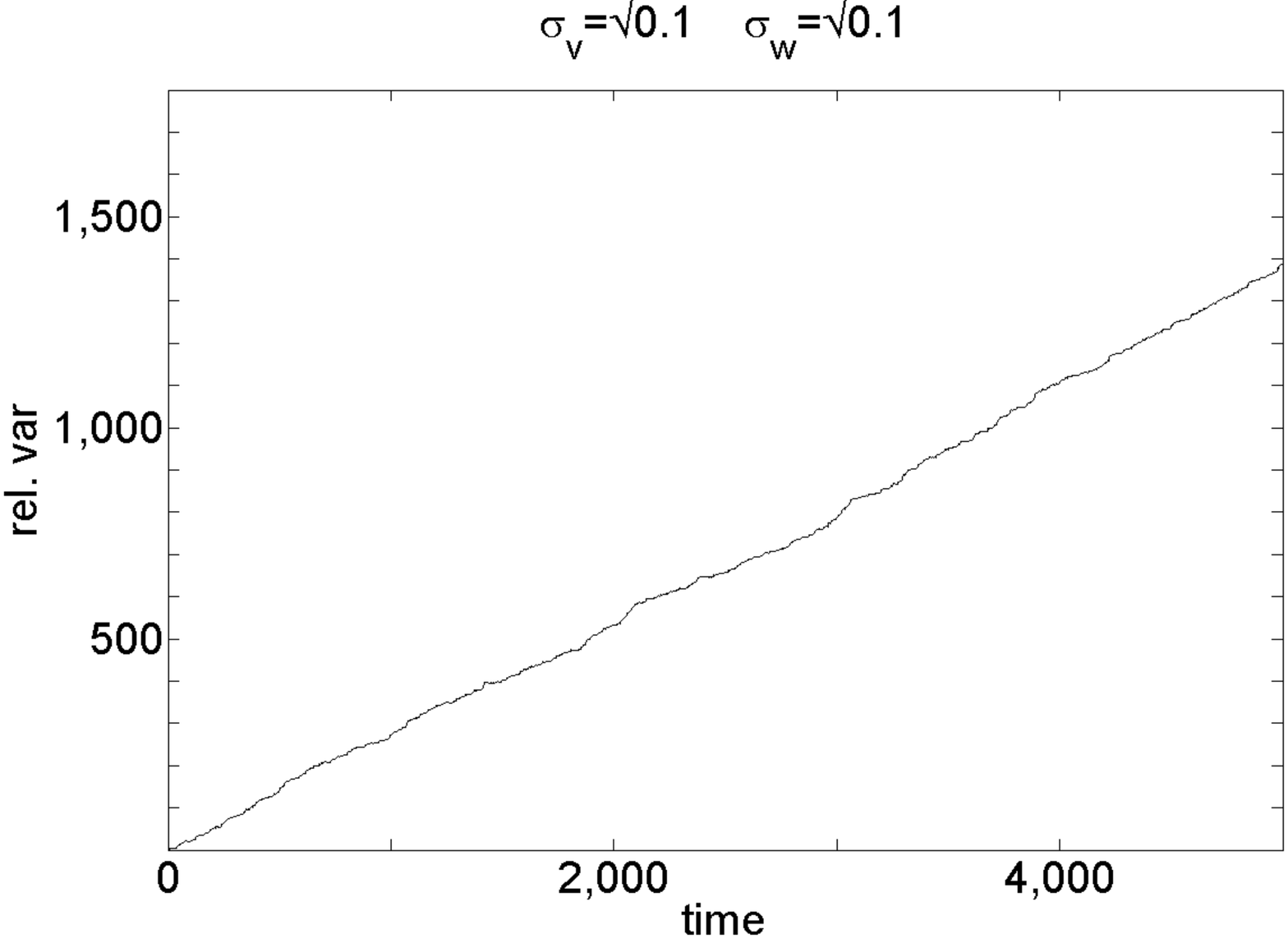}}\scalebox{0.2}{\includegraphics{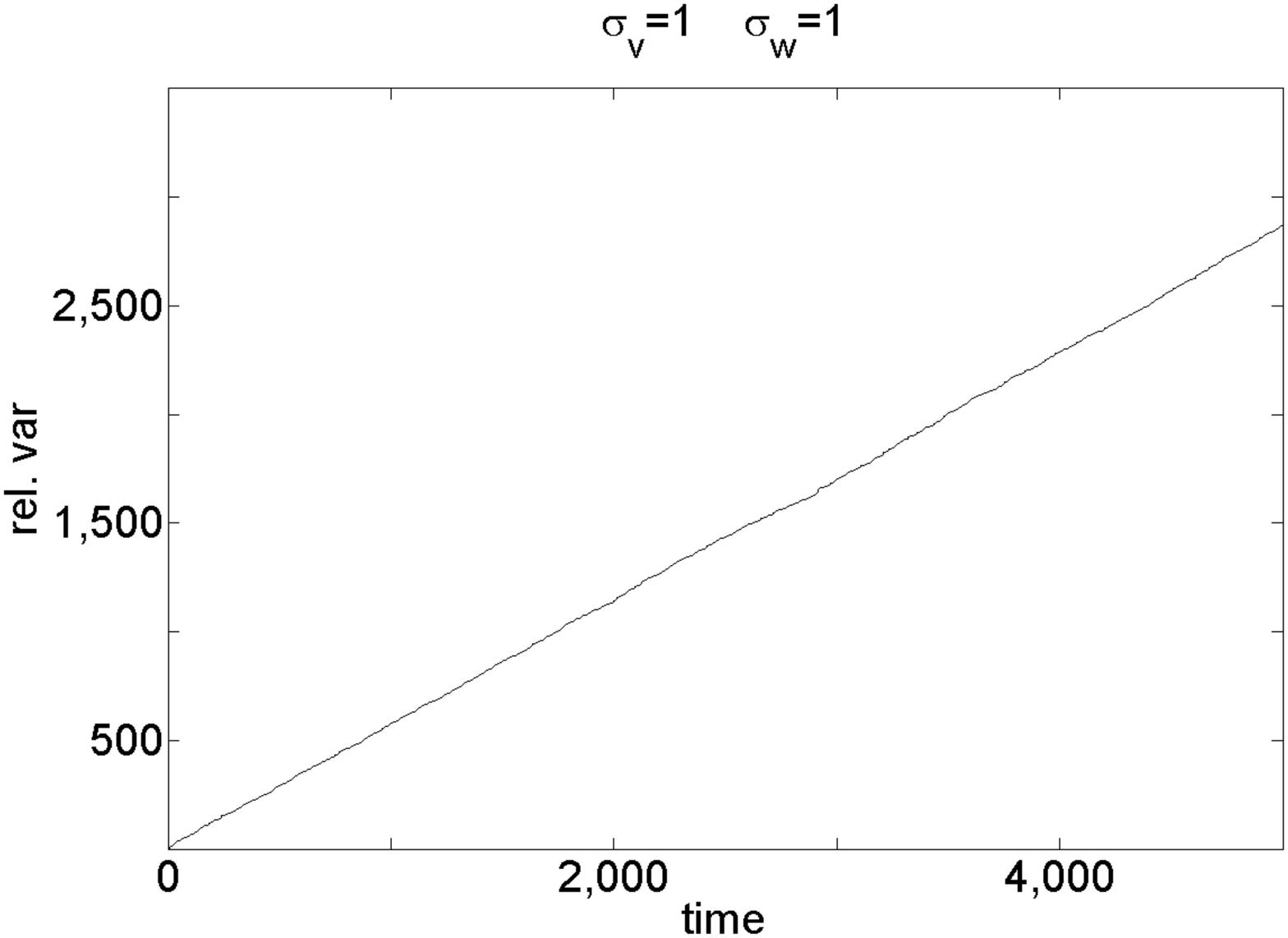}}\scalebox{0.2}{\includegraphics{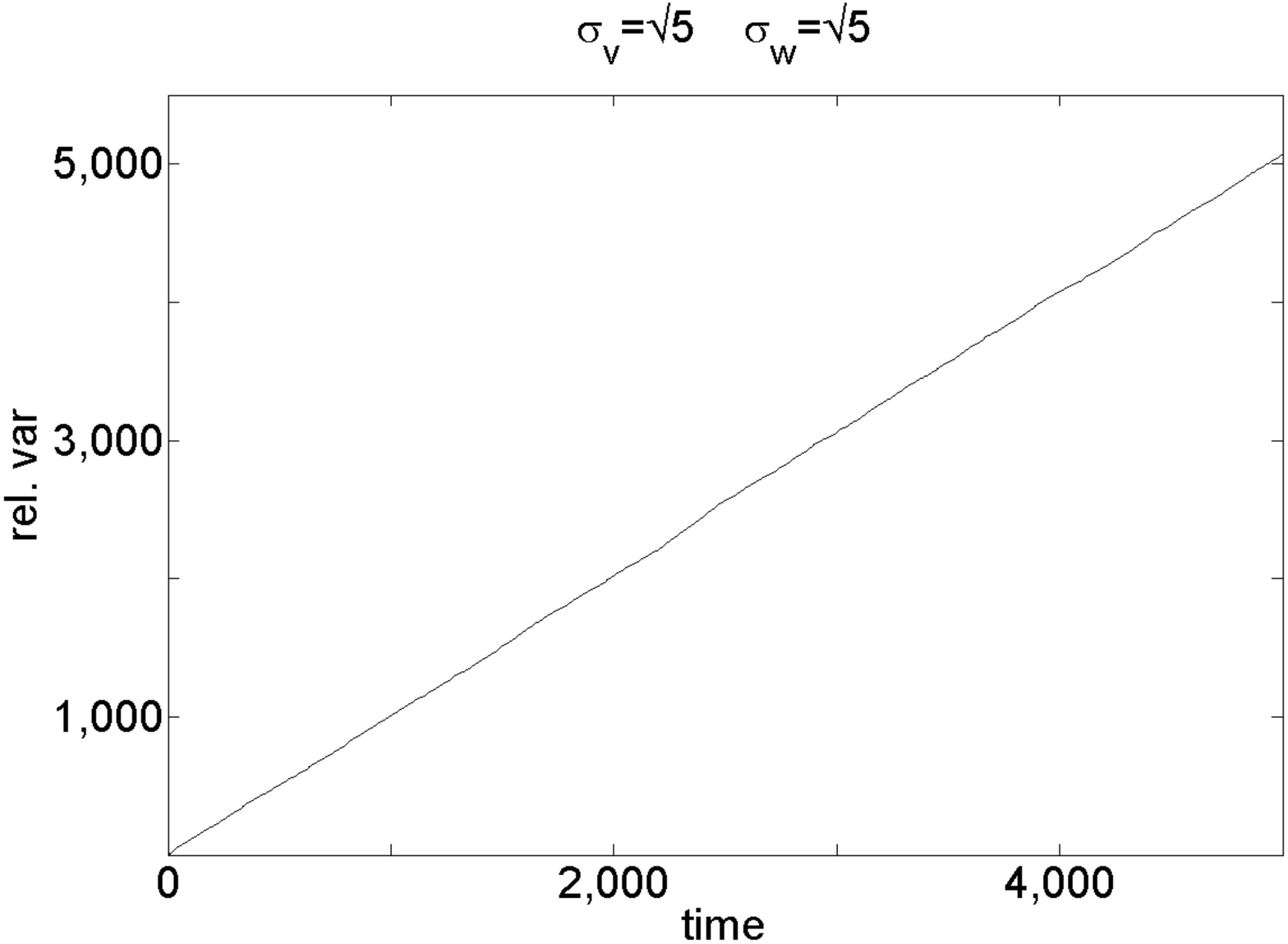}}
\caption{{\footnotesize (Log) Relative variance of normalizing constant of alive filter to Kalman filter for the linear state space model. Each panel displays the relative variance across time.}\label{fig:fig4}}
\end{center}
\end{figure}

\begin{figure}[h]
\begin{center}
\scalebox{0.2}{\includegraphics{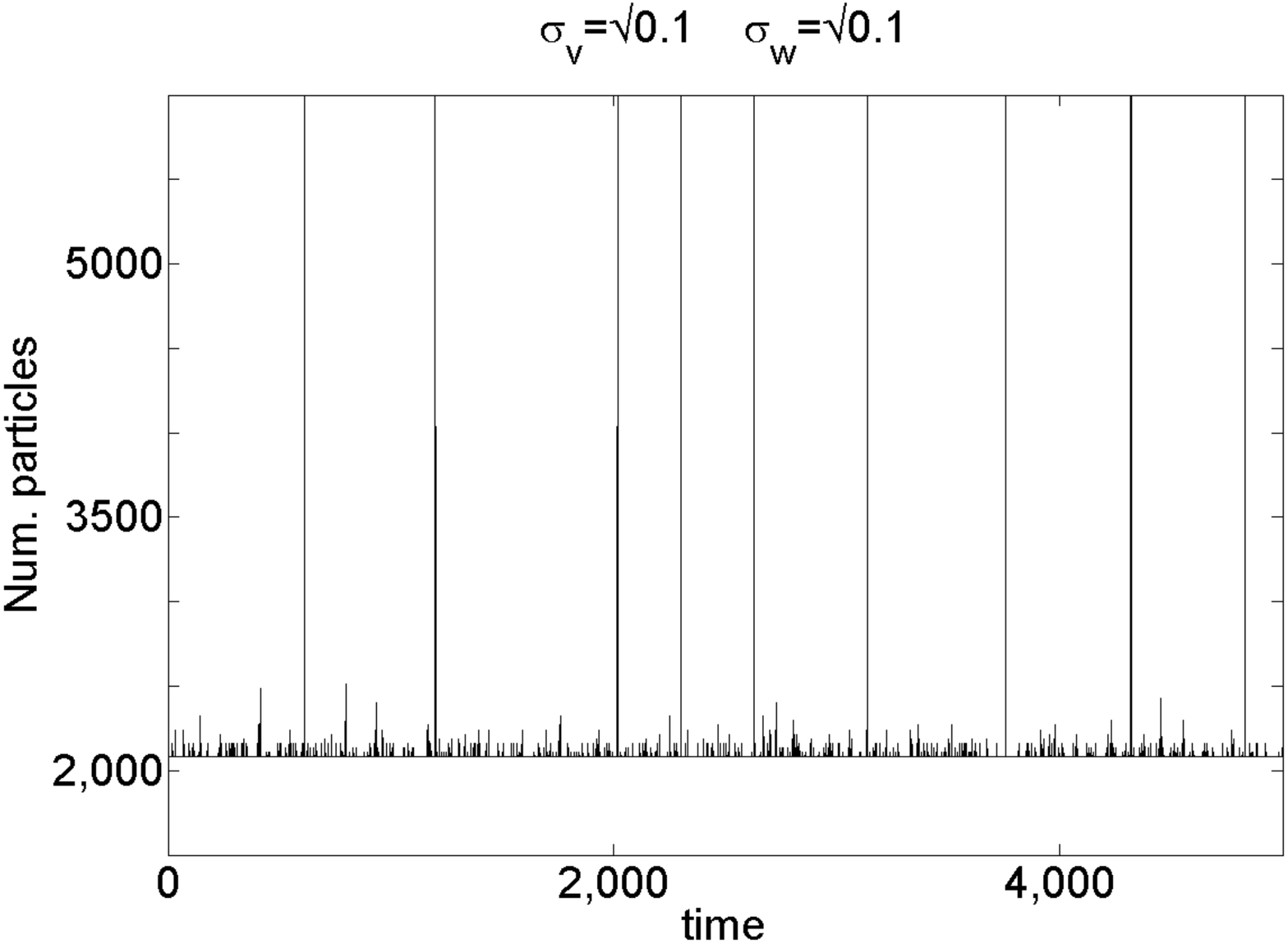}}\scalebox{0.2}{\includegraphics{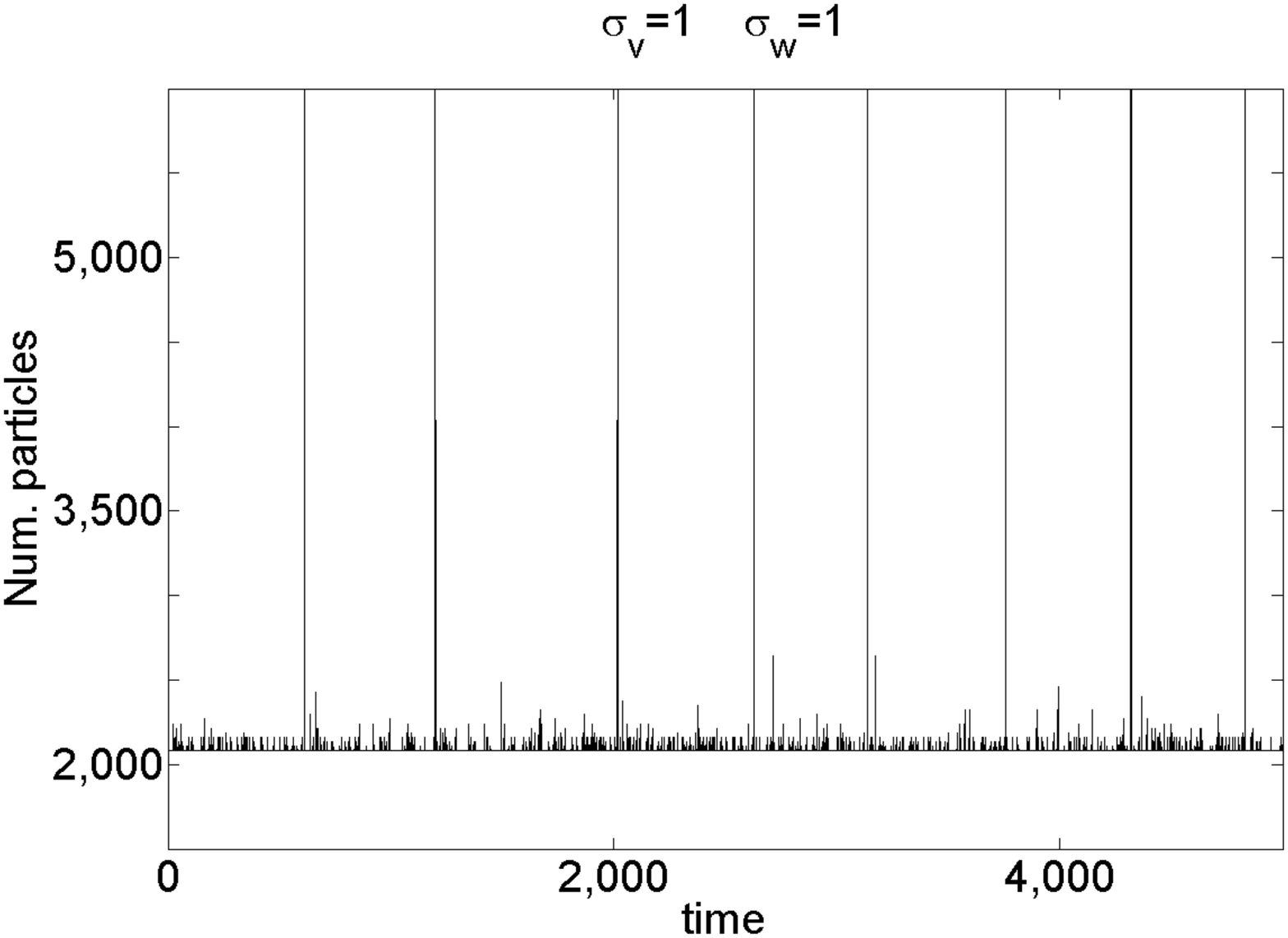}}\scalebox{0.2}{\includegraphics{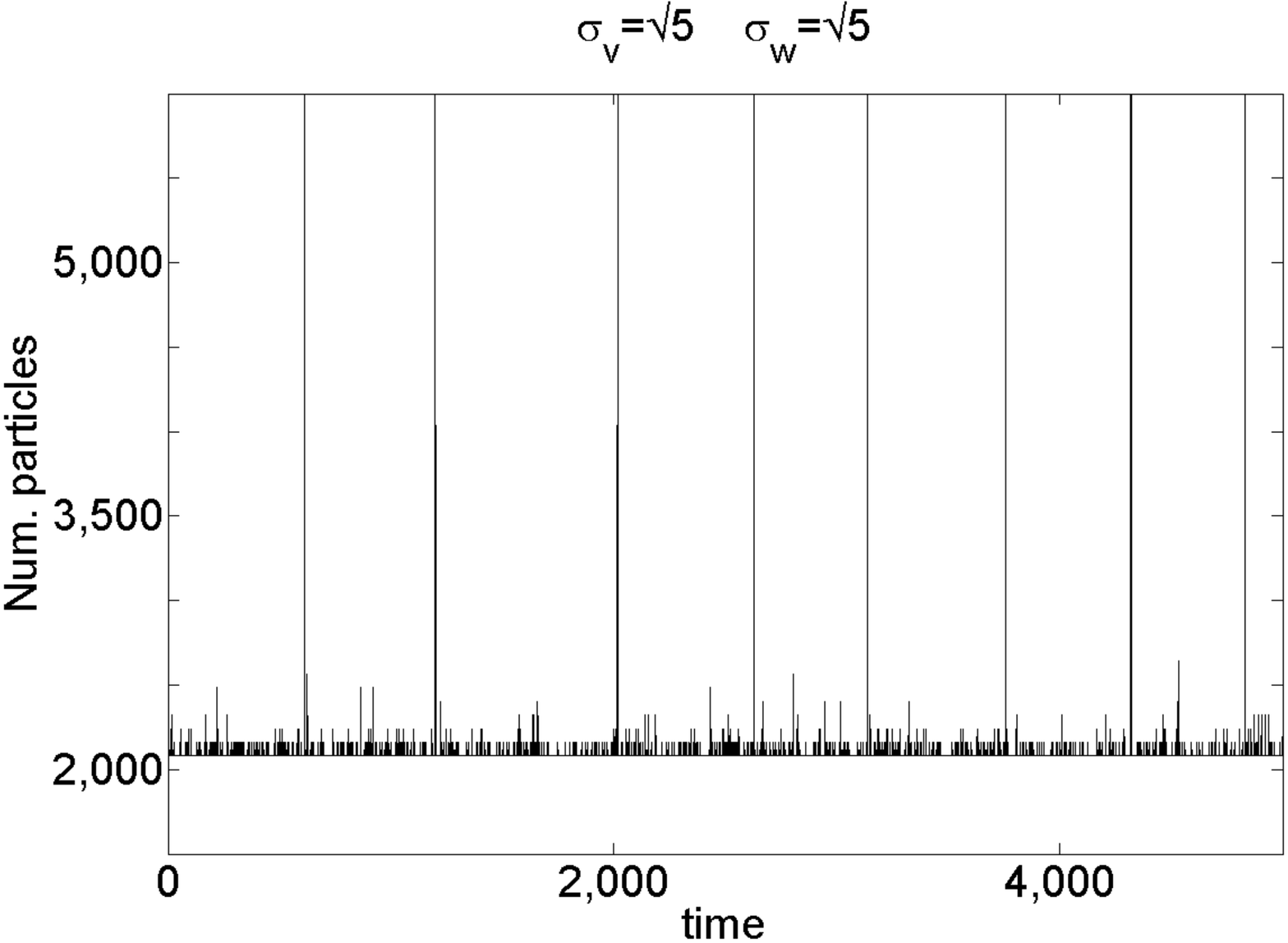}}
\caption{{\footnotesize Number of particles used for alive filter for the linear state space model. Each panel displays the number of particles across time.}\label{fig:fig5}}
\end{center}
\end{figure}

\begin{figure}[h]
\begin{center}
\scalebox{0.2}{\includegraphics{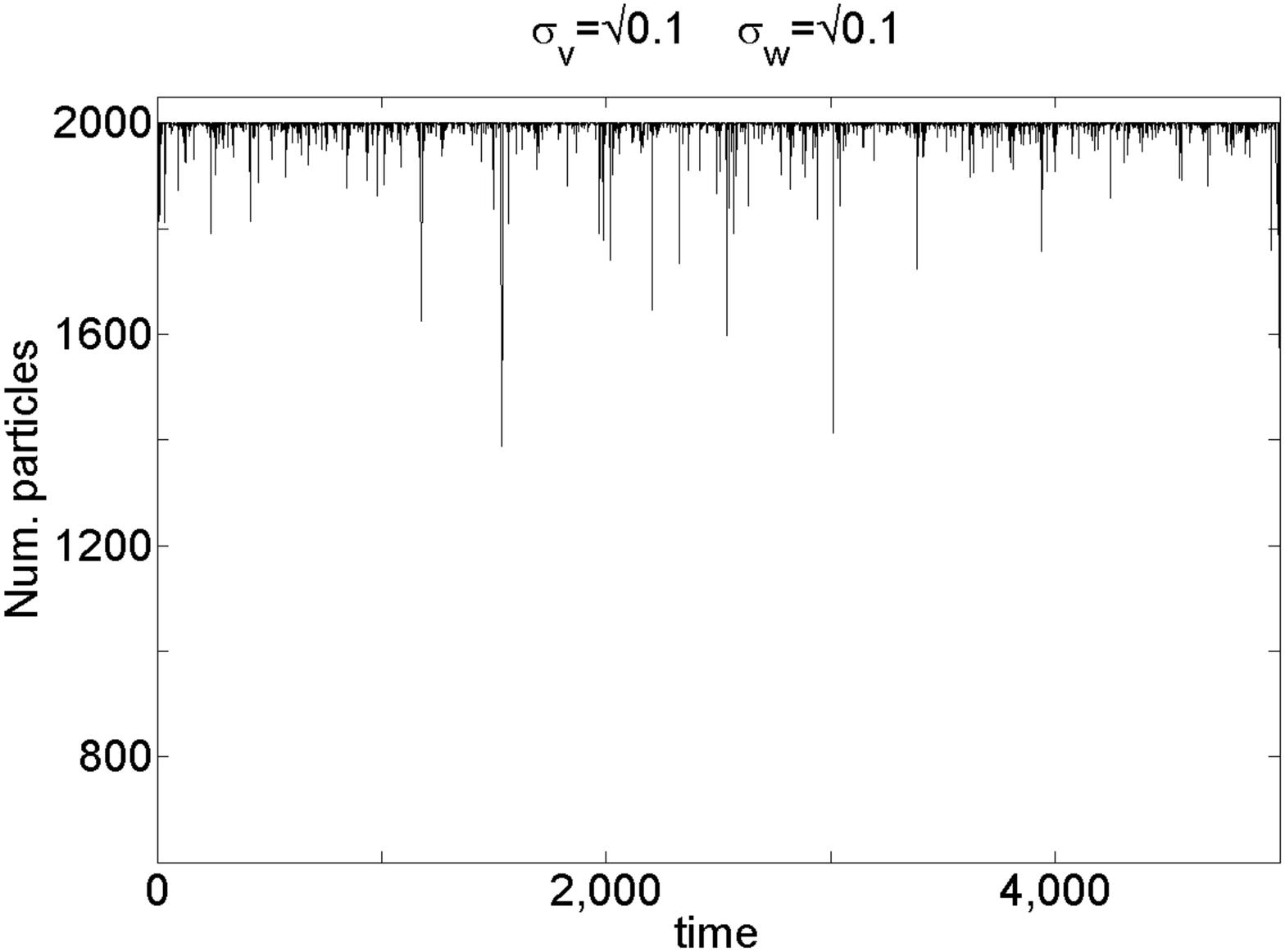}}\scalebox{0.2}{\includegraphics{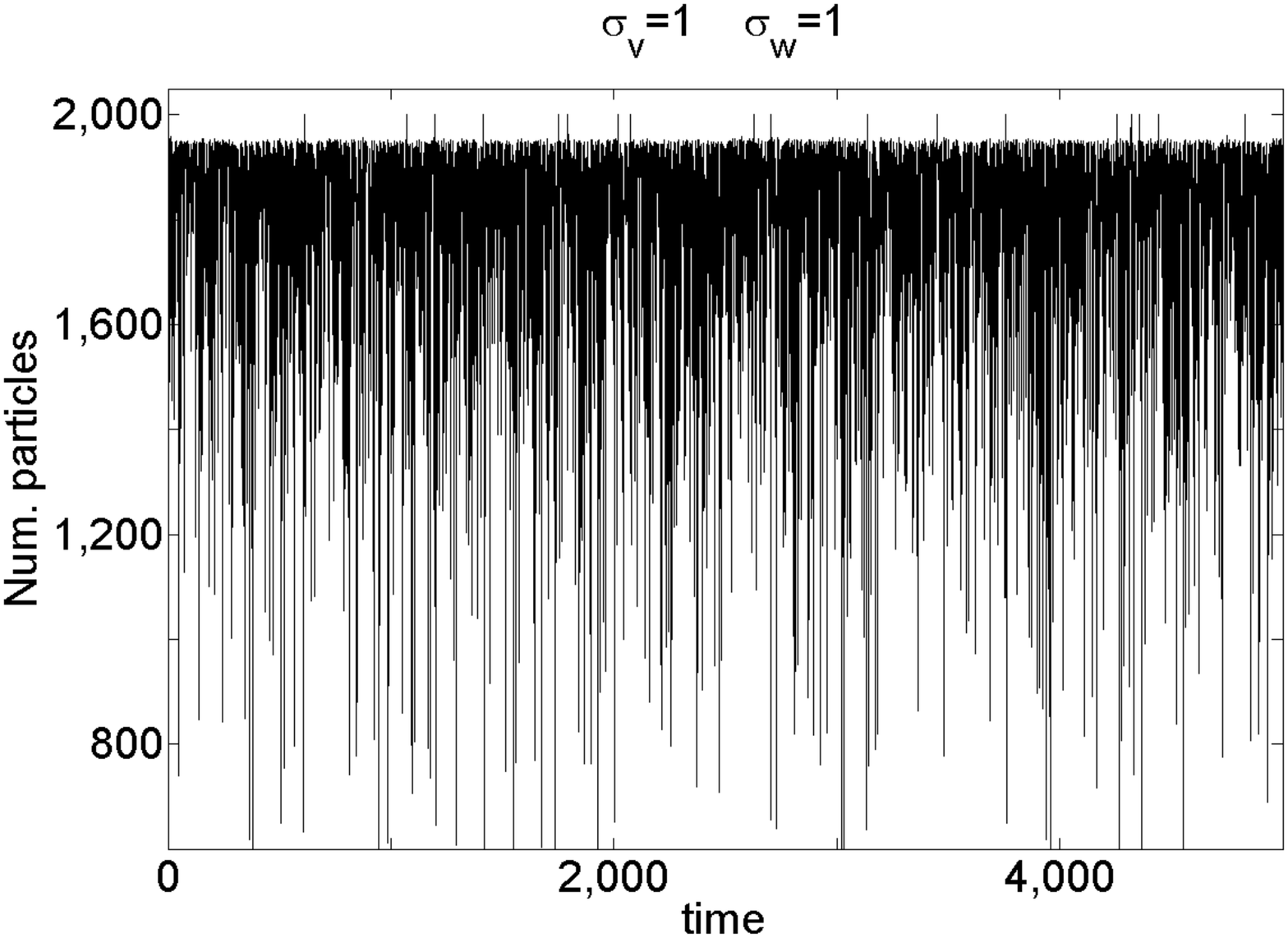}}\scalebox{0.2}{\includegraphics{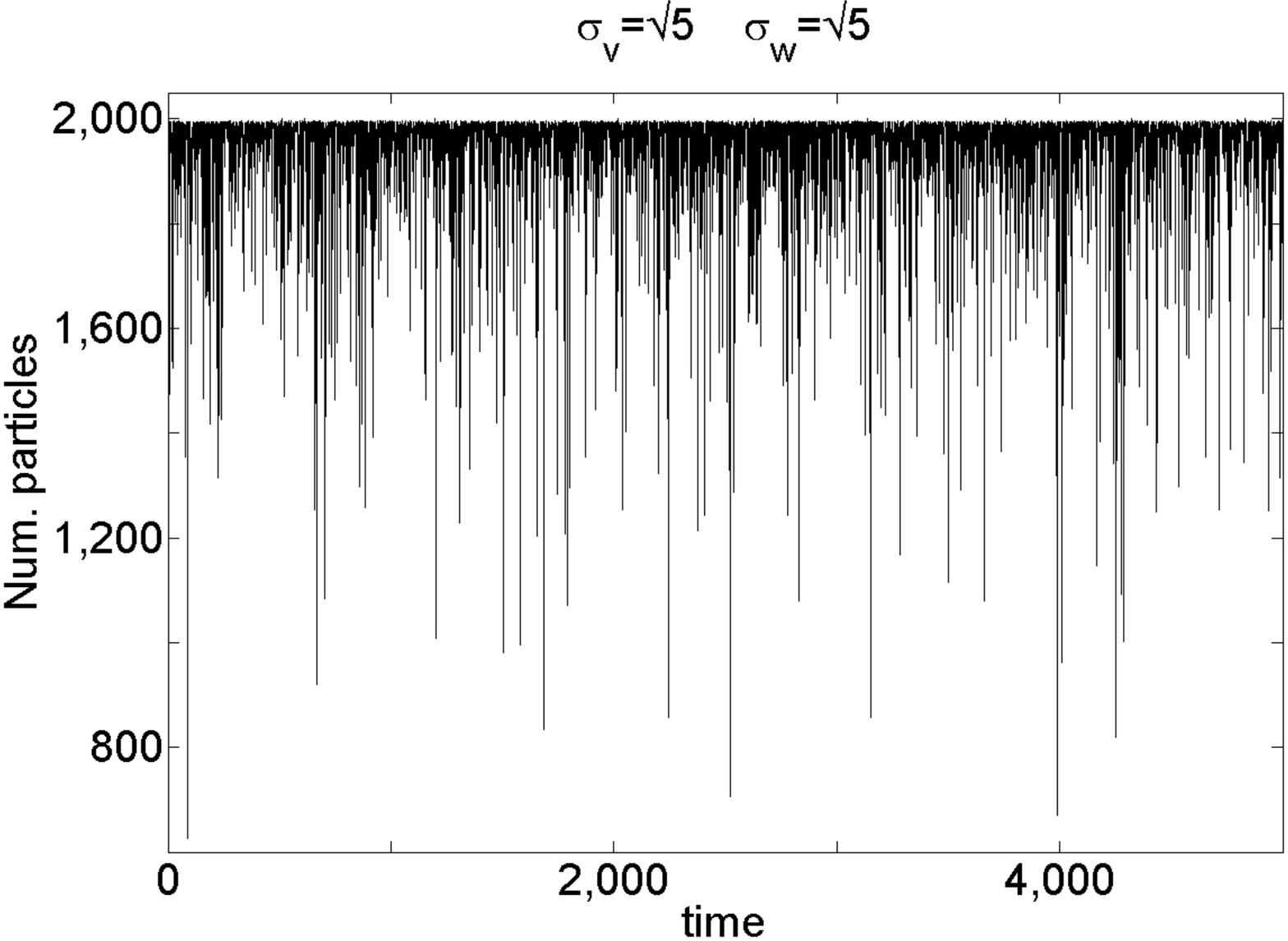}}
\caption{{\footnotesize Number of particles used for the particle filter of Jasra et al.~(2012) for the linear state space model. Each panel displays the number of particles across time.}\label{fig:fig6}}
\end{center}
\end{figure}

\subsubsection{Part \uppercase\expandafter{\romannumeral2}}

In this part, we keep the initial conditions the same as in the previous Section but change the value of $\epsilon$. Instead of using $\epsilon\in\{5, 10, 15\}$, we set smaller values to $\epsilon$, i.e.~$\epsilon\in\{3, 6, 12\}$ (recall the smaller $\epsilon$, the closer the ABC approximation is to the true HMM (Theorem 1 of Jasra et al.~(2012)). This change makes the standard particle filter collapse whereas the alive filter does not have this problem. 
All results were averaged over $50$ runs and our results are shown in Figures \ref{fig:fig7}-\ref{fig:fig8}.

In Figure \ref{fig:fig7}, we present the true simulated hidden trajectory along
with a  plot of the estimated $X_t$ given by the two particle filters across time when $(\sigma_v, \sigma_w)=(\sqrt{5}, \sqrt{5})$. As shown in Figure \ref{fig:fig7}, the alive filter can provide better estimation versus the old particle filter.
Figure \ref{fig:fig8} displays the log relative error of the alive filter to old particle filter, which supports the previous point made, with regards to estimation of the hidden state.
Based upon the results displayed, the alive filter can provide good estimation results under the same conditions when the old particle filter collapses.


\begin{figure}[h]
\begin{center}
\scalebox{0.25}{\includegraphics{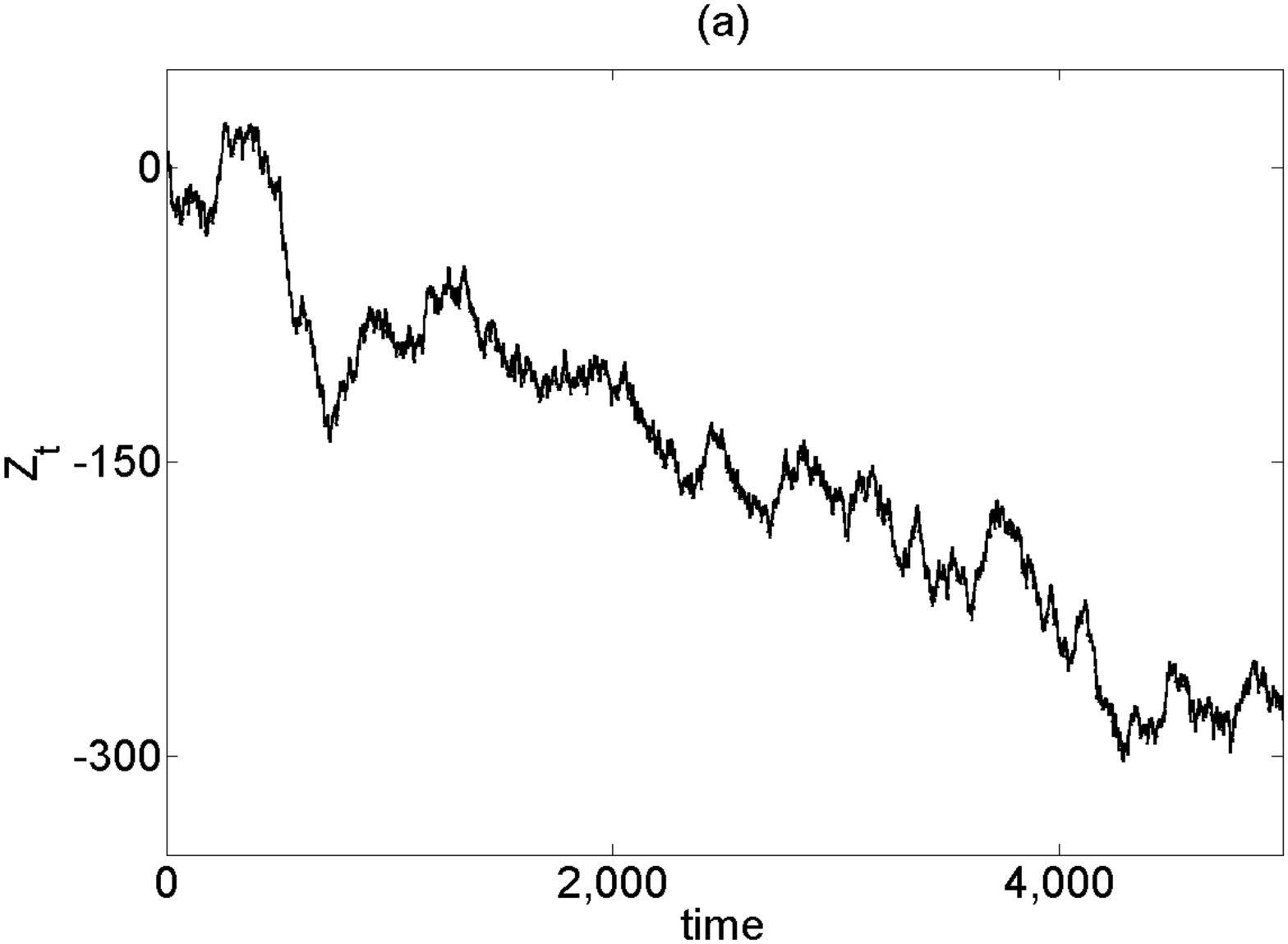}}\scalebox{0.25}{\includegraphics{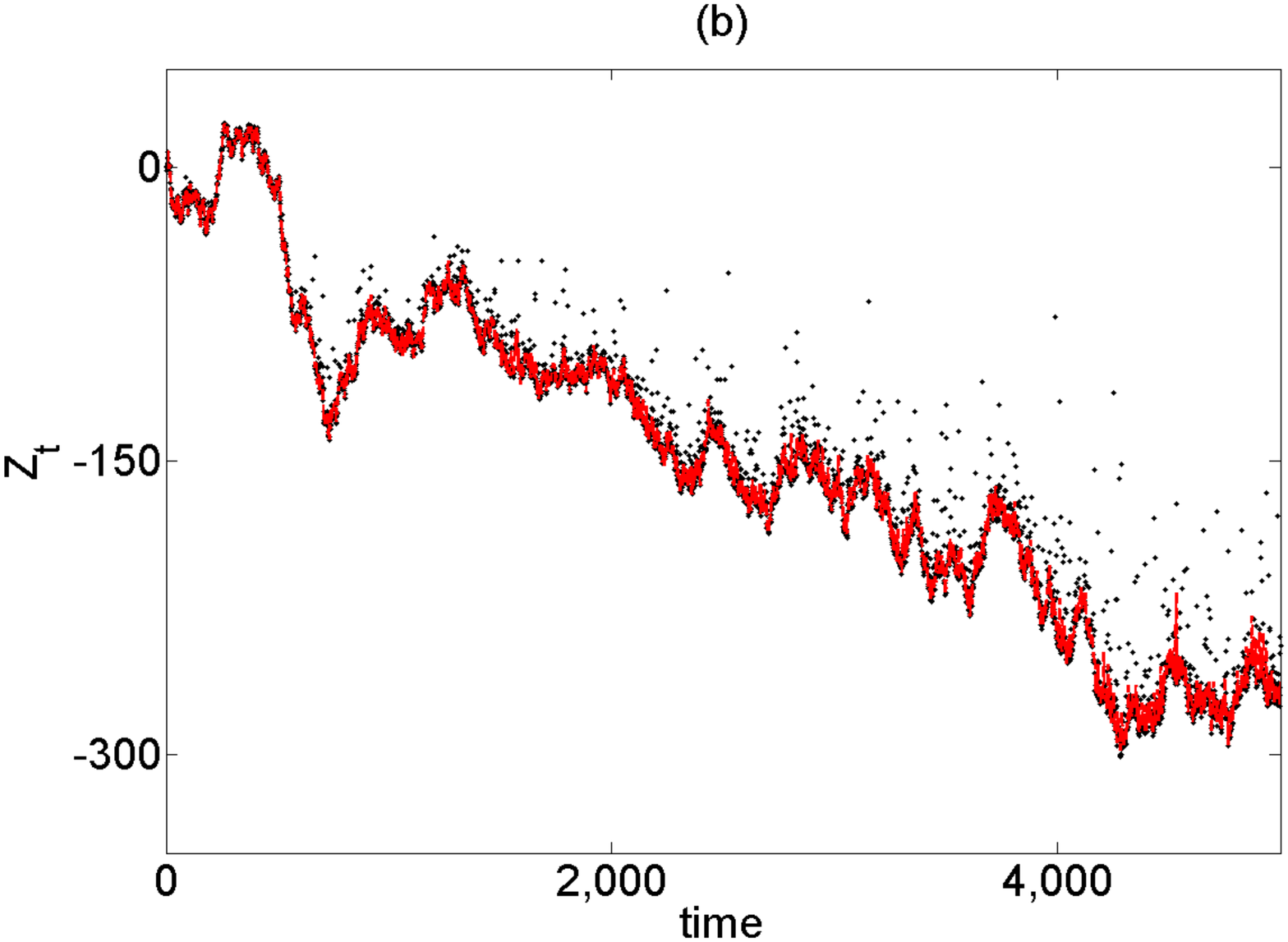}}
\caption{{\footnotesize (a) `True' $Z_t$ and (b) estimated $Z_t$ across time for the linear state space model, where red ('--') indicates the alive particle filter and black '$\cdot \cdot$' indicates the particle filter.}\label{fig:fig7}}
\end{center}
\end{figure}

\begin{figure}[h]
\begin{center}
\scalebox{0.2}{\includegraphics{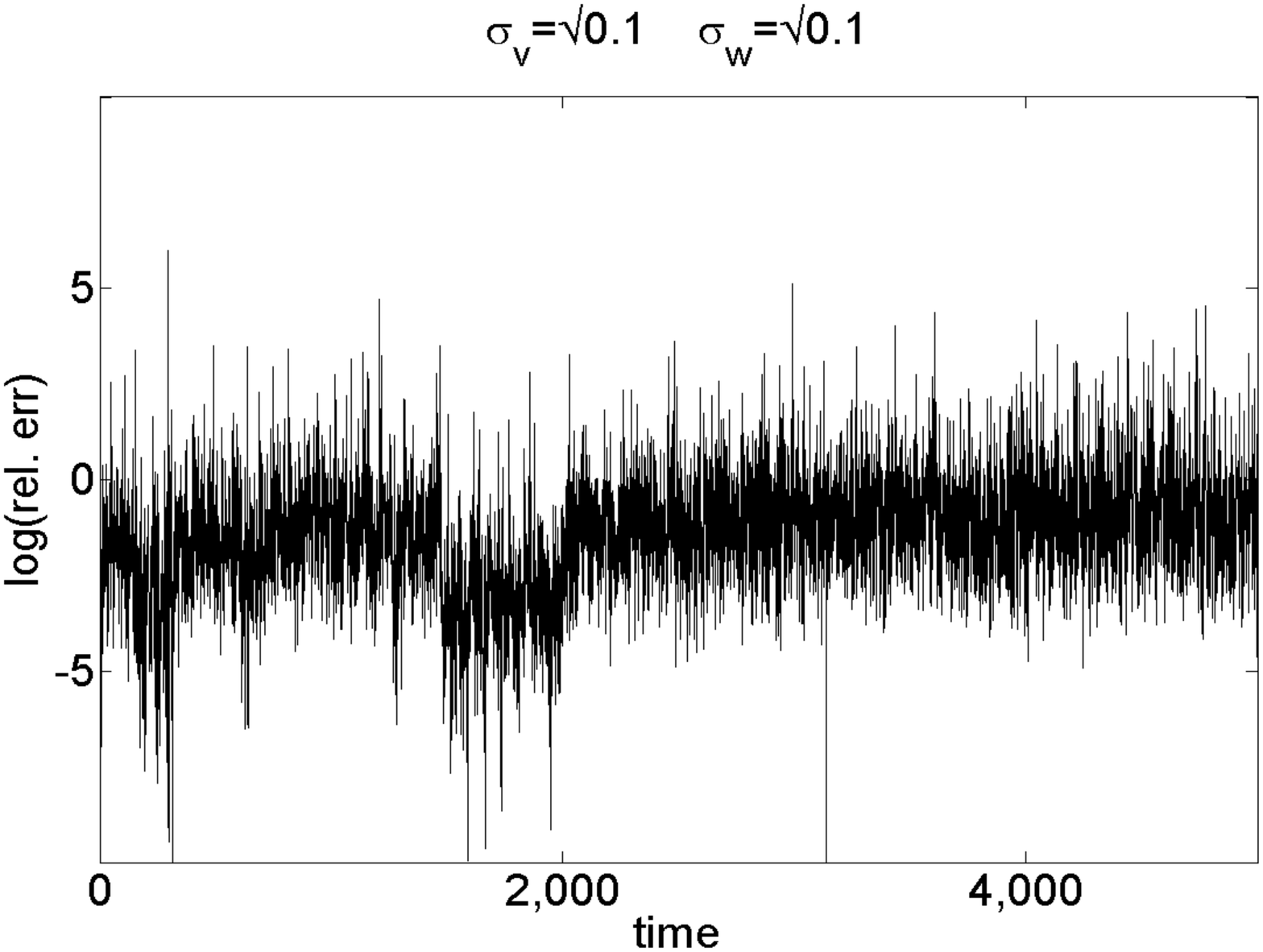}}\scalebox{0.2}{\includegraphics{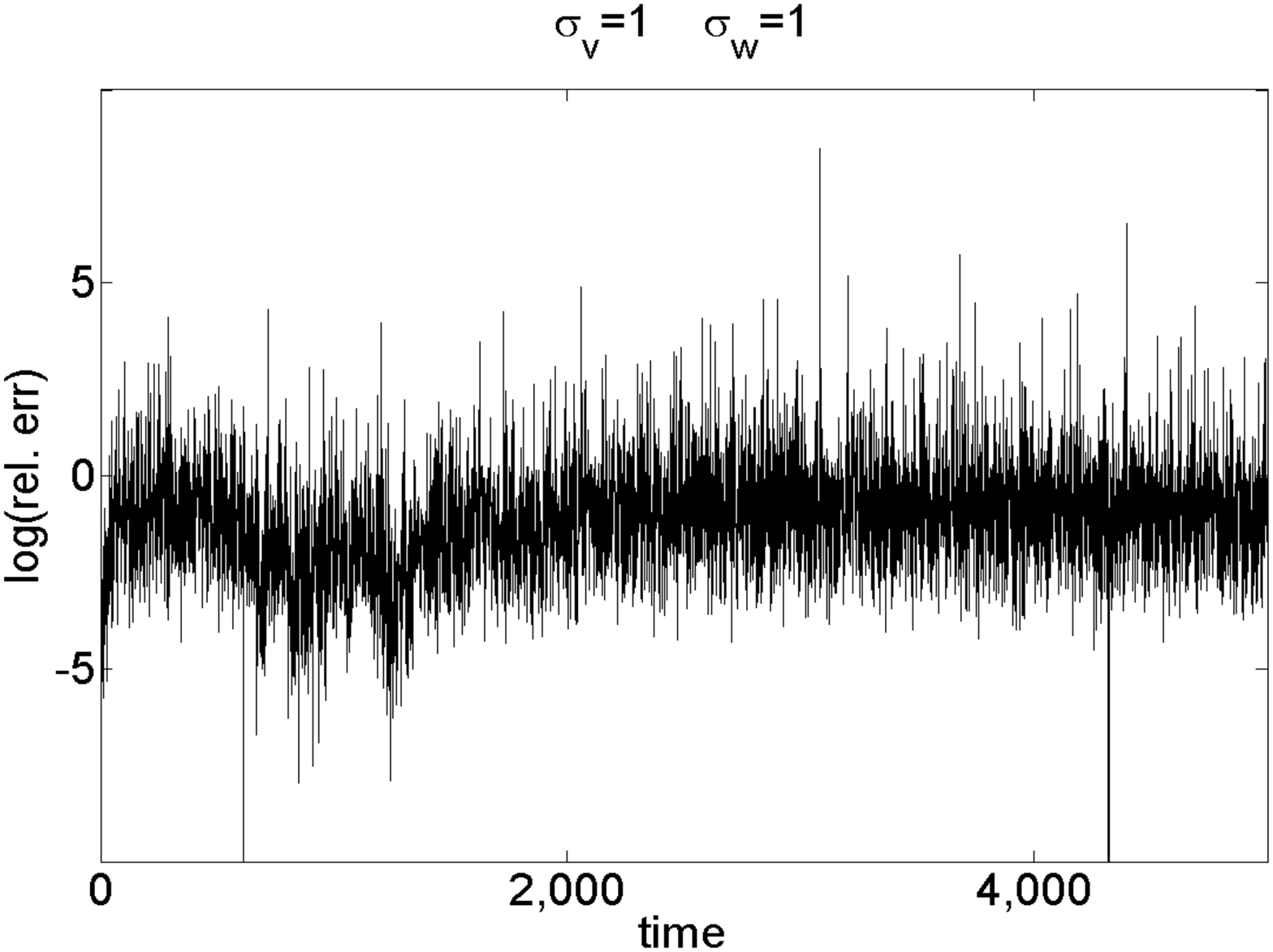}}\scalebox{0.2}{\includegraphics{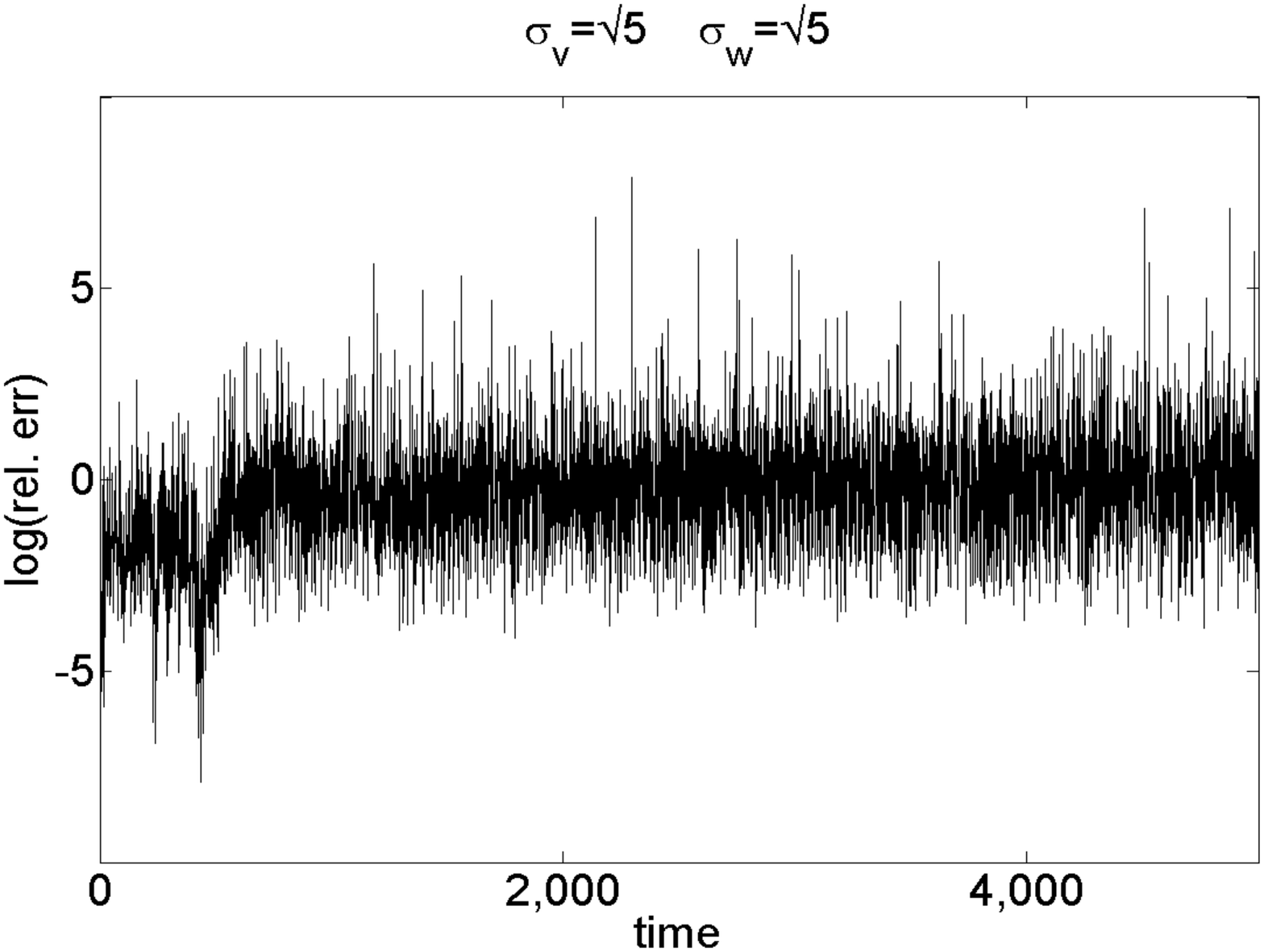}}
\caption{{\footnotesize Estimation error of the first moment for the linear state space model. Each panel displays (Log) the ratio of $\mathbb{L}_1$ error of alive filter to the particle filter.}\label{fig:fig8}}
\end{center}
\end{figure}


\subsection{Particle MCMC}\label{sec:PMCMC}

We now utilize the results in Propositions \ref{prop:unbiased}-\ref{prop:non_asymp}.
In particular, Proposition \ref{prop:unbiased} allows us to construct an MCMC method
for performing static parameter inference in the context of ABC approximations of HMMs.

Recall Section \ref{sec:motivating_ex}. Our objective is to sample from the posterior density:
\begin{equation}
\pi(\theta|y_{1:n}) = \frac{\int_{\mathsf{X}^n}\prod_{k=1}^n g_{\theta}^{\epsilon}(y_k|z_k)f_{\theta}(z_{k}|z_{k-1})dz_{1:n}\pi(\theta)}{\int_{\mathsf{X}^n\times\Theta}\prod_{k=1}^n g_{\theta}^{\epsilon}(y_k|z_k)f_{\theta}(z_{k}|z_{k-1})dz_{1:n}\pi(\theta)d\theta}
\label{eq:abc_static}
\end{equation}
where $g_{\theta}^\epsilon$, $f_{\theta}$ is as \eqref{eq:abc_smoothing} and $\pi(\theta)$
is a prior probability density on $\Theta$.
Throughout the Section, we set $N\geq 2$, $\epsilon >0$, but in general omit dependencies on these quantities.
 In practice, one often seeks to sample from an associated
probability on the extended state-space $\mathsf{E}^n\times\Theta$
$$
\tilde{\pi}(\theta,z_{1:n},u_{1:n}|y_{1:n}) \propto \prod_{k=1}^n \mathbb{I}_{B_{\epsilon}(y_k)}(u_k) g_{\theta}(u_k|z_k)f_{\theta}(z_{k}|z_{k-1})\pi(\theta).
$$
It is then easily verified that for any fixed $\theta\in\Theta$
$$
\pi(\theta|y_{1:n}) = \int_{\mathsf{E}^n} \tilde{\pi}(\theta,z_{1:n},u_{1:n}|y_{1:n}) dz_{1:n}du_{1:n}.
$$

A typical way to sample from $\tilde{\pi}(\theta,z_{1:n},u_{1:n}|y_{1:n})$ is via the Metropolis-Hastings method, with proposing to move from $(\theta,z_{1:n},u_{1:n})$ to $(\theta',z_{1:n}',u_{1:n}')$ via the probability density:
$$
q(\theta'|\theta) \prod_{k=1}^n g_{\theta'}(u_k'|z_k')f_{\theta'}(z_{k}'|z_{k-1}')
$$
such a proposal removes the need to evaluate $g_{\theta}$ which is not available in this context.
As is well known e.g.~Andrieu et al.~(2010), such procedures typically do not work very well and lead to
slow mixing on the parameter space $\Theta$. This proposal can be greatly improved by
running a particle-filter (the particle marginal Metropolis-Hastings (PMMH) algorithm) as in Andrieu et al.~(2010); that is a Metropolis-Hastings move that will
first move $\theta$, via $q(\theta'|\theta)$ and then run the algorithm in Section \ref{sec:old_algo} picking a whole path, $l$, $x_{1:n}^l\in\mathsf{E}^n$
the sample used with a probability proportional to
$G_n(x_n^i)$. Remarkably, this procedure
yields samples from \eqref{eq:abc_static} via an auxiliary probability density; the details can be found in Andrieu et al.~(2010), but the apparently \emph{fundamental} property is that the estimate of the normalizing constant is unbiased. Note also that the sample from the Markov chain $(\theta,x_{1:n}^l)$ also provides a sample from $\tilde{\pi}(\theta,z_{1:n},u_{1:n}|y_{1:n})$.

As we have seen in the context of both theory and applications, it appears that the alive
filter in Section \ref{sec:new_smc} out-performs the standard one, for a given computational complexity. In addition,
as seen in Proposition \ref{prop:unbiased}, the estimate of the normalizing constant is unbiased.
It is therefore a reasonable conjecture that one can construct a new PMMH algorithm, with the alive particle filter investigated previously in this article and that this might perform better (in some sense) than the standard PMMH just described. We remark that the justification of this new PMMH follows from the statements in Andrieu \& Vihola (2012) (see also Andrieu \& Roberts (2009)) and Proposition \ref{prop:unbiased}, but we provide details for completeness.


\subsubsection{New PMMH Kernel}

We will define an appropriate target probability to produce samples from
\eqref{eq:abc_static}, but we first give the algorithm:
\begin{enumerate}
\item{Sample $\theta(0)$ from any absolutely continuous distribution. Then run the particle filter (with parameter value $\theta(0)$) in Section \ref{sec:new_smc} up-to time $n$, storing $\gamma_{n+1}^N(1)$ (now denoted $\gamma_{n+1,\theta(0)}^N(1)$). Pick a trajectory $x_{1:n}^i(0)$, $i\in\{1,\dots,T_n(0)-1\}$, with probability
$$
\frac{G_n(x_n^i(0))}{\sum_{i=1}^{T_n(0)-1} G_n(x_n^i(0))}.
$$
Set $i=1$.}
\item{Propose $\theta'|\theta(i-1)$ from a proposal with positive density on $\Theta$ (write
it $q(\theta'|\theta)$).Then run the particle filter (with parameter value $\theta'$) in Section \ref{sec:new_smc} up-to time $n$, storing $\gamma_{n+1,\theta'}^N(1)$. Pick a trajectory $(x_{1:n}^i)'$ with probability
$$
\frac{G_n((x_n^i)')}{\sum_{i=1}^{T_n(0)-1} G_n((x_n^i)')}.
$$
Set $\theta(i) =\theta'$, $\gamma_{n+1,\theta(i)}^N(1)=\gamma_{n+1,\theta'}^N(1)$ with probability:
$$
1\wedge \frac{\gamma_{n+1,\theta'}^N(1)}{\gamma_{n+1,\theta(i-1)}^N(1)}
\frac{\pi(\theta')q(\theta(i-1)|\theta')}{\pi(\theta(i-1))q(\theta'|\theta(i-1))}.
$$
Otherwise set $\theta(i) =\theta(i-1)$, $\gamma_{n+1,\theta(i)}^N(1)=\gamma_{n+1,\theta(i-1)}^N(1)$,
$i=i+1$ and return to the start of 2.}
\end{enumerate}
For readers interested in the numerical implementation, they can skip to the next Section,
noting that the $\theta$ samples will come from the posterior \eqref{eq:abc_static}; this is now justified in the rest of the section.

We construct the following auxiliary target probability on the state-space:
\begin{eqnarray*}
\bar{\mathsf{E}} & = & 
\Theta\times \Big(\bigcup_{T_1=N}^{\infty}\Big(\mathsf{E}^{T_1}\times\{T_1\}
\times \Big(\bigcup_{T_2=N}^{\infty}\Big(\mathsf{E}^{T_2}\times\{1,\dots,T_1-1\}^{T_2} \times\{T_2\}\times\cdots\times\bigcup_{T_n=N}^{\infty}\Big(
\mathsf{E}^{T_n}\times \\ & & \{1,\dots,T_{n-1}-1\}^{T_n}\times\{T_n\}\times
\{1,\dots,T_n-1\}\Big)\cdots\Big)\Big)\Big).
\end{eqnarray*}
Whilst the state-space looks complicated it corresponds to the static parameter and
all the variables (the states and the resampled indices) sampled by the alive particle filter up-to time-step $n$ and then just the picking of one of the final paths.

For $n\geq 2$ (we omit $\theta$ from our notation) define
$$
\Psi_n\Big(d(x_n^1,\dots,x_n^{T_n}),a_{n-1}^{1},\dots,a_{n-1}^{T_n},T_n|x_{n-1}^{1:T_{n-1}},T_{n-1}\Big)  :=
$$
$$
 \frac{
\mathbb{I}_{\mathsf{S}_n}
(x_n^1,\dots,x_n^{T_n},T_n)\binom{T_n-1}{N-1}
\prod_{i=1}^{T_n} \frac{G_{n-1}(x_{n-1}^{a_{n-1}^i})}{\sum_{i=1}^{T_{n-1}-1}G_{n-1}(x_{n-1}^i)}
M_n(x_{n-1}^{a_{n-1}^i}, dx_n^i)
}
{\sum_{T_n=N}^{\infty} \sum_{a_{n-1}^{1:T_{n}}\in\{1,\dots,T_{n-1}-1\}}\binom{T_n-1}{N-1}
\int_{\mathsf{E}^{T_n}} \mathbb{I}_{\mathsf{S}_n}
(x_n^1,\dots,x_n^{T_n},T_n)
\Big[\prod_{i=1}^{T_n} \frac{G_{n-1}(x_{n-1}^{a_{n-1}^i})}{\sum_{i=1}^{T_{n-1}-1}G_{n-1}(x_{n-1}^i)} M_n(x_{n-1}^{a_{n-1}^i}, dx_n^i)\Big]
}
$$
where for $n\geq 1$
$$
\mathsf{S}_n = \{(u_n^1,\dots,u_n^{T_n},T_n)\in\mathsf{Y}^n\times\{N,N+1,\dots\}:
\sum_{i=1}^{T_n-1}\mathbb{I}_{B_{\epsilon}(y_n)}(u_n^i) = N-1\cap u_n^{T_{n}}\in B_{\epsilon}(y_n)
\}.
$$
In addition, set
$$
\Psi_1\Big(d(x_1^1,\dots,x_1^{T_1}),T_1\Big)  :=
 \frac{
\mathbb{I}_{\mathsf{S}_1}
(x_1^1,\dots,x_n^{T_1},T_1)\binom{T_1-1}{N-1}
\prod_{i=1}^{T_1} M_n(x_{0}, dx_1^i)
}
{\sum_{T_1=N}^{\infty} \binom{T_1-1}{N-1}
\int_{\mathsf{E}^{T_1}} \mathbb{I}_{\mathsf{S}_1}
(x_1^1,\dots,x_1^{T_1},T_1)
\Big[\prod_{i=1}^{T_1} M_1(x_{0}, dx_n^i)\Big].
}
$$

Then the PMMH algorithm just defined samples from the target
\begin{eqnarray*}
\bar{\pi}(\theta,d(\mathbf{x}_1,\dots,\mathbf{x}_n),\mathbf{a}_{1:n-1},l,T_{1:n}|y_{1:n})
& \propto & G_n(x_n^l) \gamma_{n+1,\theta}^N(1) \prod_{k=2}^n \Psi_k\Big(d(x_k^1,\dots,x_k^{T_k}),a_{k-1}^{1},\dots,a_{k-1}^{T_k},T_k|x_{k-1}^{1:T_{k-1}},T_{k-1}\Big) \times \\
& & 
\Psi_1\Big(d(x_1^1,\dots,x_1^{T_1}),T_1\Big) \pi(\theta).
\end{eqnarray*}
where $\mathbf{a}_k = (a_k^1,\dots,a_k^{T_k})$, $\mathbf{x}_k = (x_k^1,\dots,x_k^{T_k})$ and
$l\in\{1,\dots,T_n-1\}$. Using Proposition \ref{prop:unbiased}, one can easily verify that
for any fixed $\theta\in\Theta$
$$
\pi(\theta|y_{1:n}) = \int_{\bar{\mathsf{E}}\setminus\Theta}\bar{\pi}(\theta,d(\mathbf{x}_1,\dots,\mathbf{x}_n),\mathbf{a}_{1:n-1},k,t_{1:n}|y_{1:n}).
$$
Note also that the samples $(\theta,x_{1:n}^l)$ from $\bar{\pi}$ are marginally distributed according to $\tilde{\pi}(\theta,z_{1:n},u_{1:n}|y_{1:n})$.
The associated ergodicity of the new PMMH algorithm follows the construction in Andrieu et al.~(2010)
and we omit details for brevity.

\subsubsection{Implementation on Real Data}

We consider the following state-space model, for $n\geq 1$
\begin{align*}
Y_n =& \varepsilon_n \beta \exp (Z_n) \\
Z_n =& \phi Z_{n-1} + \sigma V_n
\end{align*}
where $\varepsilon_n\sim\mathcal{S}t(0,\xi_1,\xi_2,\xi_3)$ (a stable distribution with location
parameter 0, scale $\xi_1$, skewness parameter $\xi_2$ and stability parameter $\xi_3$) and $V_n \sim \mathcal{N}(0,c)$. 
We set $\theta=(\beta,c,\phi)$, with priors 
$c\sim \mathcal{IG}(2,1/100)$, $\phi\sim \mathcal{IG}(2,1/50)$
($\mathcal{IG}(a,b)$ is an inverse Gamma distribution with mode $b/(a+1)$)
and $\beta\sim \mathcal{N}(0,10)$. Note that the inverse Gamma distributions have infinite variance.

We consider the daily (adjust closing) index of the S \& P 500 index between 03/01/2011 $-$ 14/02/2013 (533 data points). Our data are the log-returns, that is, if $I_n$ is the index value at time $n$, $Y_n=\log(I_n/I_{n-1})$. The data are displayed in Figure \ref{fig:ex3sp500}.
The stable distribution may help us to more realistically capture heavy tails prevalent in financial data, than perhaps a standard Gaussian. In most scenarios, the probability density function of a stable distribution is intractable, which suggests that an ABC approximation might be a sensitive way to approximate the true model.

\begin{figure}[H]
\begin{center}
\scalebox{0.3}{\includegraphics{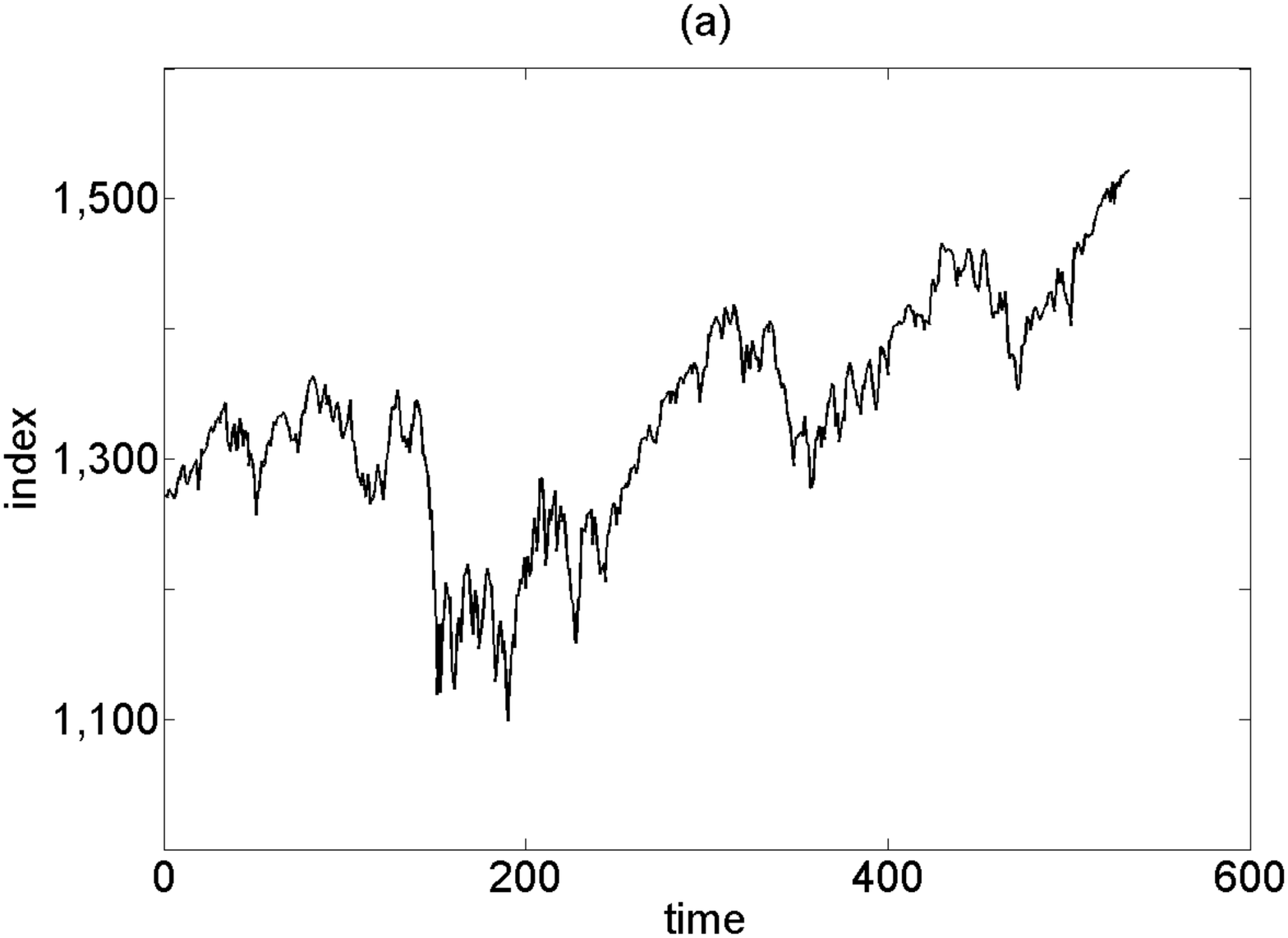}}\scalebox{0.3}{\includegraphics{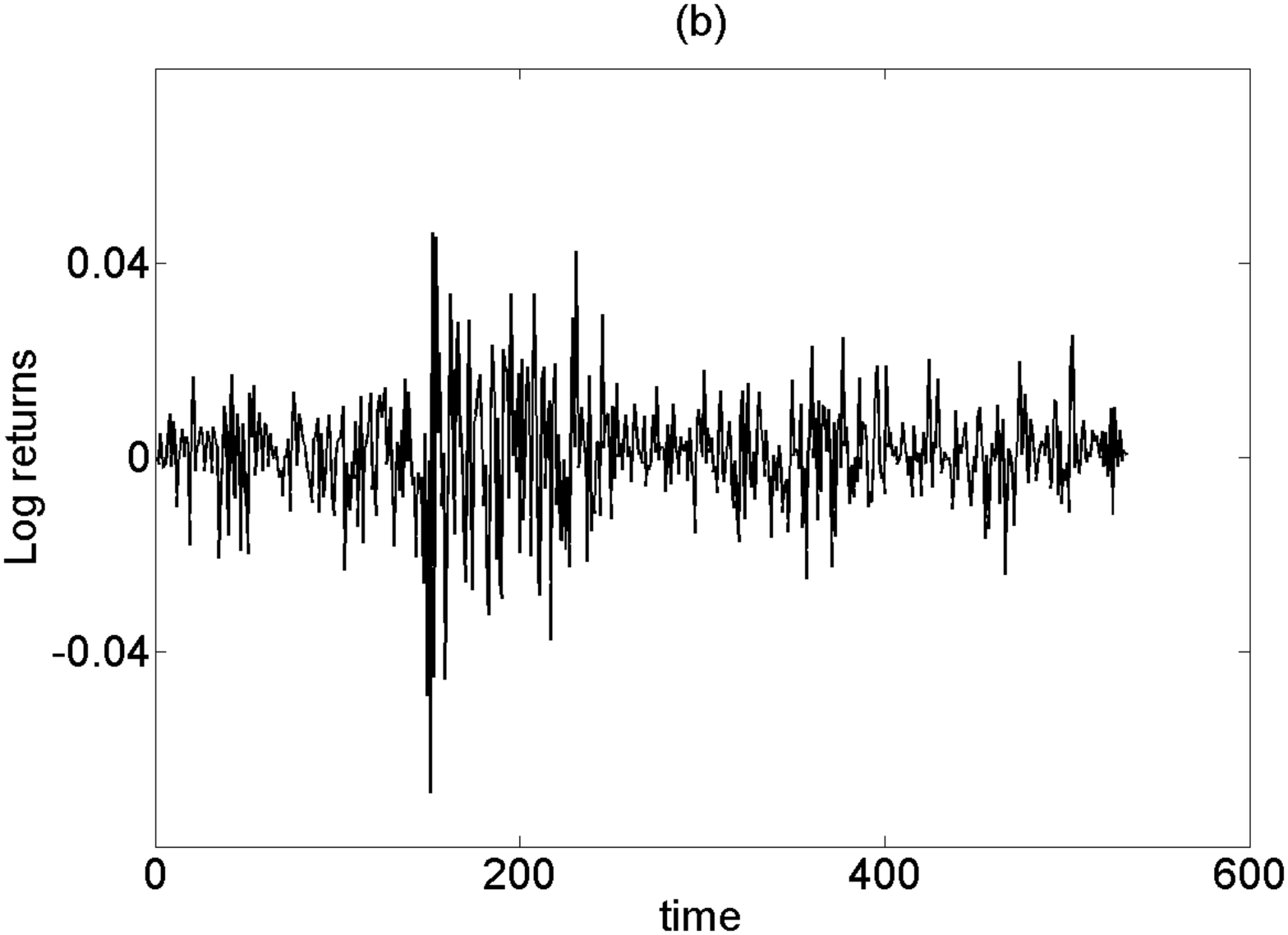}}
\caption{{\footnotesize S $\&$ P 500 (a) index data and (b) (Log) Daily return }}\label{fig:ex3sp500}
\end{center}
\end{figure}

\subsubsection{Algorithm setup}\label{sec:paset}

We consider two scenarios to compare the standard PMMH algorithm and the new one developed
above. In the first situation we set $\xi_3=1.75$ and in the second, $\xi_3=1.2$, with $\xi_1=\xi_2=1$ in both situations. In the first case, we make $\epsilon$ a suitable value as the data are not expected to jump off the same scale as the initial data. 
In the second, $\epsilon$ is significantly reduced; this is to illustrate a point about the algorithm we introduce.
Both algorithms are run for about the same computational time, such that the new PMMH algorithm has 20000 iterations. The parameters are initialized with draws
from the priors. The proposal on $\beta$ is a normal random walk
and for $(c,\phi)$ a gamma proposal centered at the current point with proposal variance scaled to obtain reasonable acceptance rates. We consider $N\in\{10,100,1000\}$ and for the new PMMH algorithm this value is lower to allow the same computational time.

\subsubsection{Results}

Our results are presented in Figures \ref{fig:traceplot_old}-\ref{fig:traceplot_col_new}.
In Figures \ref{fig:traceplot_old}-\ref{fig:traceplot_new} we can see the output in the case that $\xi_3=1.75$. 
For all cases, it appears that both algorithms perform very well; the acceptance rates were around 0.25 for each case.
For the PMMH algorithm the average number of simulations of the data, per-iteration and data-point, were $(1636, 745, 365)$ for
$N\in\{1000,100,10\}$ respectively (recall we have modified $N$ to make the computational time similar to the standard PMMH).
For this scenario one would prefer the standard PMMH as the algorithmic performance is very good, with a removal of a random computation
cost per iteration.

In Figures \ref{fig:traceplot_col_old}-\ref{fig:traceplot_col_new} the output when $\xi=1.2$ is displayed. In Figure
\ref{fig:traceplot_col_old} we can see that the standard PMMH algorithm performs very badly, barely moving across the parameter space,
whereas the new PMMH algorithm has very reasonable performance (Figure \ref{fig:traceplot_col_new}). In this case, $\epsilon$ is very small,
and the standard SMC collapses very often, which leads to the undesirable performance displayed. We note that considerable effort was expended
in trying to get the standard PMMH algorithm to work in this case, but we did not manage to do so (so we do not claim that the algorithm cannot be made to work).
Note also that whilst these are just one run of the algorithms, we have seen this behaviour in many other cases and it is typical in these examples.
The results here suggest that the new PMMH kernel might be preferred in difficult sampling scenarios, but in simple cases it does not seem to be required.

\begin{figure}[H]
\begin{center}
\scalebox{0.2}{\includegraphics{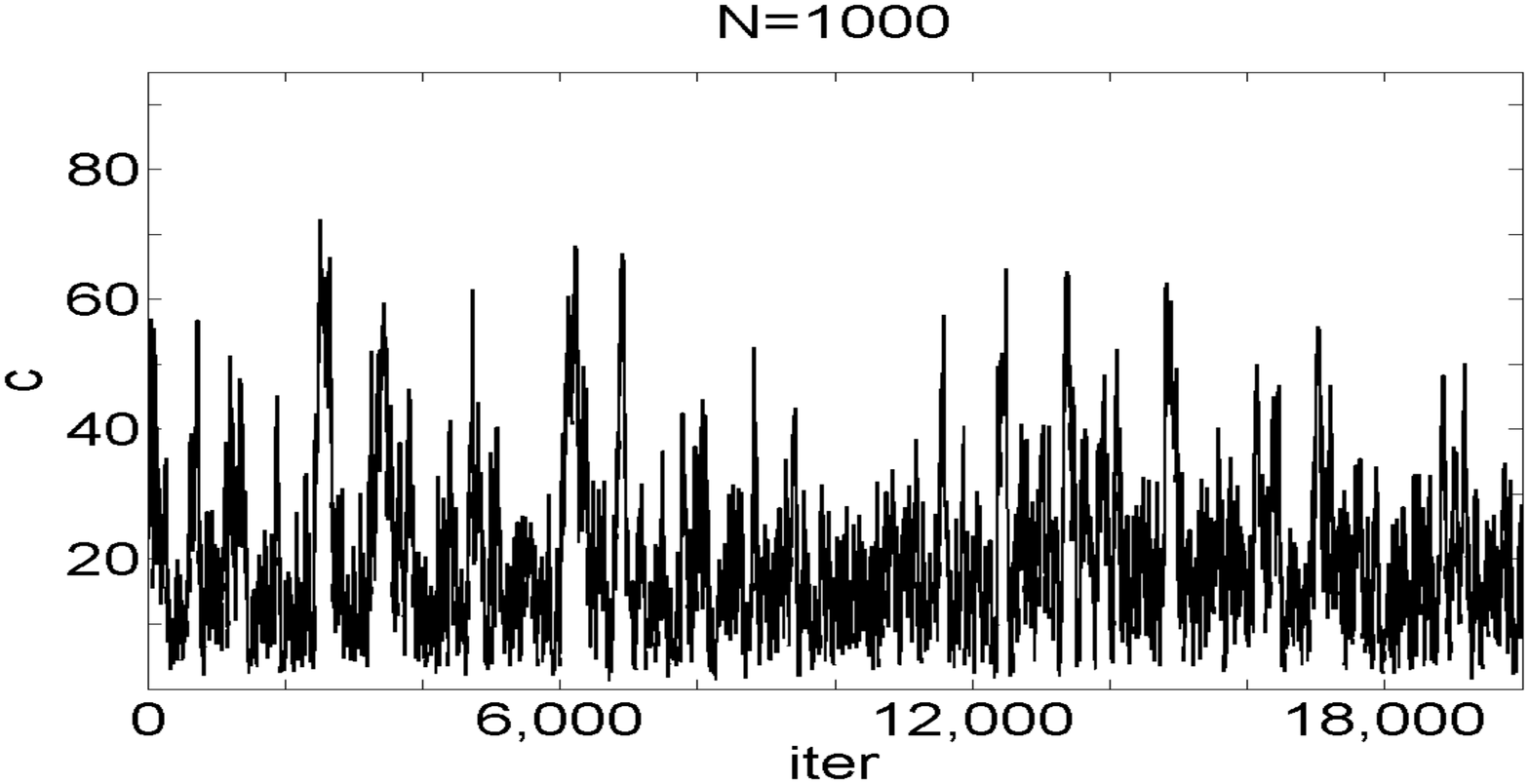}}\scalebox{0.2}{\includegraphics{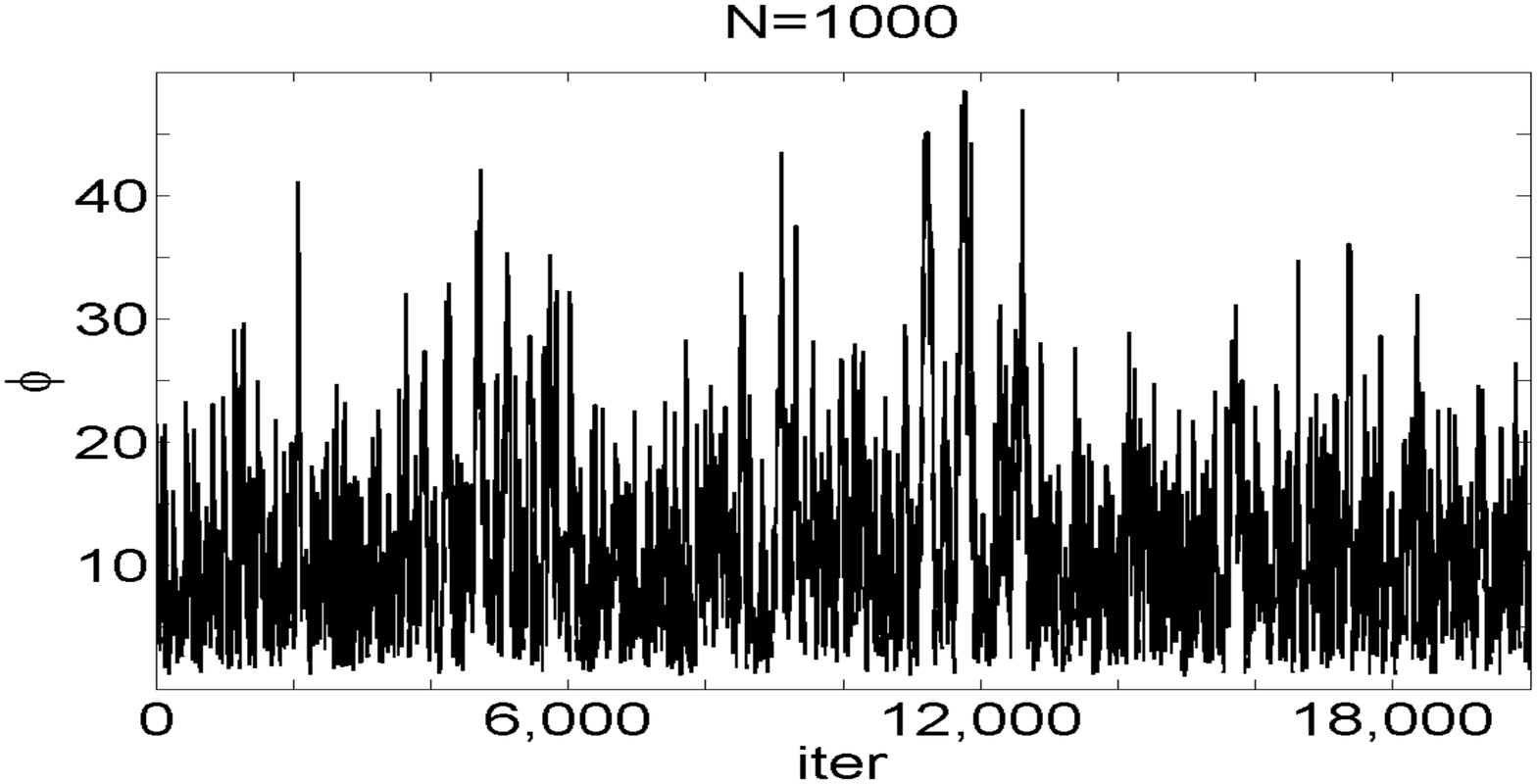}}\scalebox{0.2}{\includegraphics{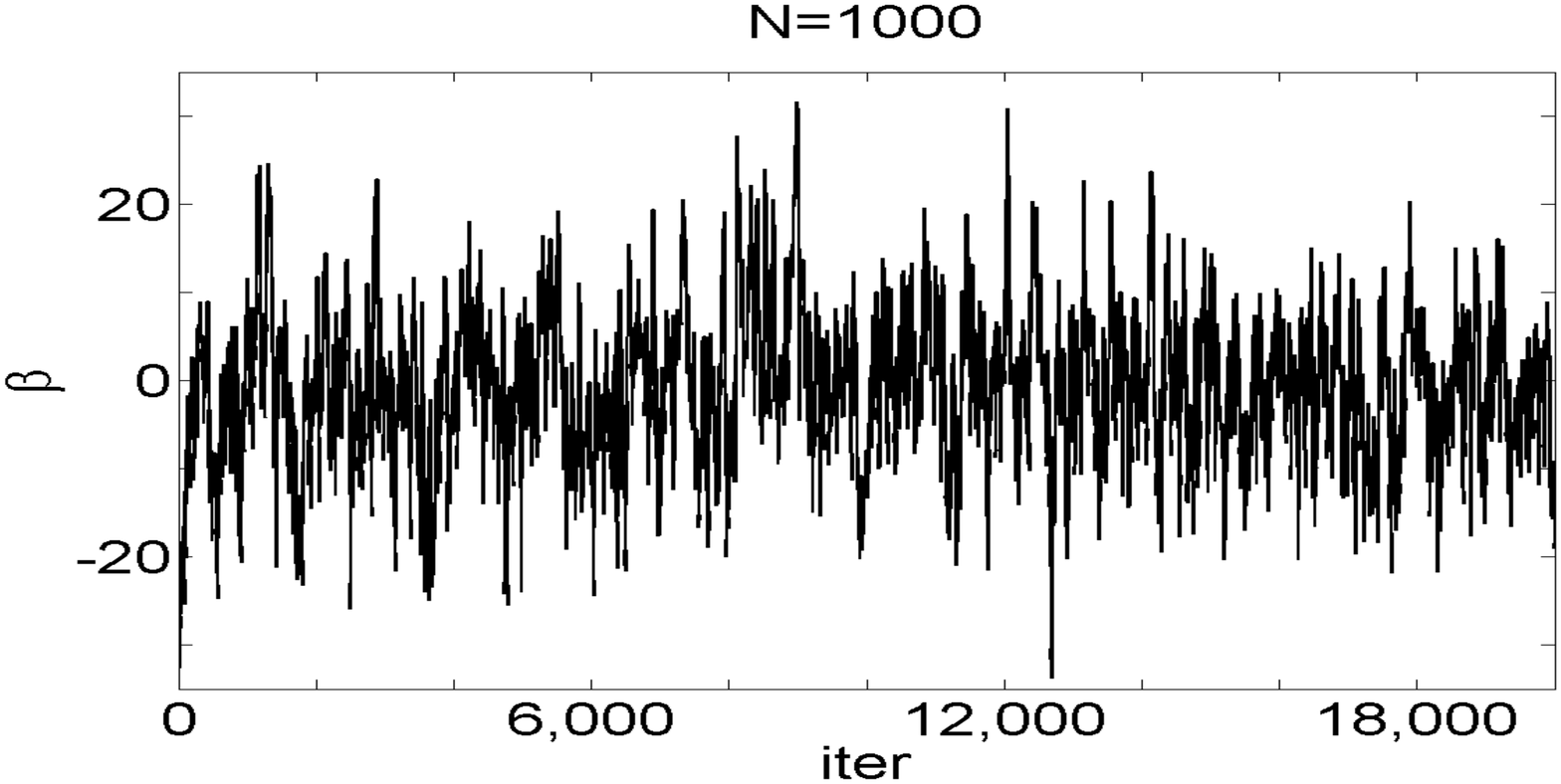}}
\end{center}
\end{figure}

\begin{figure}[H]
\begin{center}
\scalebox{0.2}{\includegraphics{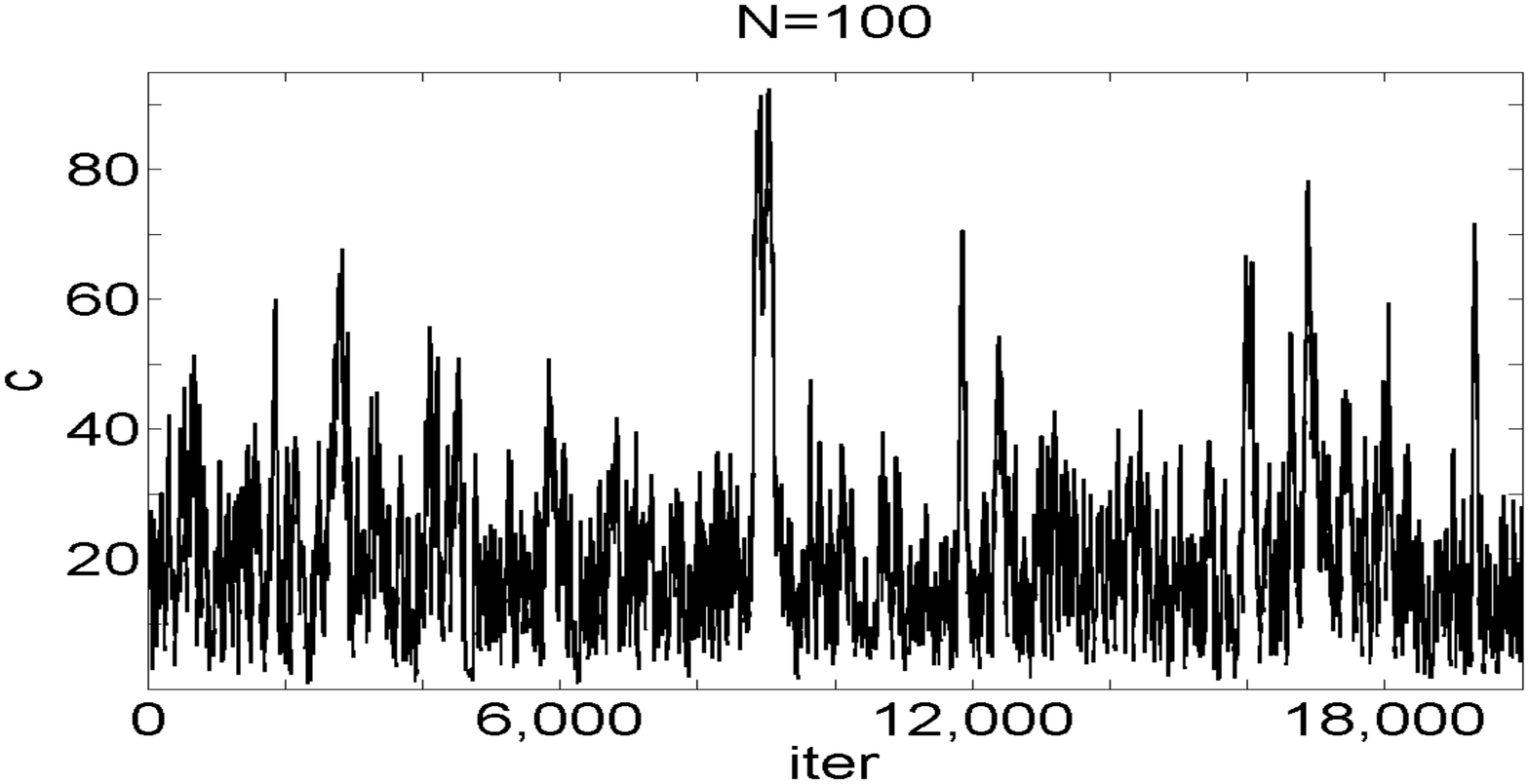}}\scalebox{0.2}{\includegraphics{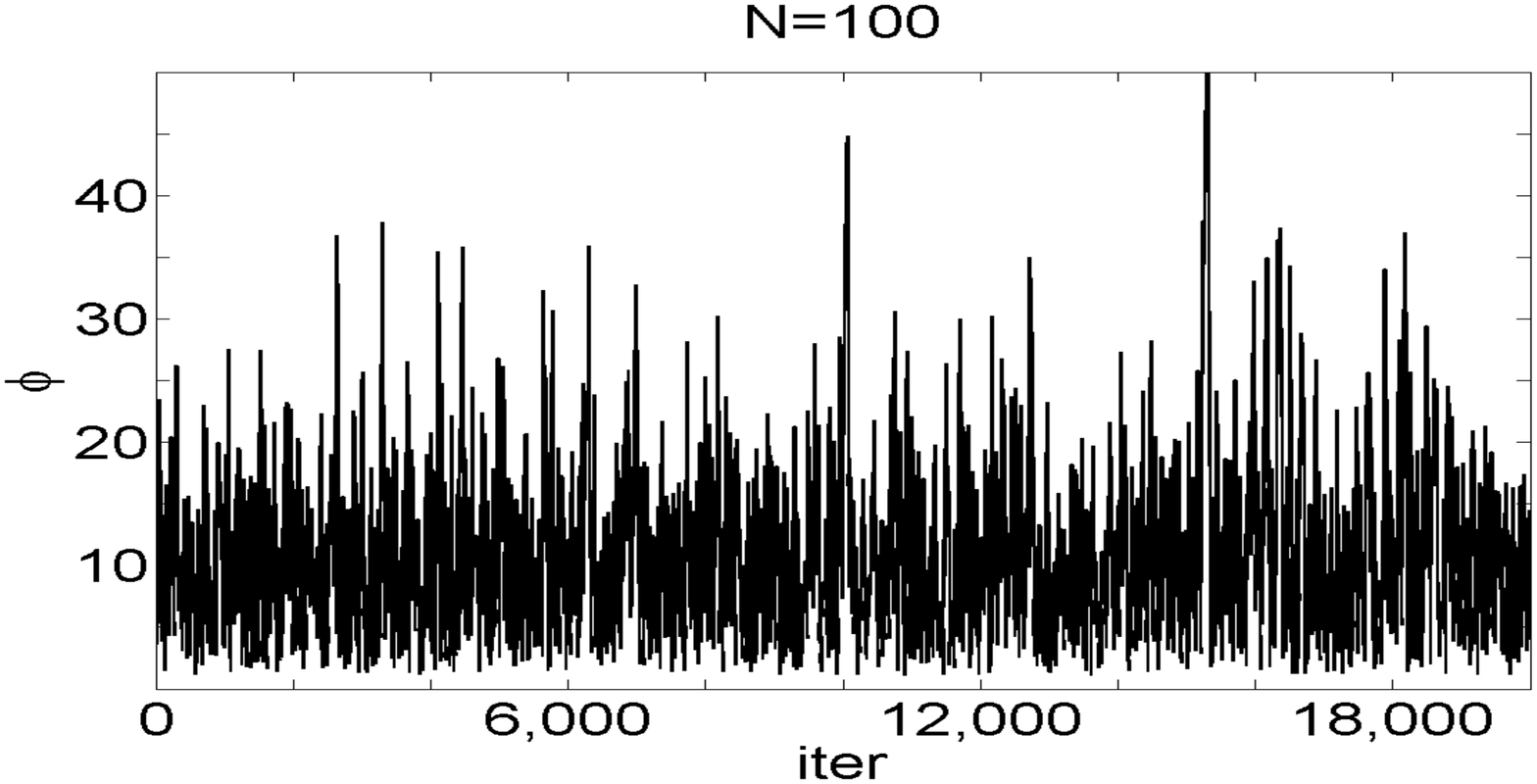}}\scalebox{0.2}{\includegraphics{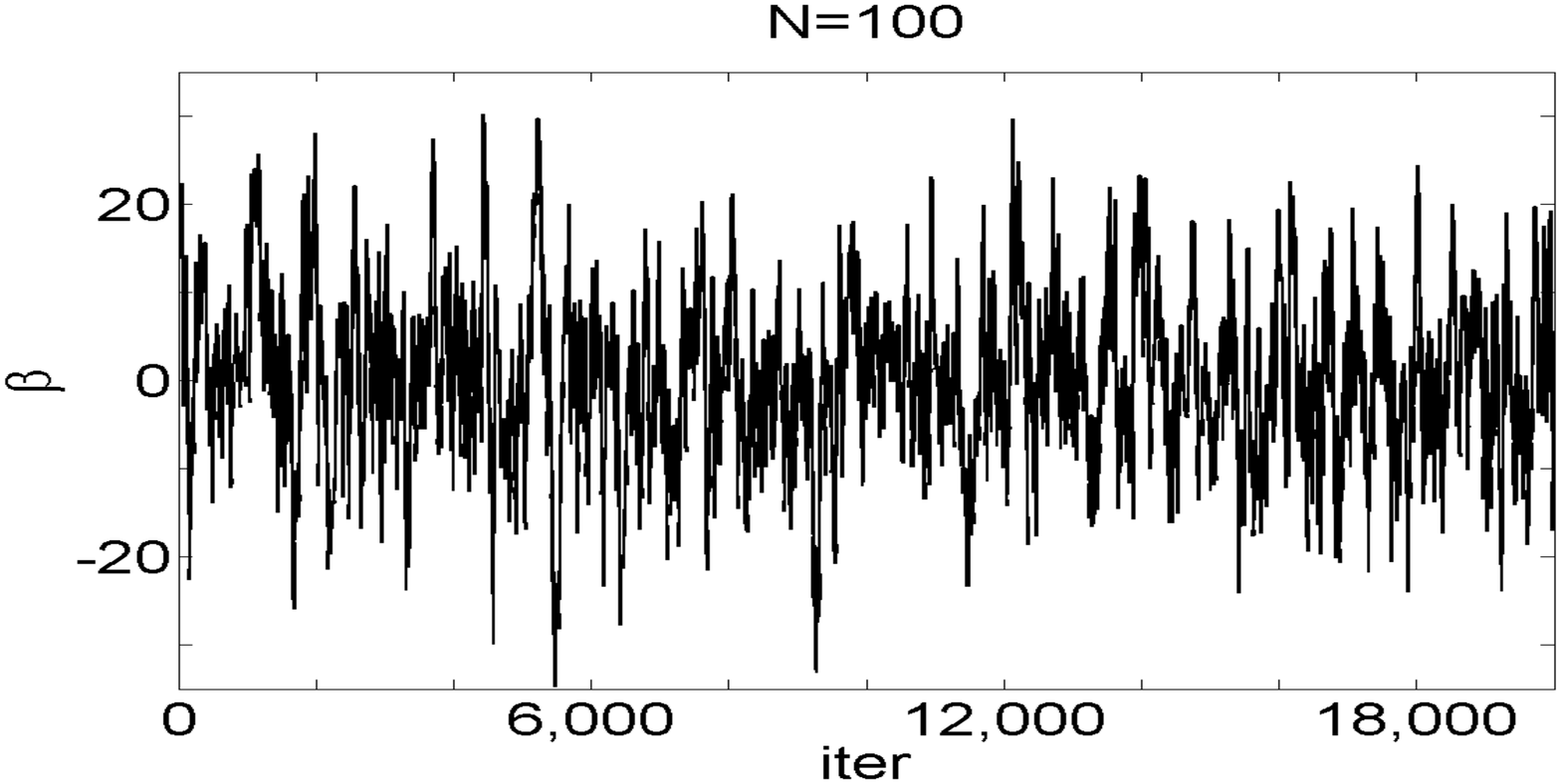}}
\end{center}
\end{figure}

\begin{figure}[H]
\begin{center}
\scalebox{0.2}{\includegraphics{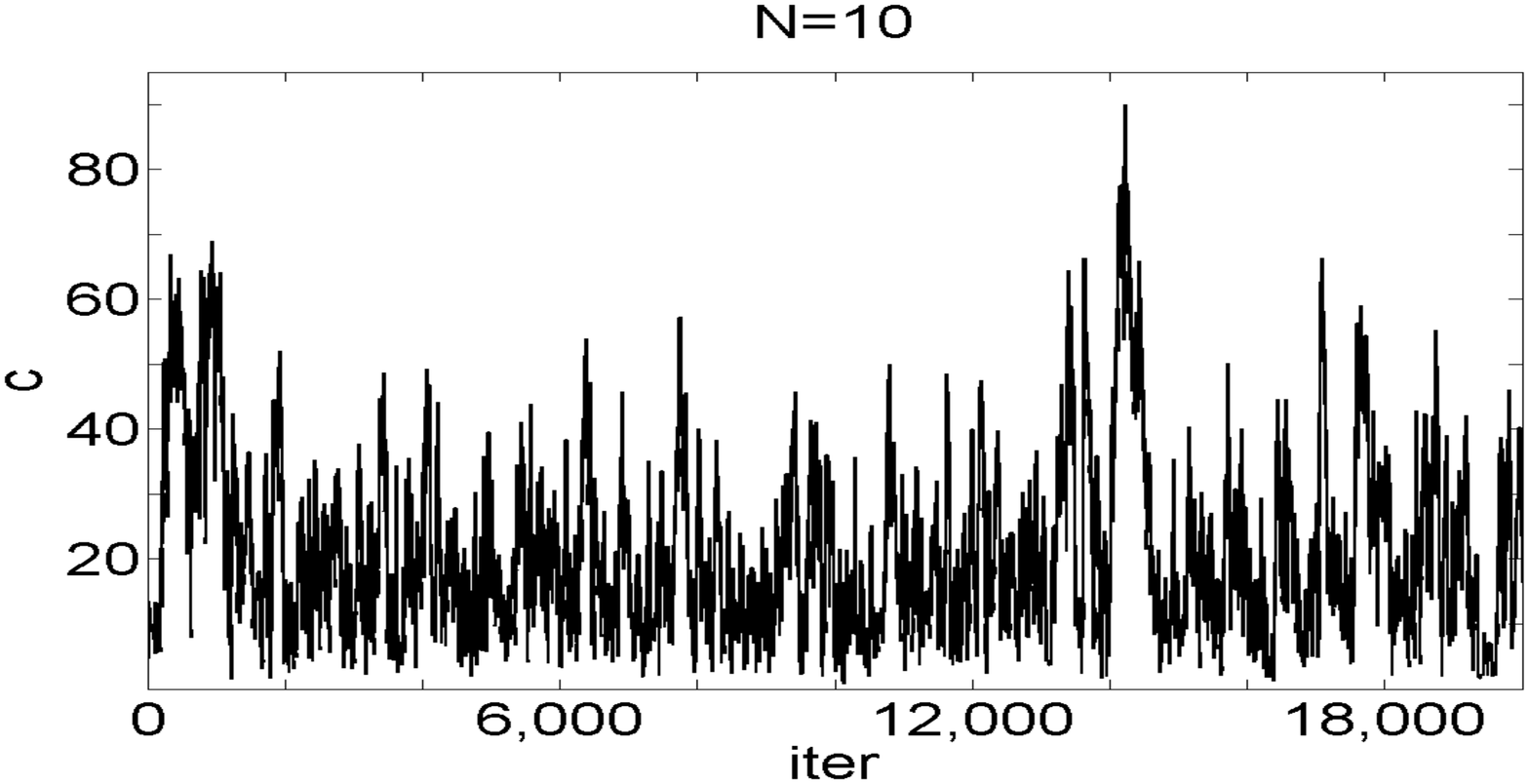}}\scalebox{0.2}{\includegraphics{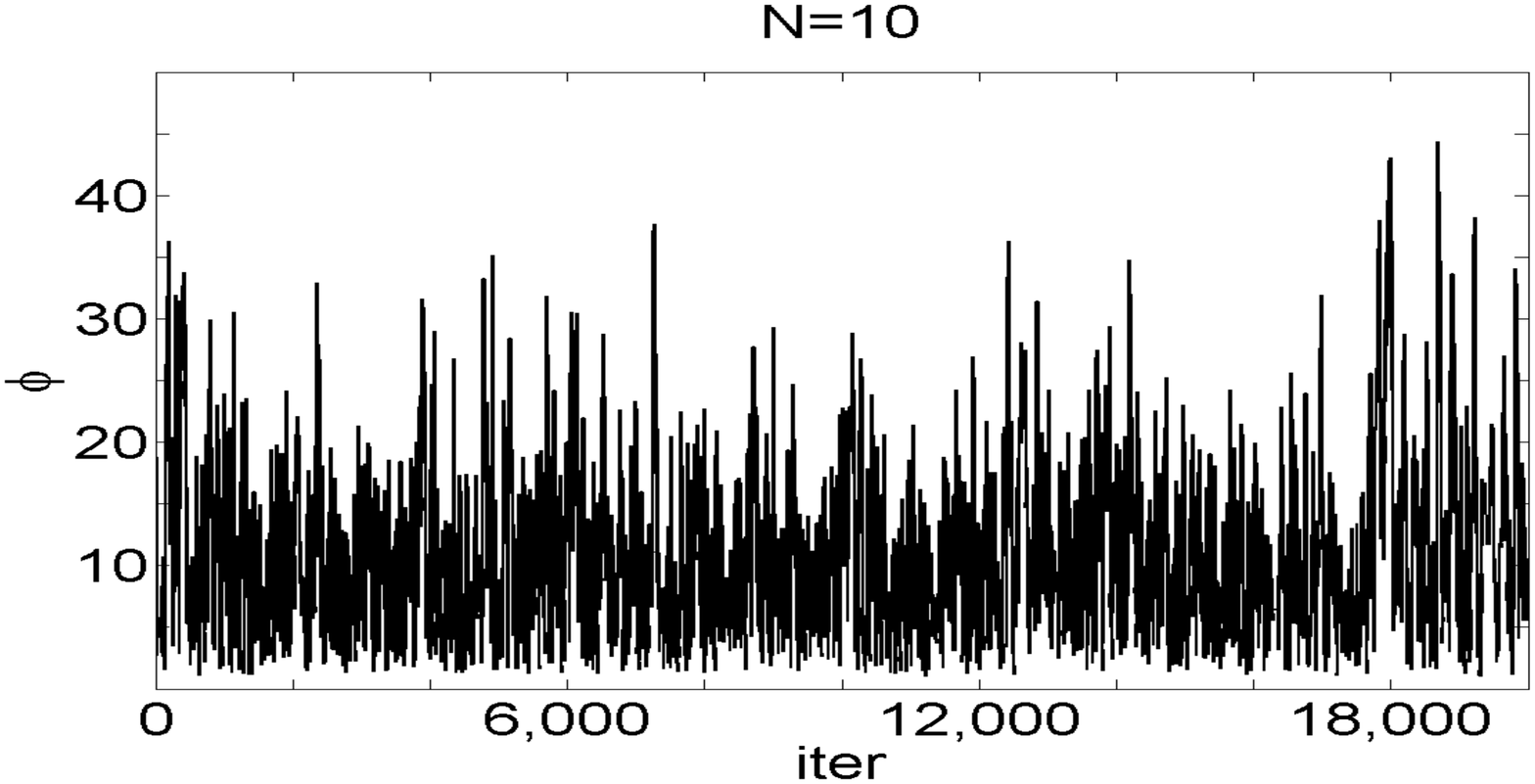}}\scalebox{0.2}{\includegraphics{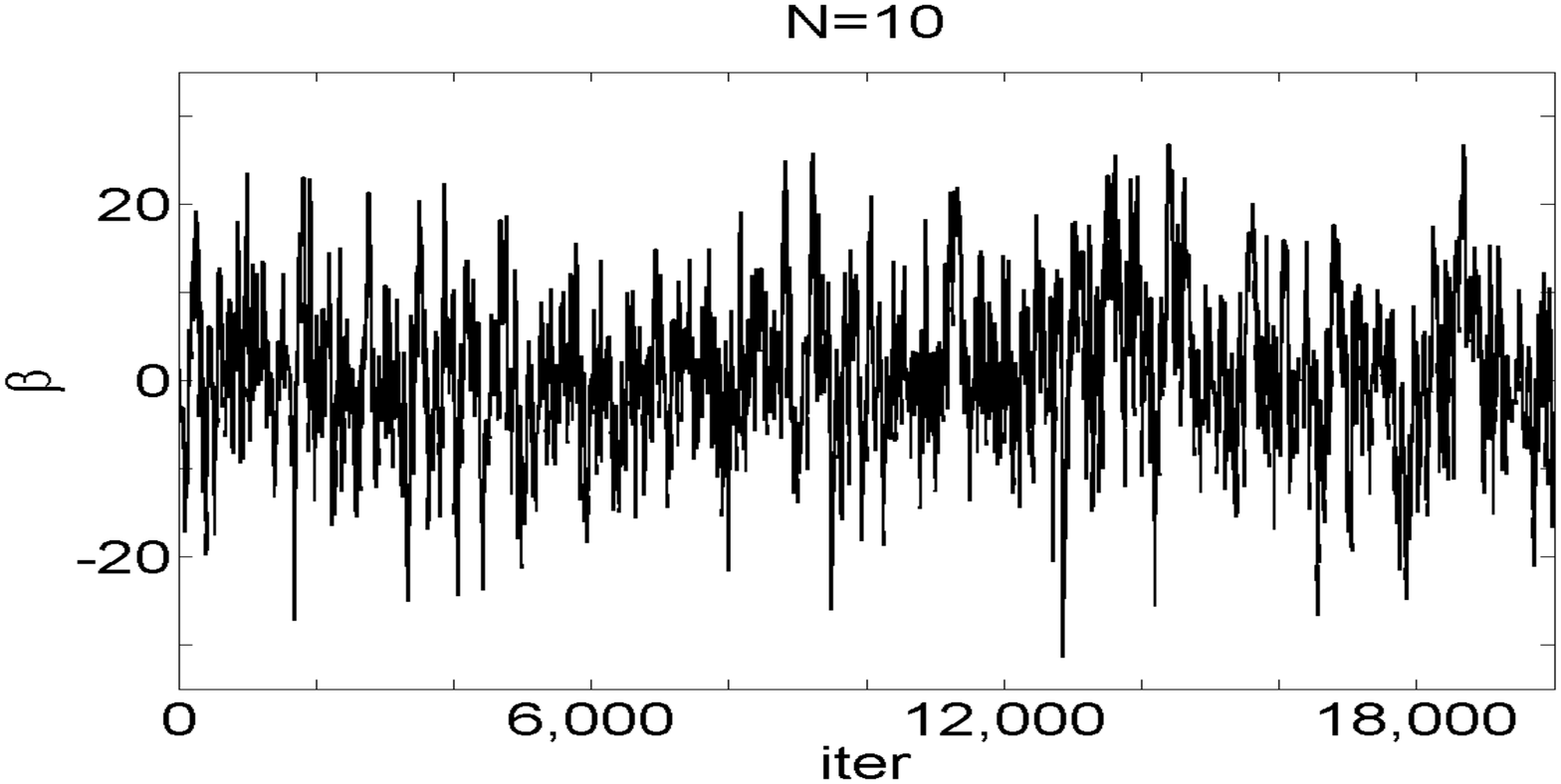}}
\caption{{\footnotesize Trace plot of each parameter across iterations for a PMMH algorithm using the SMC algorithm in Section \ref{sec:old_algo}. Each row displays the samples with different $N$}. Here $\xi_3=1.75$.}
\label{fig:traceplot_old}
\end{center}
\end{figure}

\begin{figure}[H]
\begin{center}
\scalebox{0.2}{\includegraphics{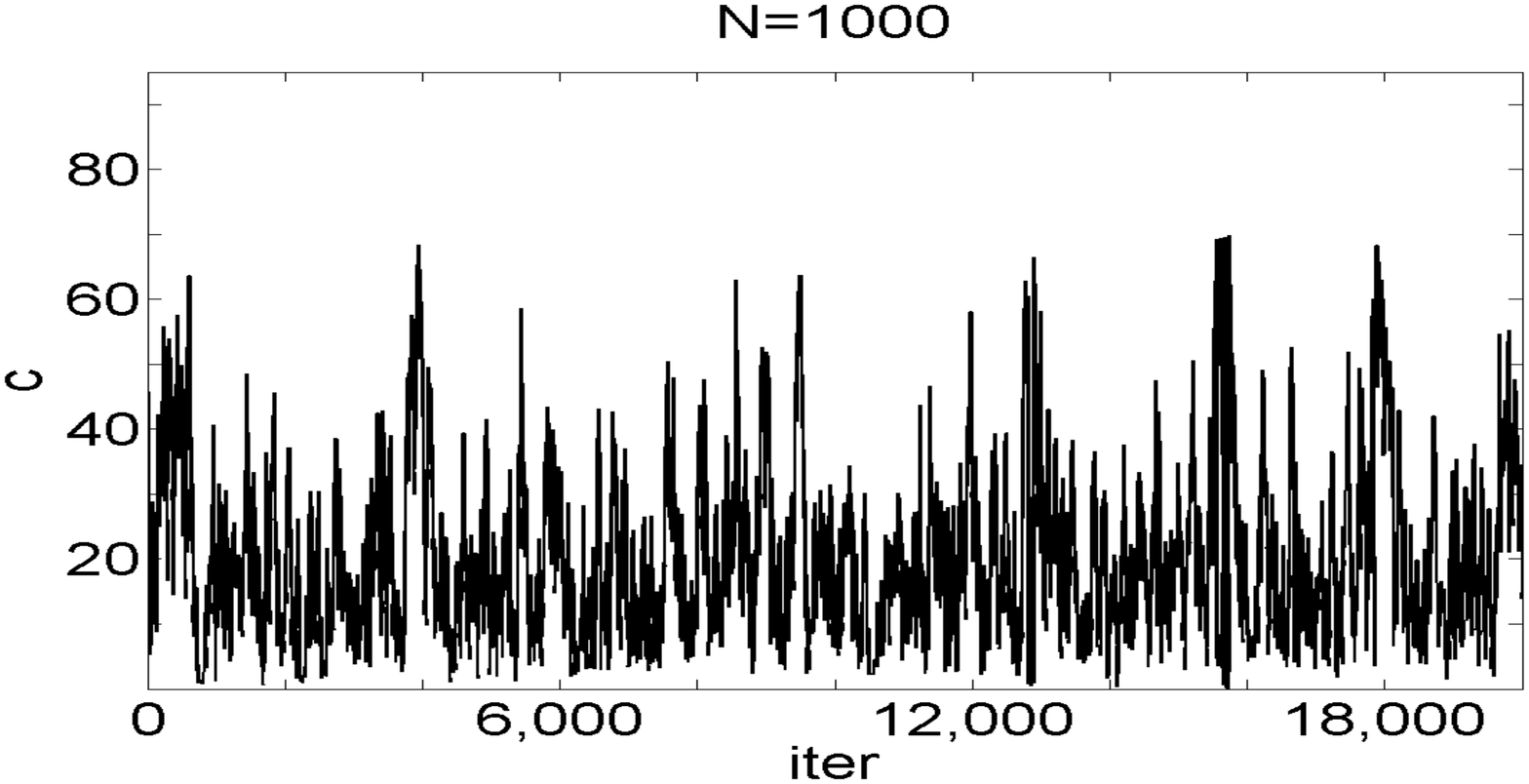}}\scalebox{0.2}{\includegraphics{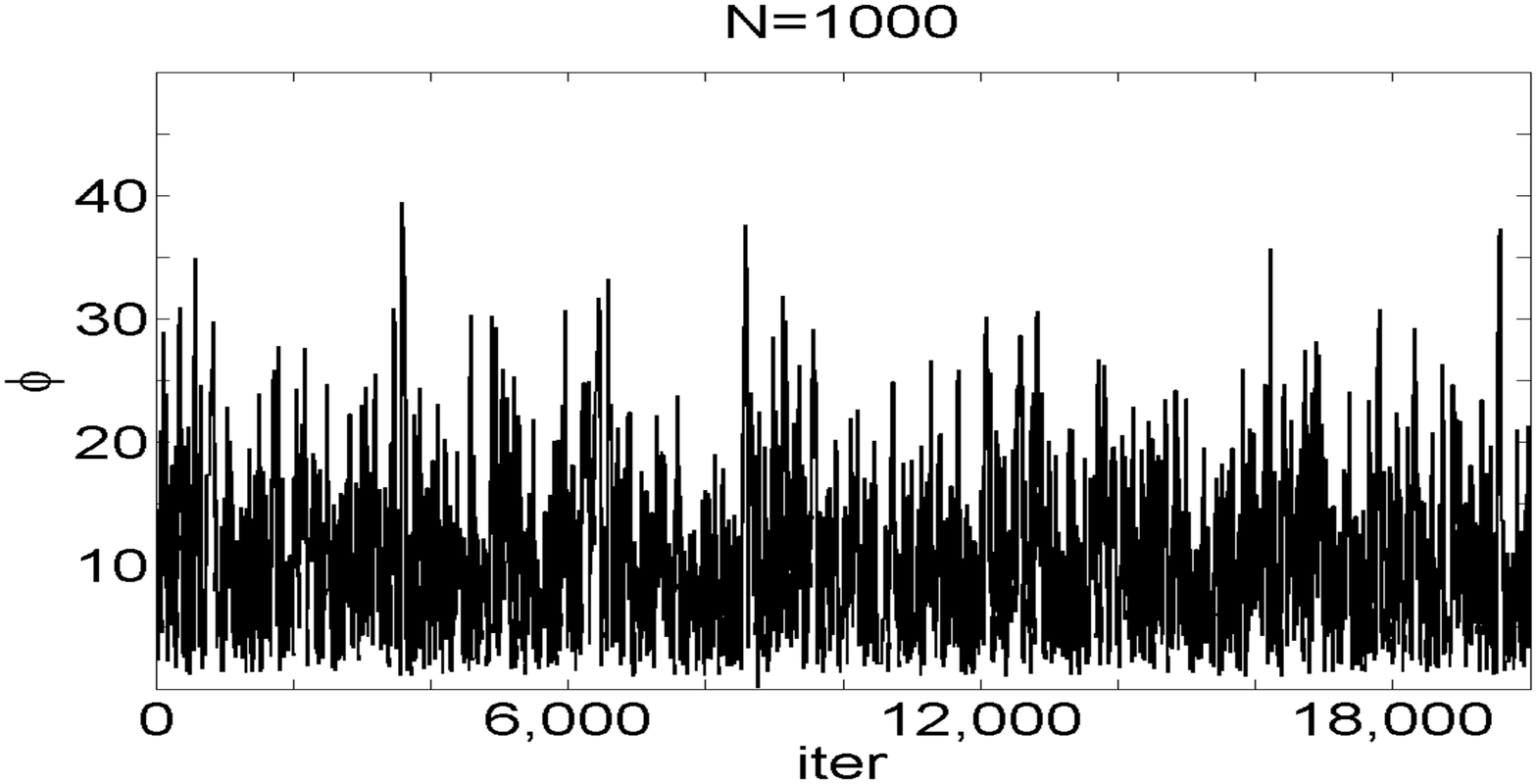}}\scalebox{0.2}{\includegraphics{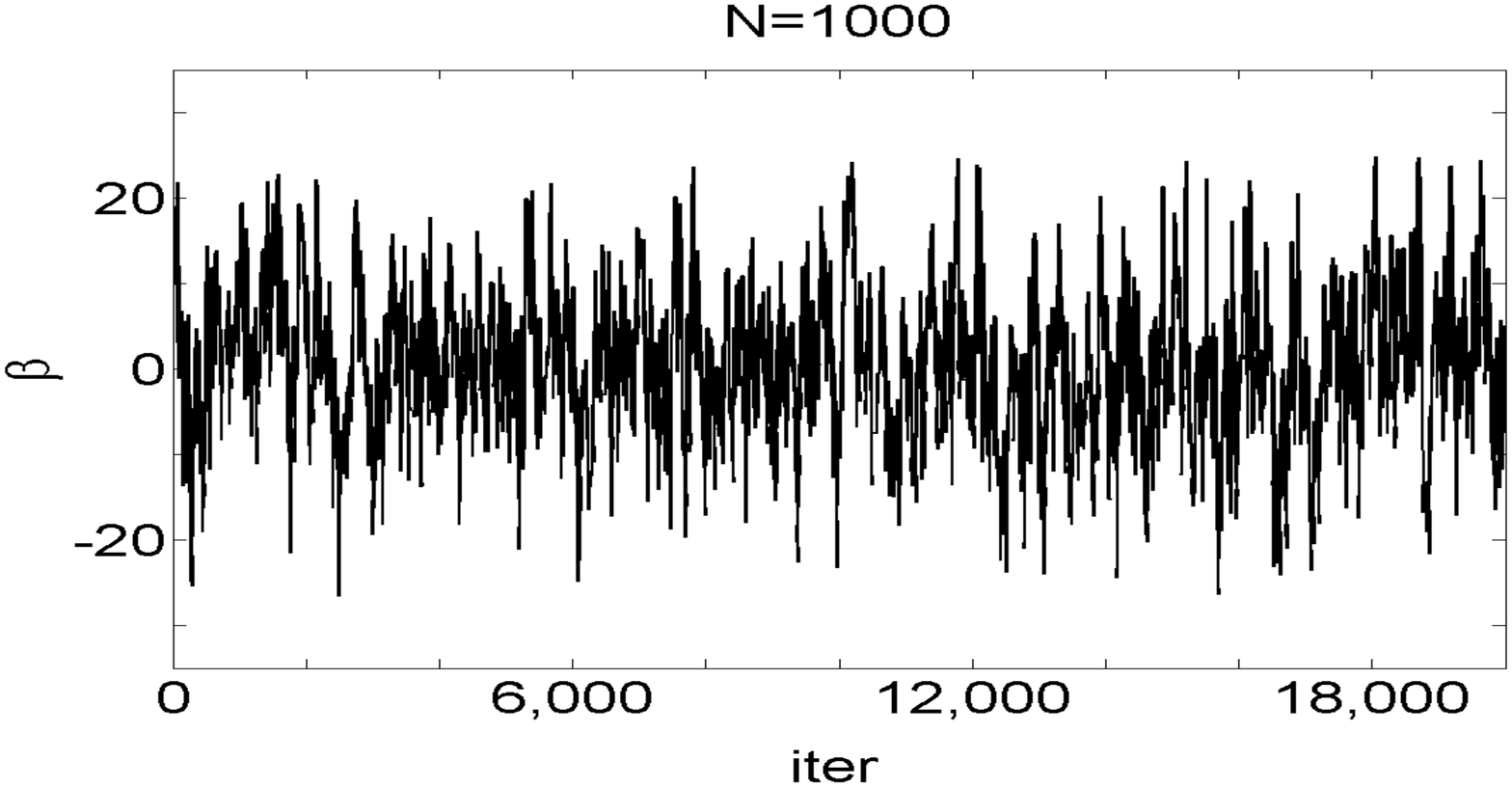}}
\end{center}
\end{figure}

\begin{figure}[H]
\begin{center}
\scalebox{0.2}{\includegraphics{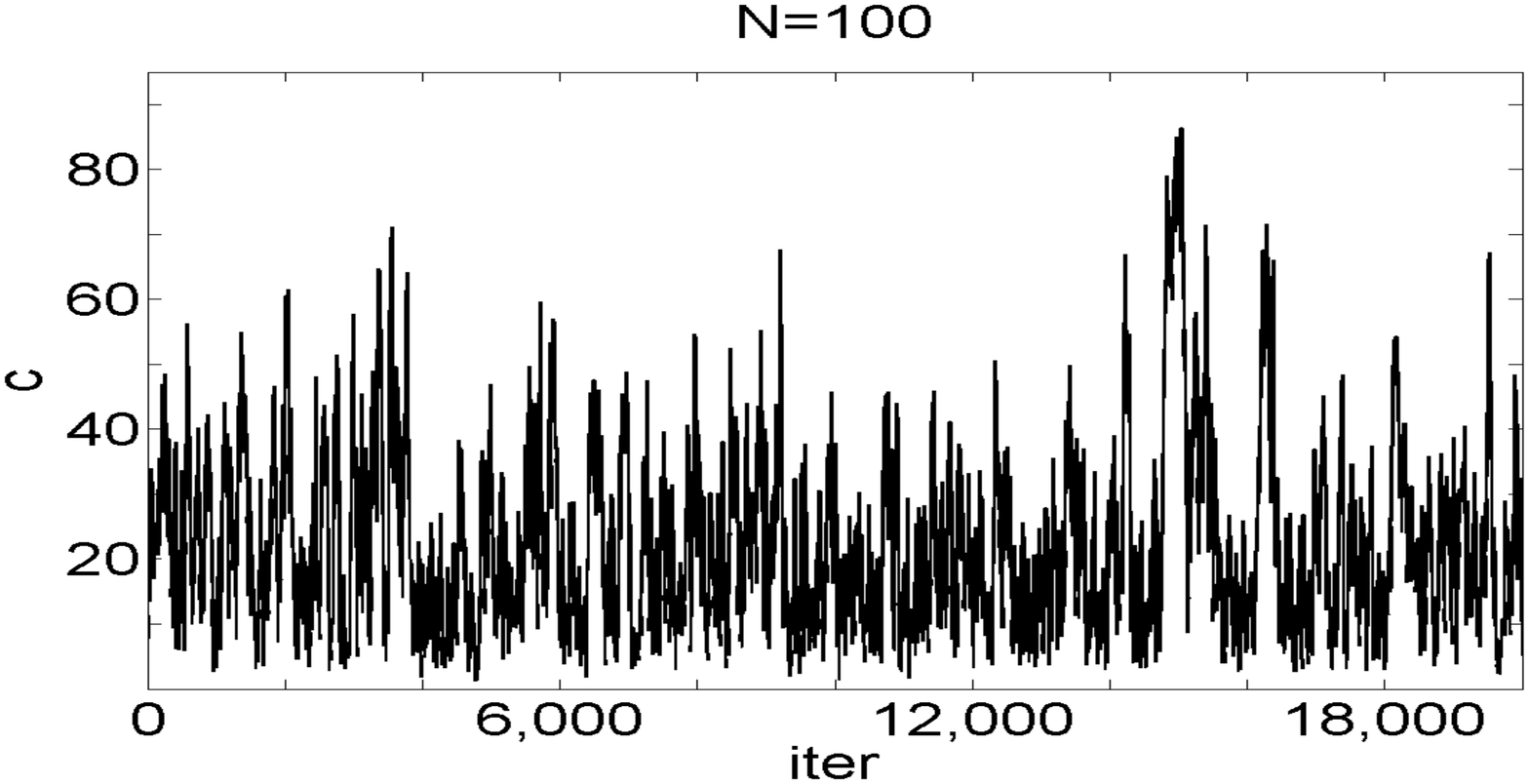}}\scalebox{0.2}{\includegraphics{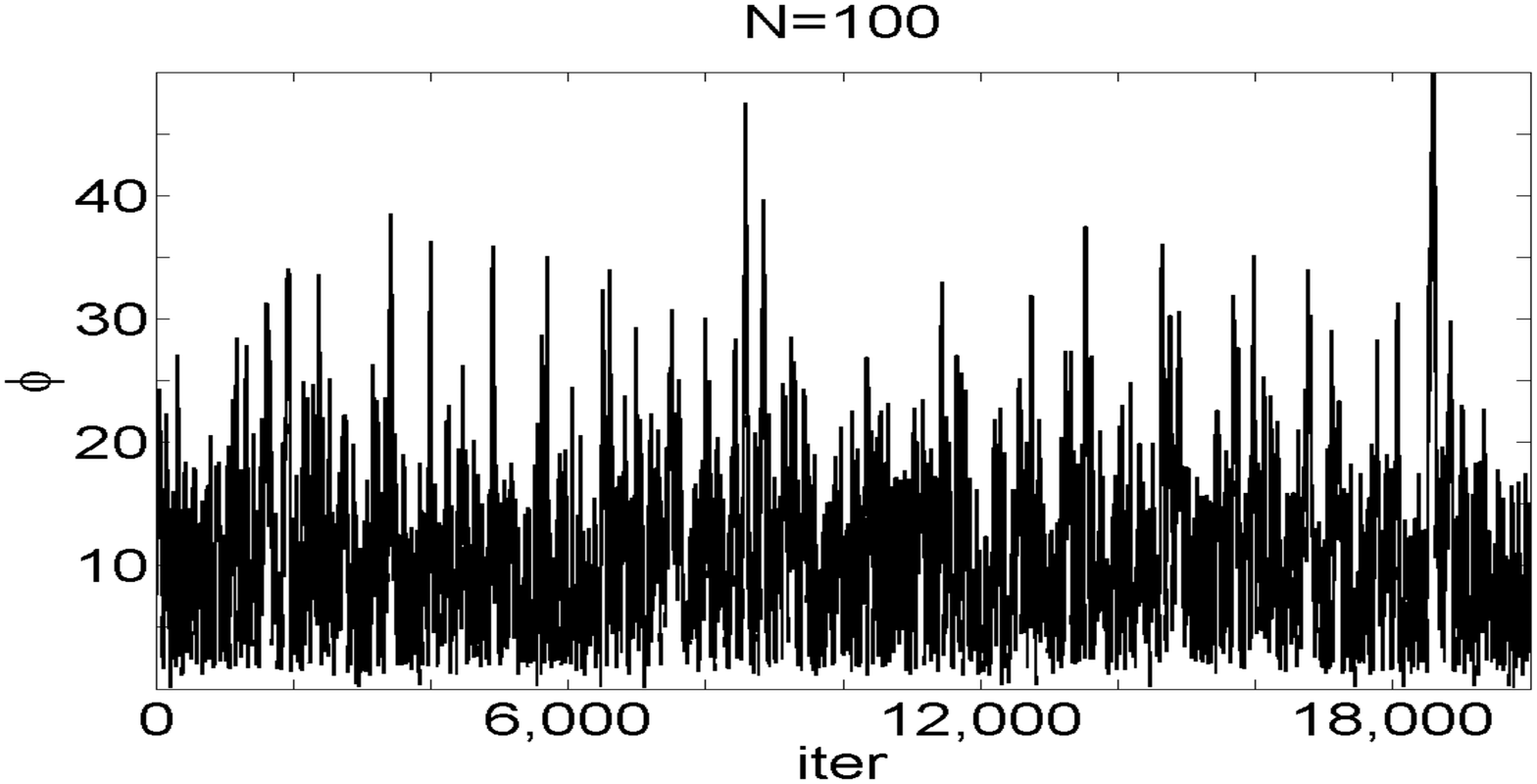}}\scalebox{0.2}{\includegraphics{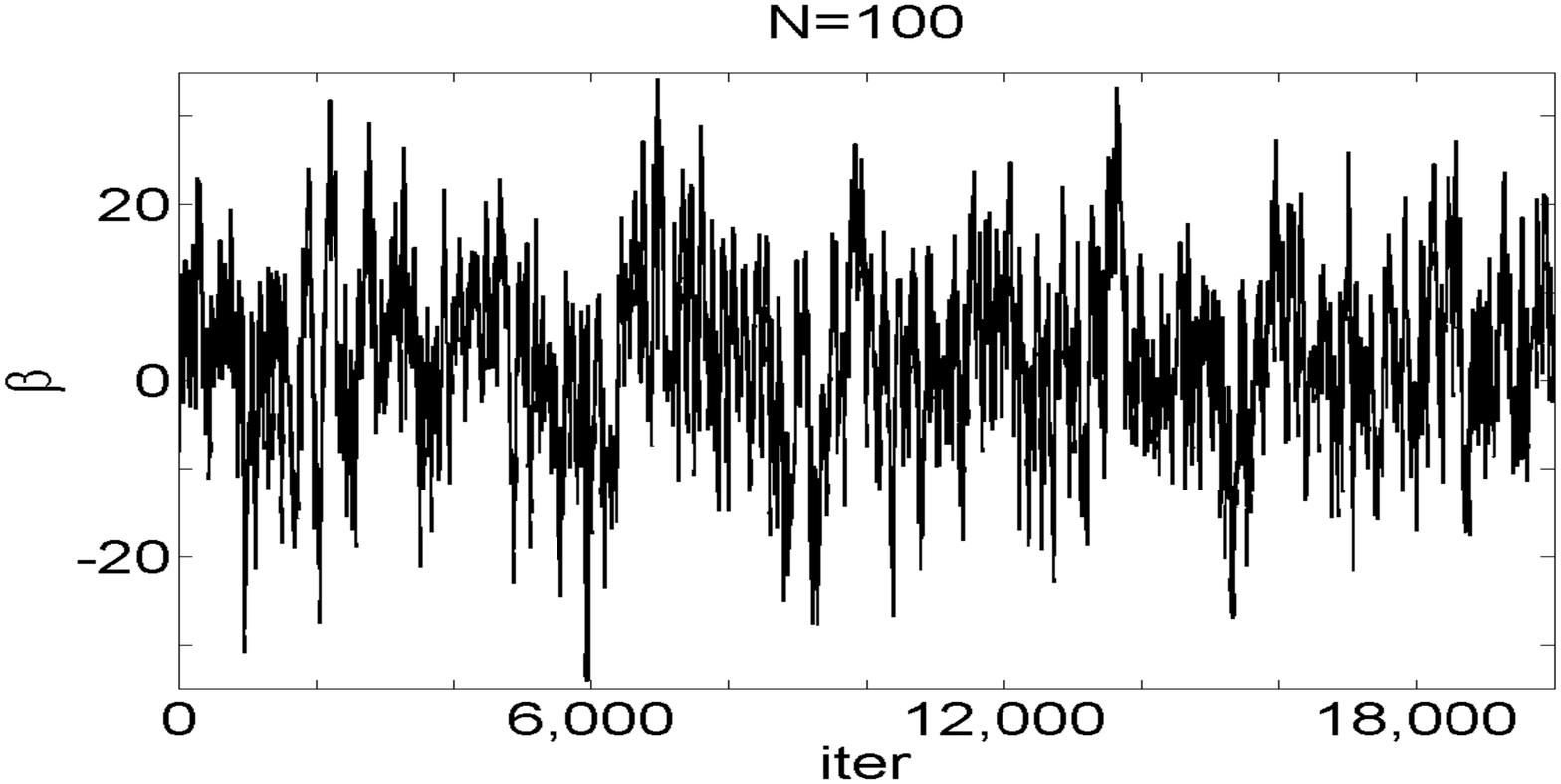}}
\end{center}
\end{figure}

\begin{figure}[H]
\begin{center}
\scalebox{0.2}{\includegraphics{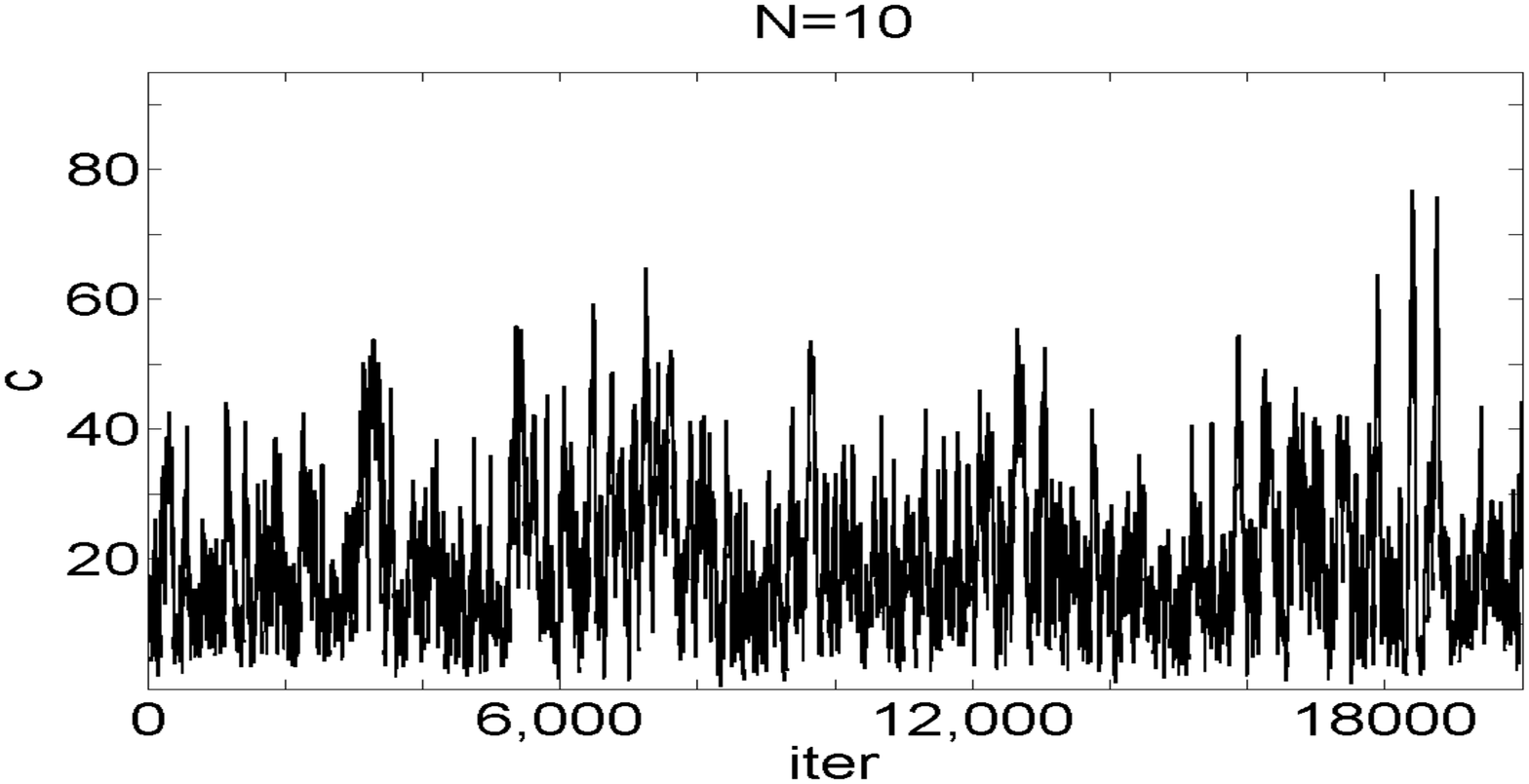}}\scalebox{0.2}{\includegraphics{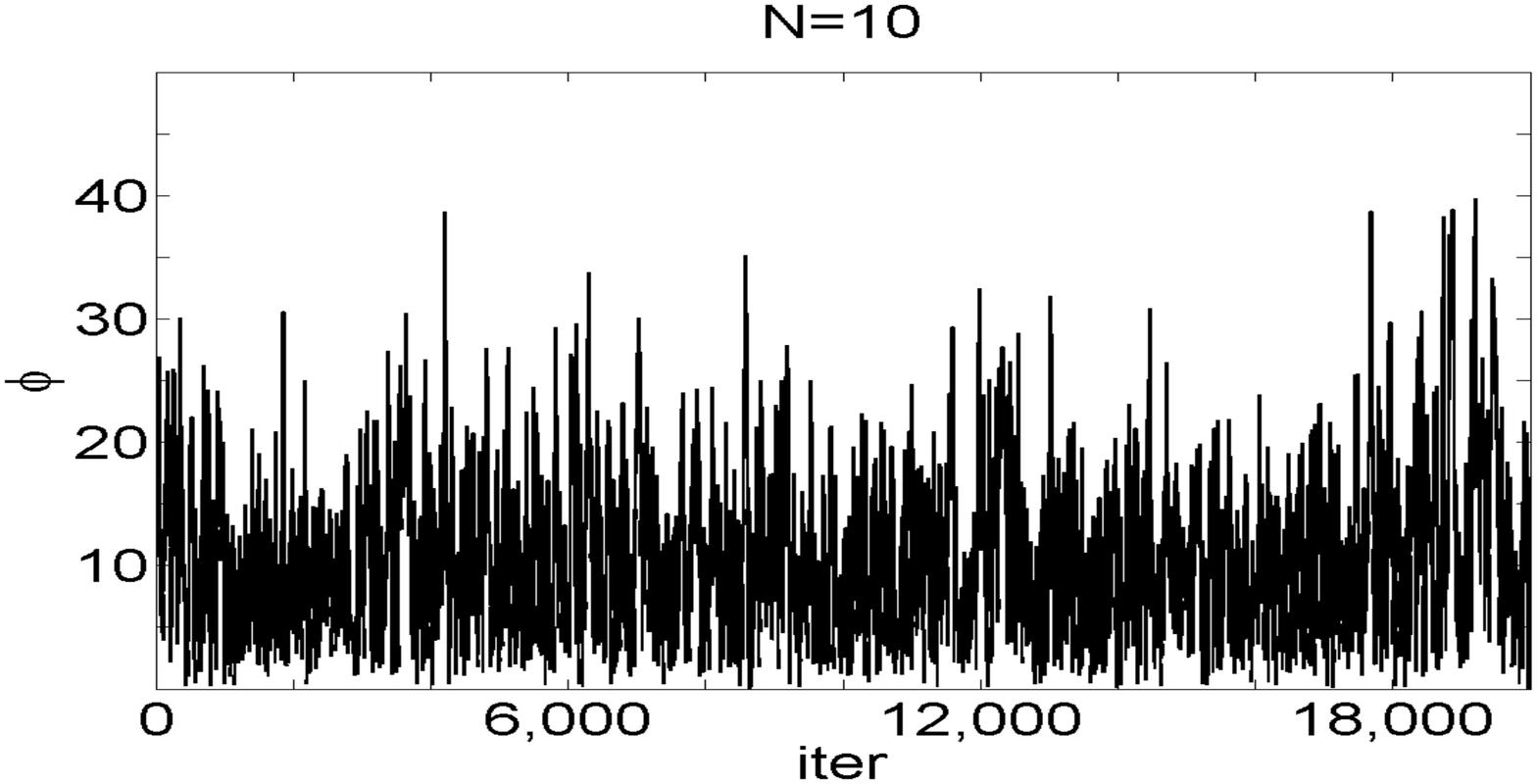}}\scalebox{0.2}{\includegraphics{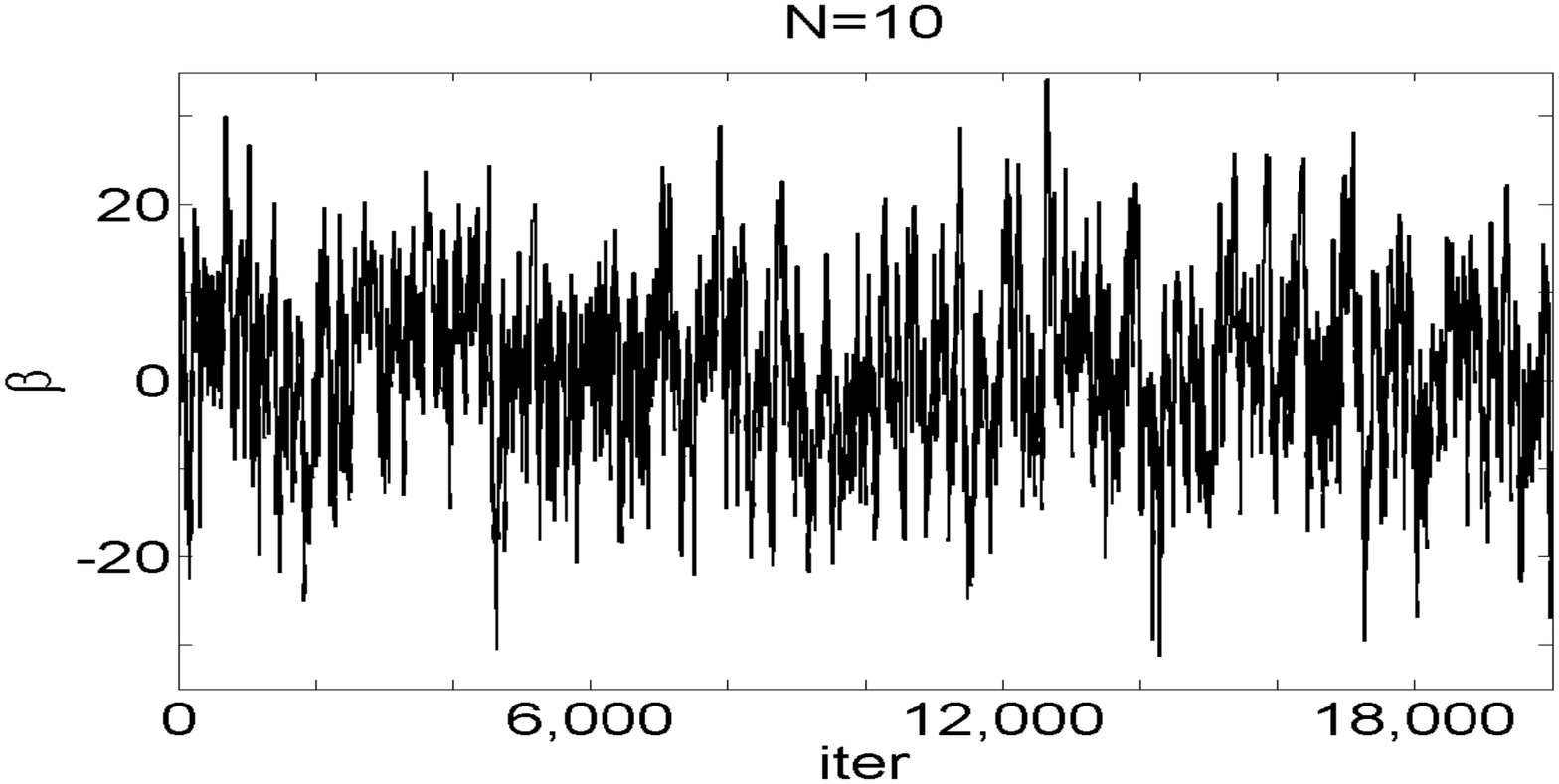}}
\caption{{\footnotesize  Trace plot of each parameter across iterations for a PMMH algorithm using the SMC algorithm in Section \ref{sec:new_smc}. Each row displays the samples with different $N$}. Here $\xi_3=1.75$.}\label{fig:traceplot_new}
\end{center}
\end{figure}

%
%

\begin{figure}[H]
\begin{center}
\scalebox{0.2}{\includegraphics{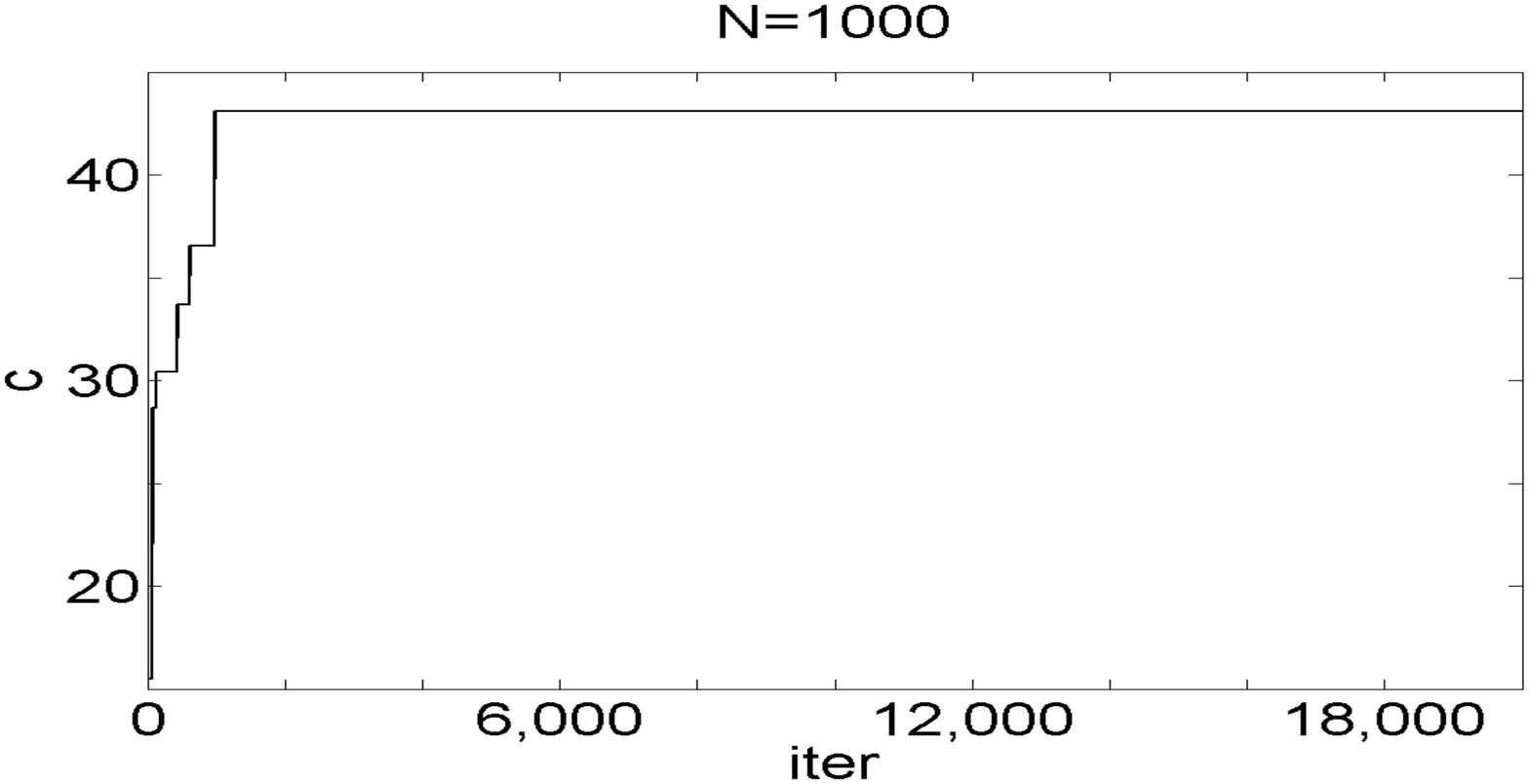}}\scalebox{0.2}{\includegraphics{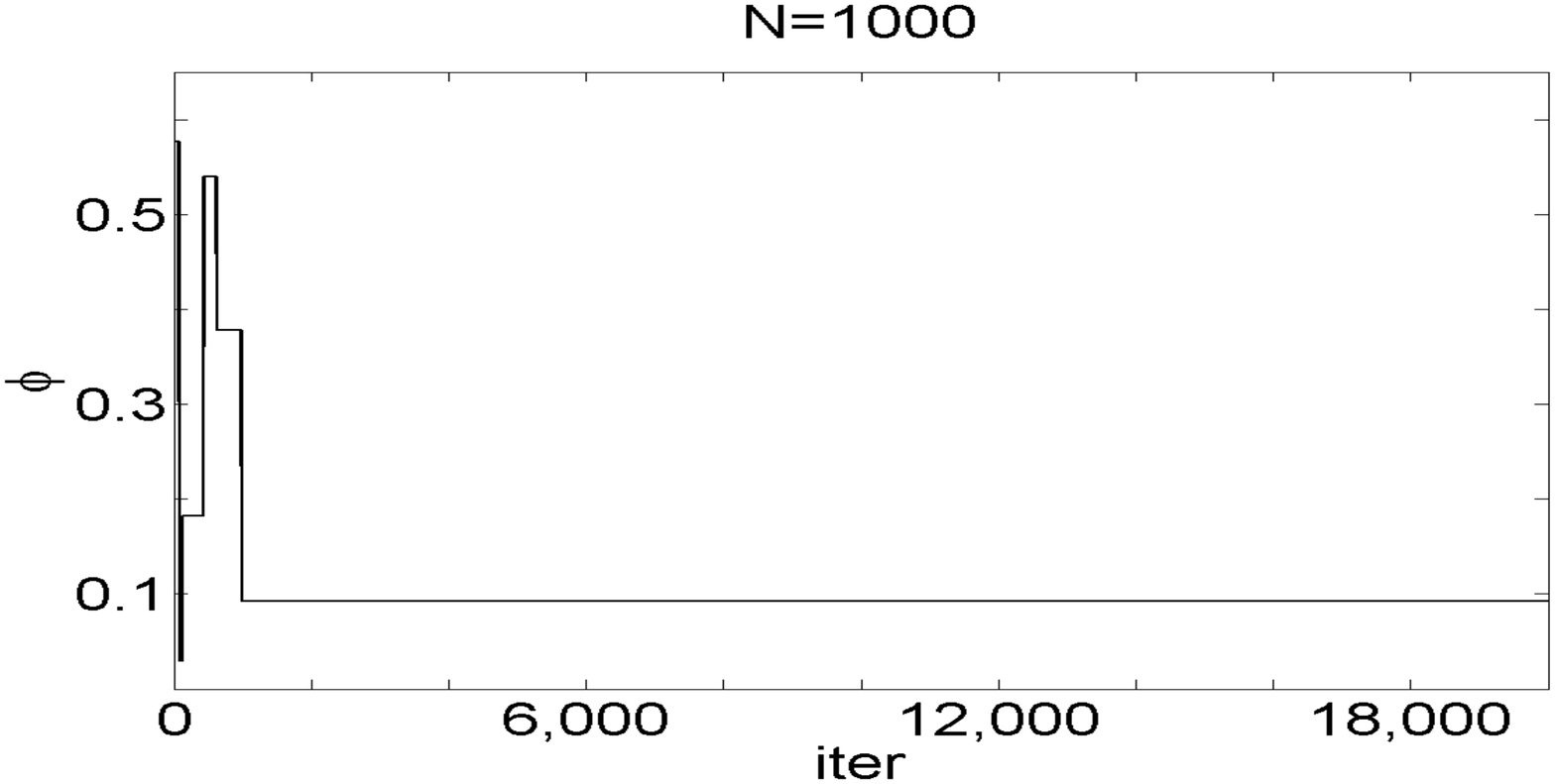}}\scalebox{0.2}{\includegraphics{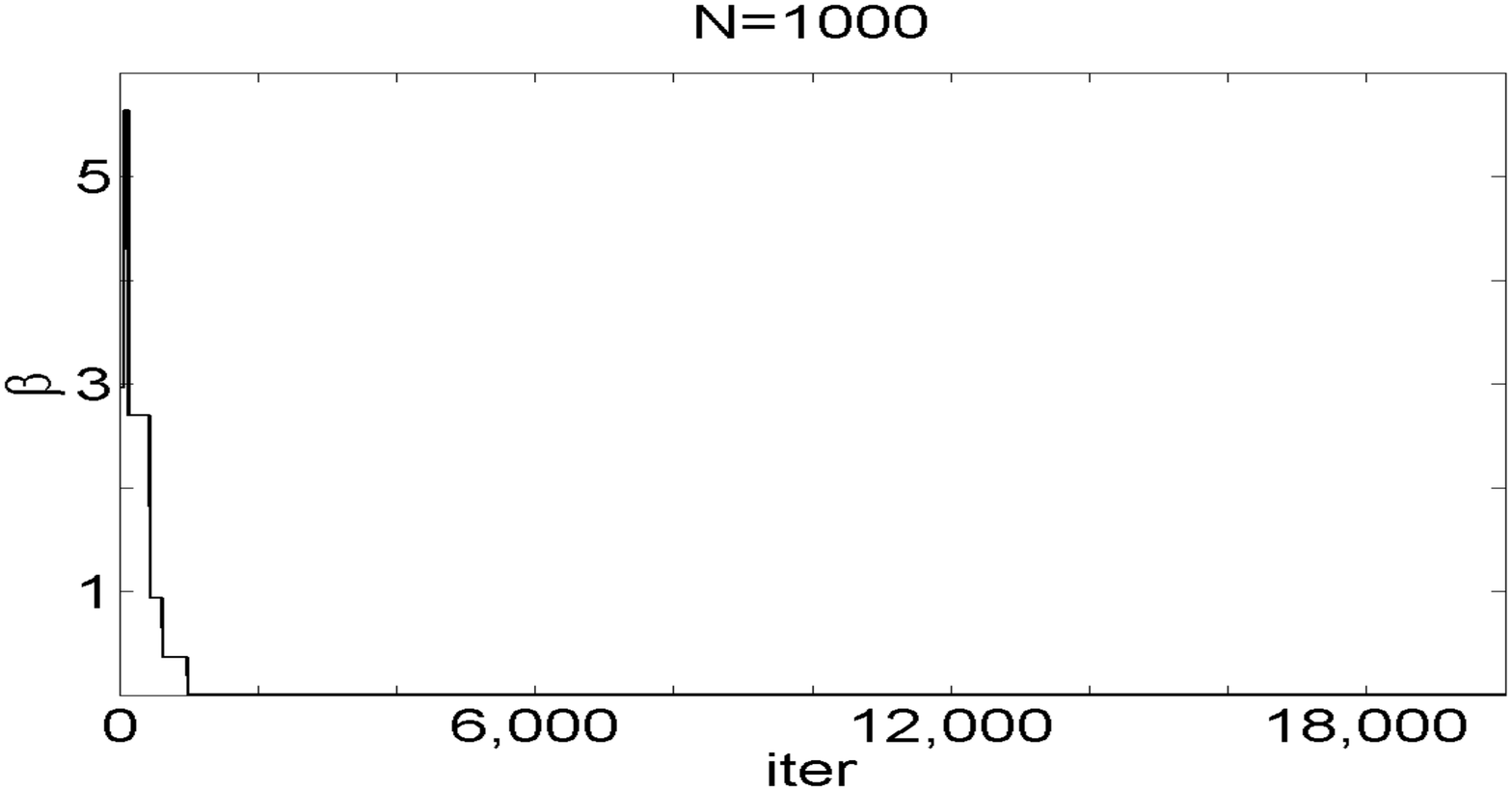}}
\end{center}
\end{figure}

\begin{figure}[H]
\begin{center}
\scalebox{0.2}{\includegraphics{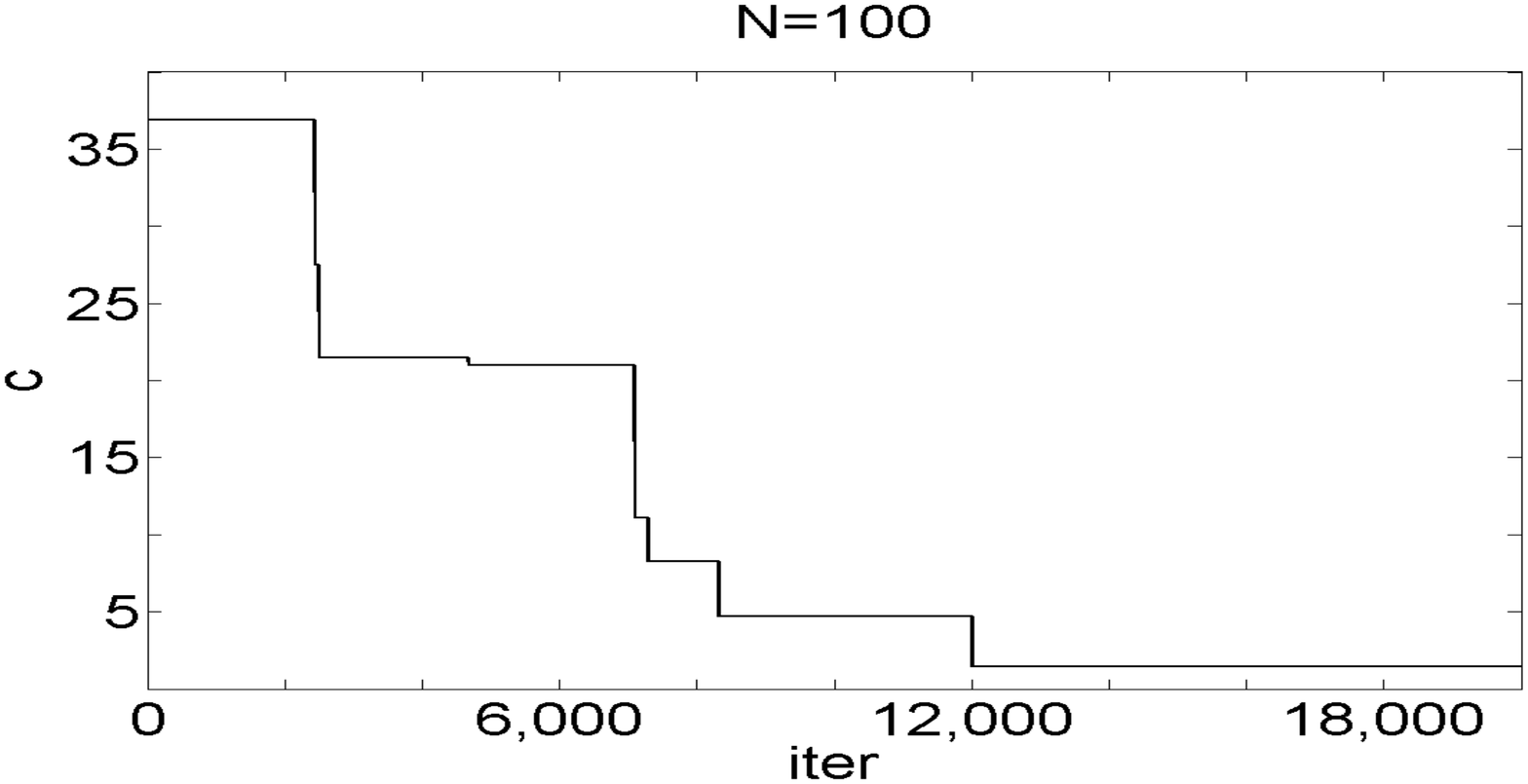}}\scalebox{0.2}{\includegraphics{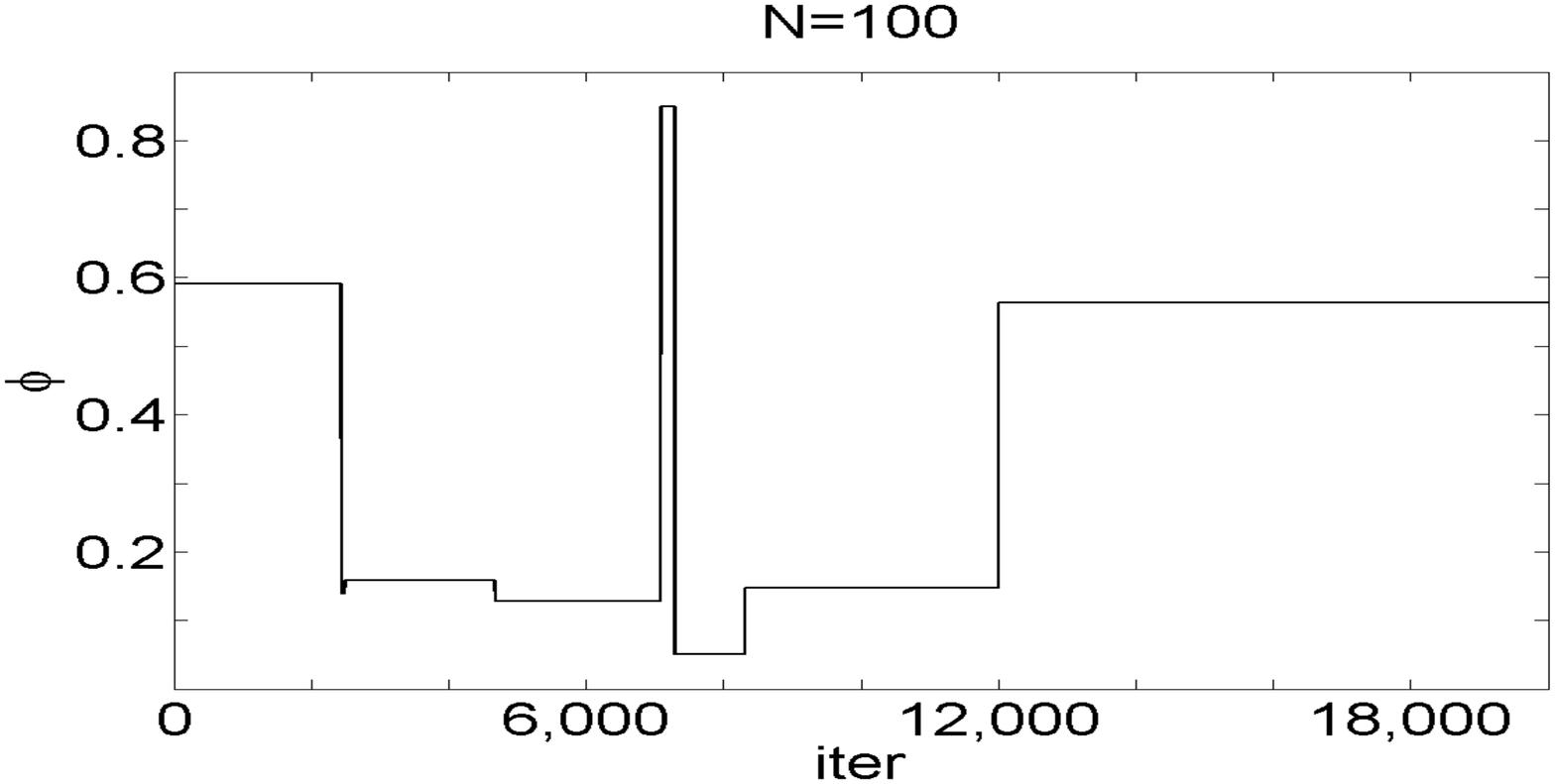}}\scalebox{0.2}{\includegraphics{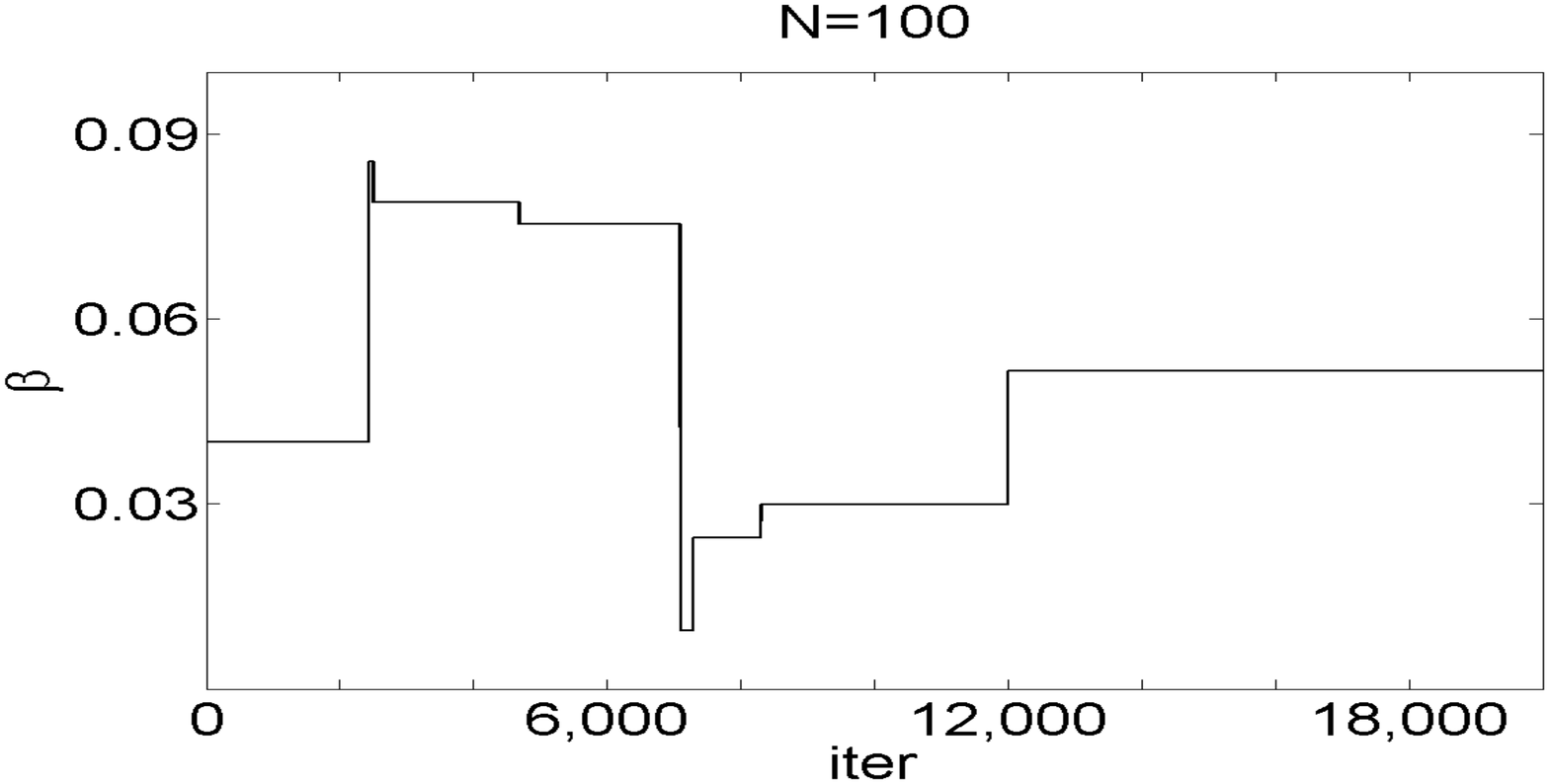}}
\end{center}
\end{figure}

\begin{figure}[H]
\begin{center}
\scalebox{0.2}{\includegraphics{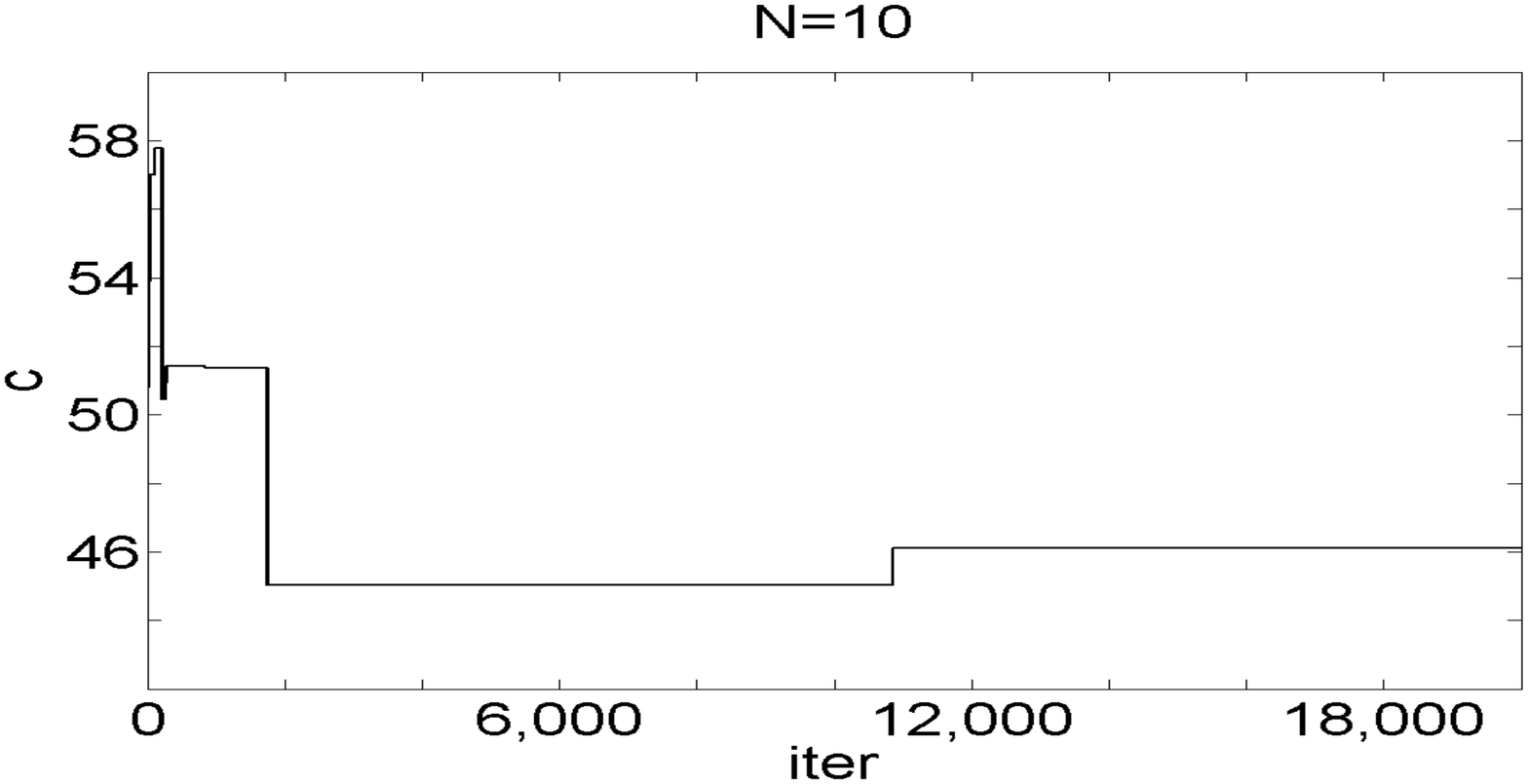}}\scalebox{0.2}{\includegraphics{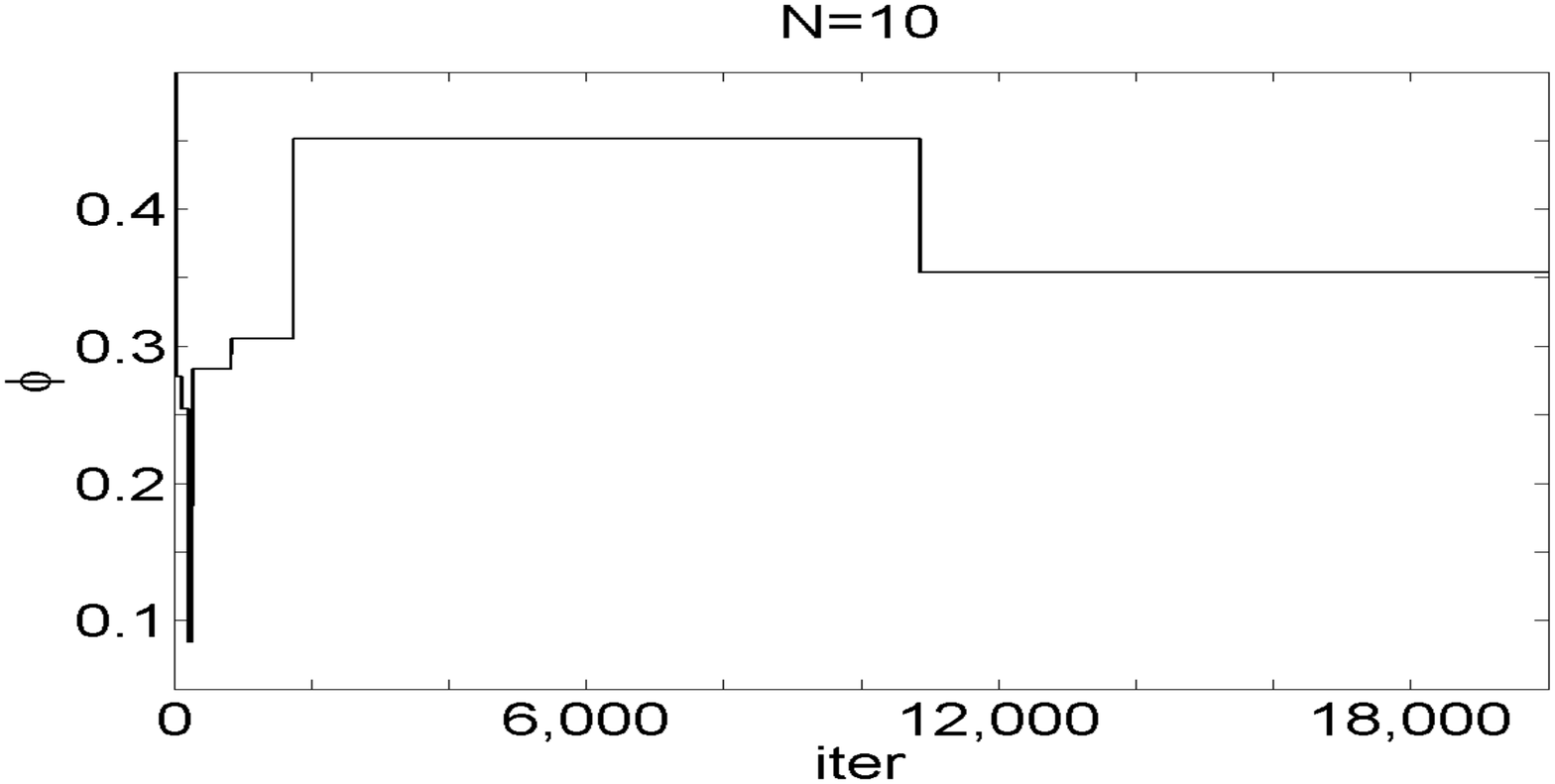}}\scalebox{0.2}{\includegraphics{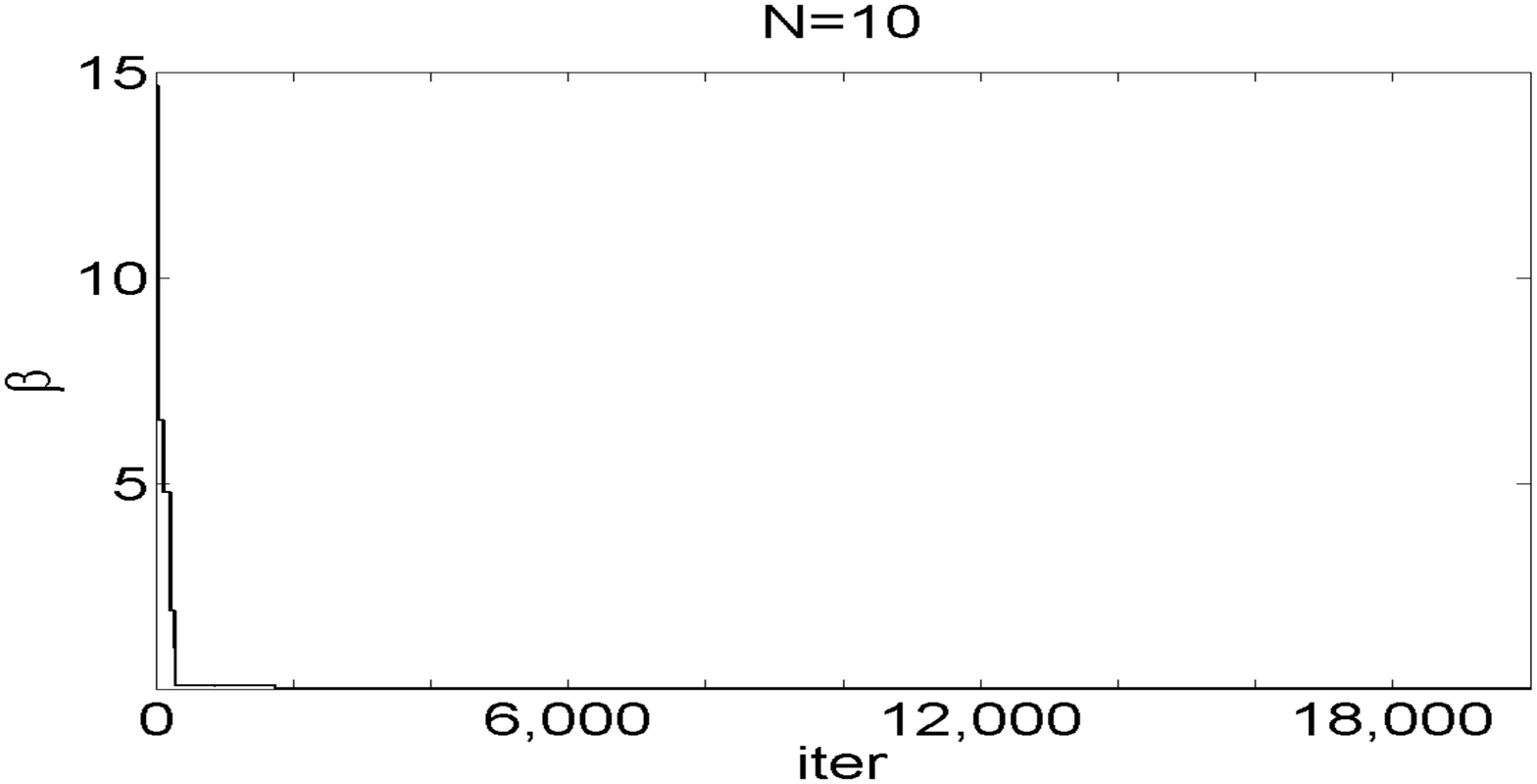}}
\caption{{\footnotesize Trace plot of each parameter across iterations for a PMMH algorithm using the SMC algorithm in Section \ref{sec:old_algo}. Each row displays the samples with different $N$}. Here $\xi_3=1.2$.}\label{fig:traceplot_col_old}
\end{center}
\end{figure}

\begin{figure}[H]
\begin{center}
\scalebox{0.2}{\includegraphics{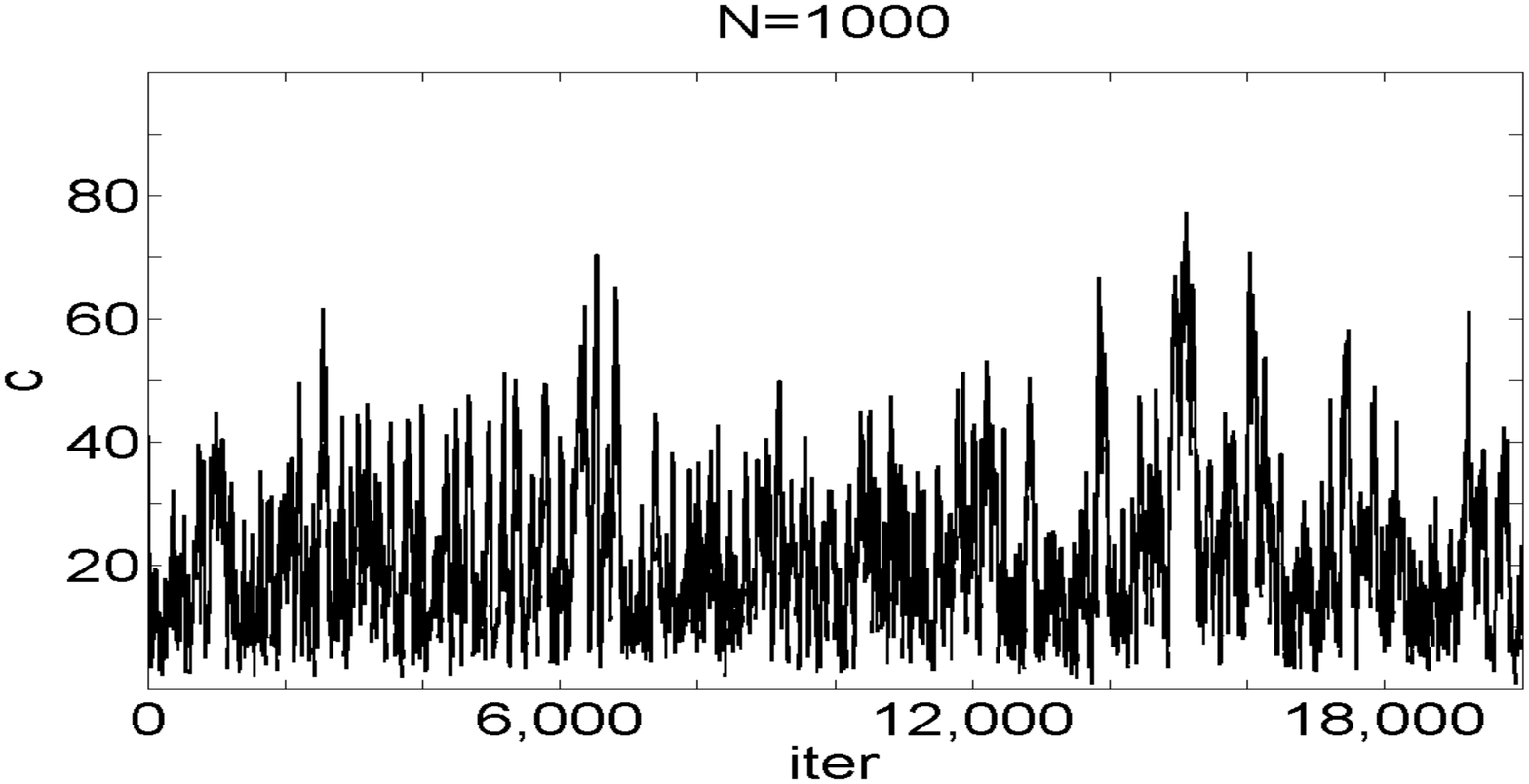}}\scalebox{0.2}{\includegraphics{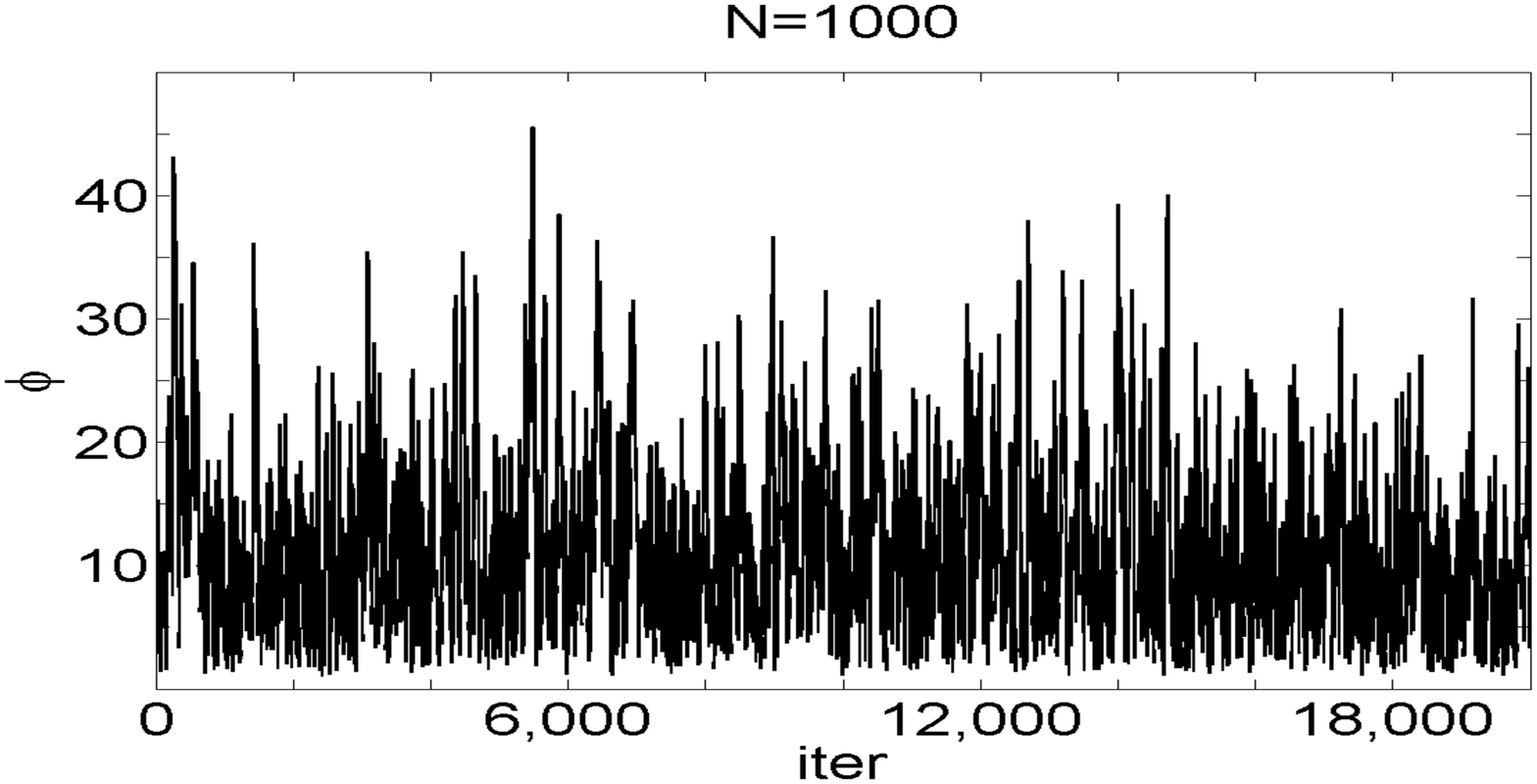}}\scalebox{0.2}{\includegraphics{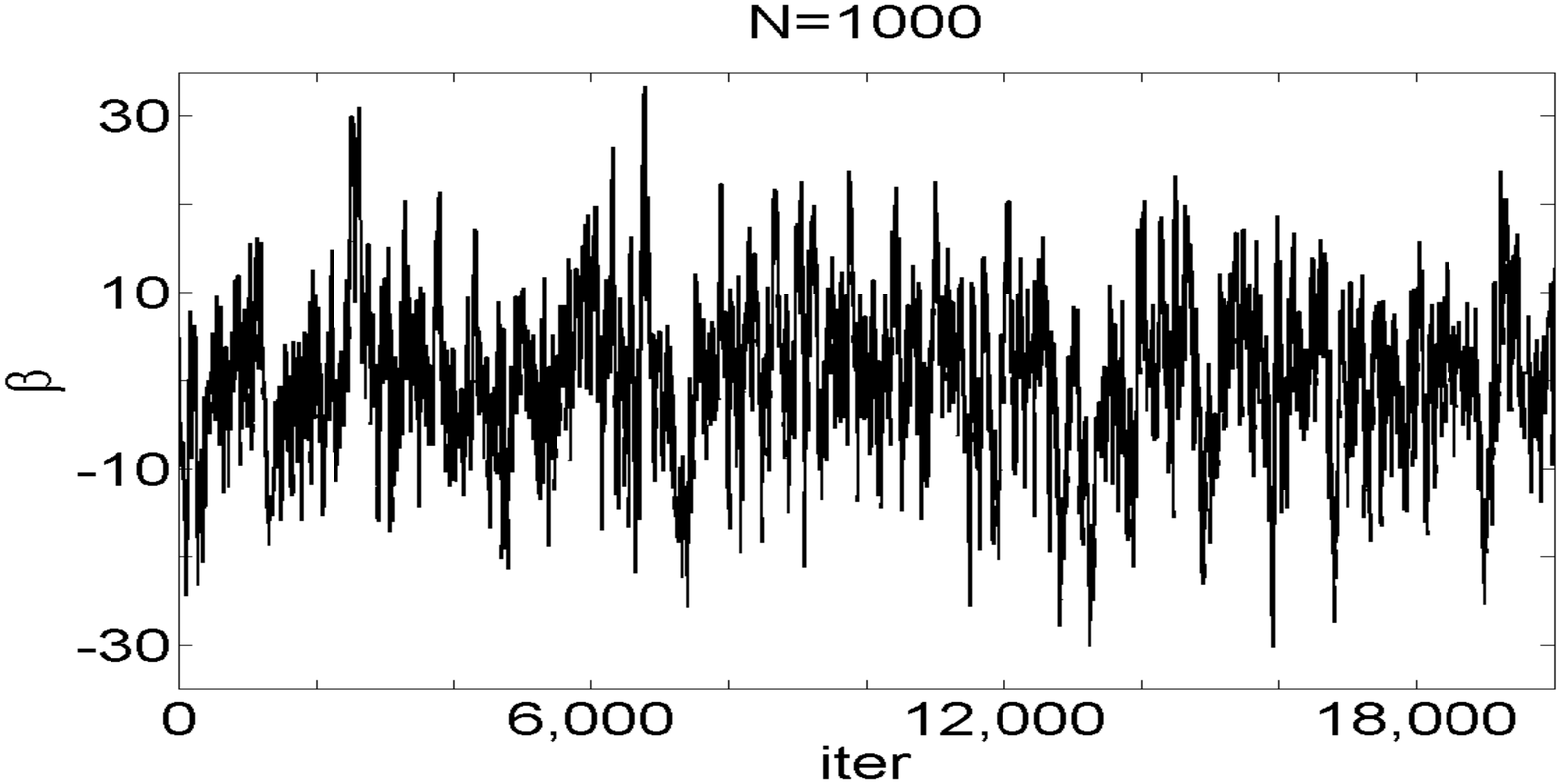}}
\end{center}
\end{figure}

\begin{figure}[H]
\begin{center}
\scalebox{0.2}{\includegraphics{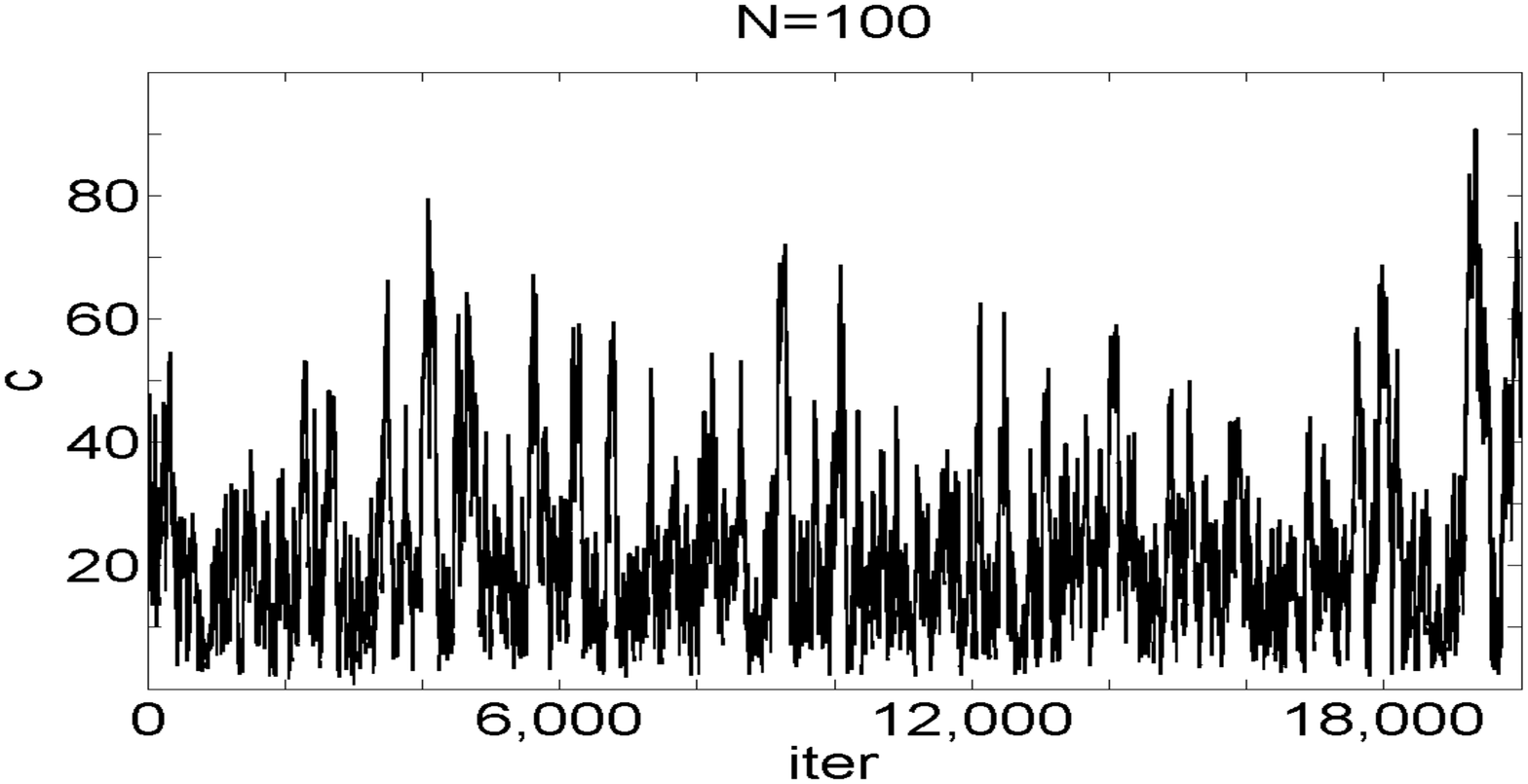}}\scalebox{0.2}{\includegraphics{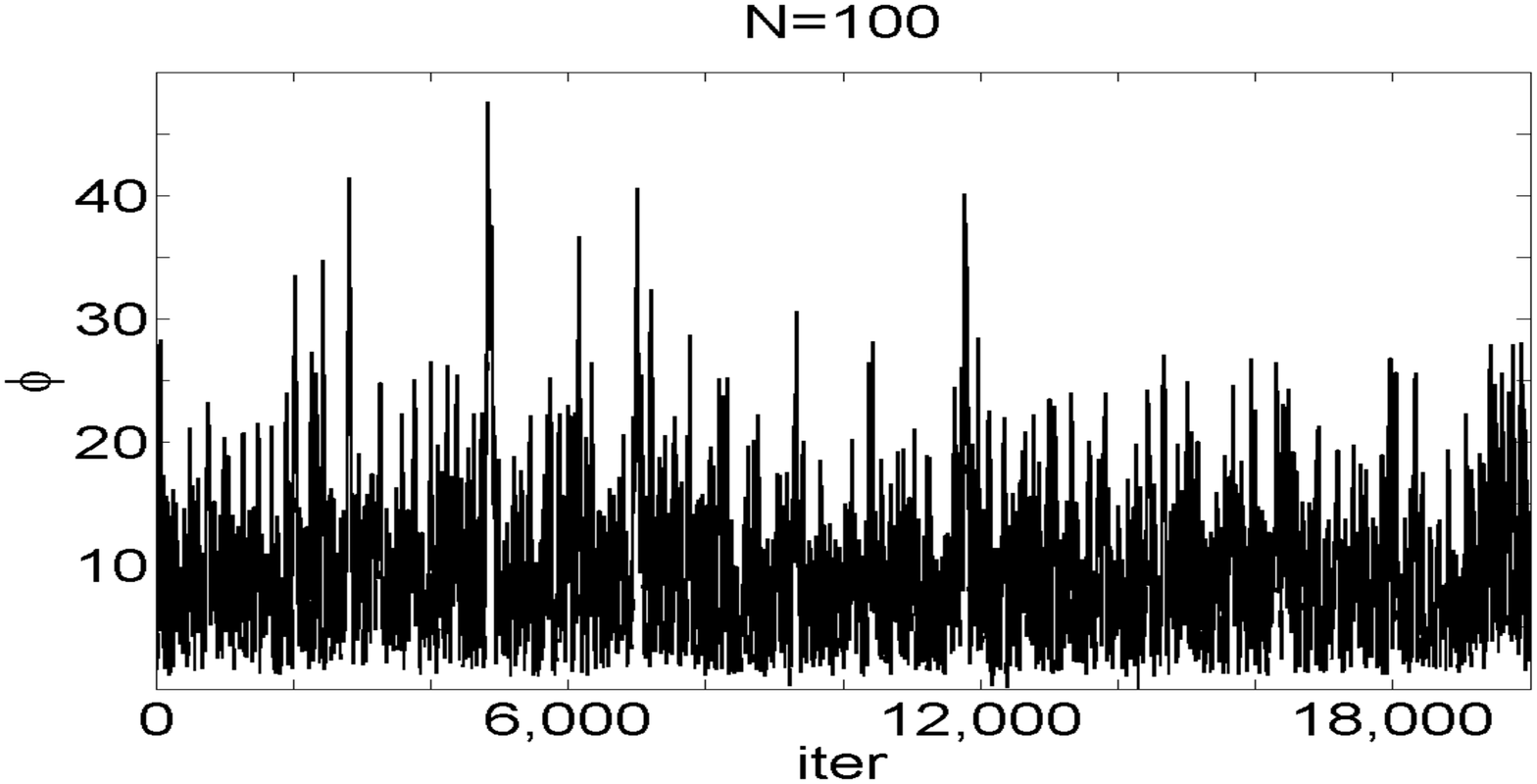}}\scalebox{0.2}{\includegraphics{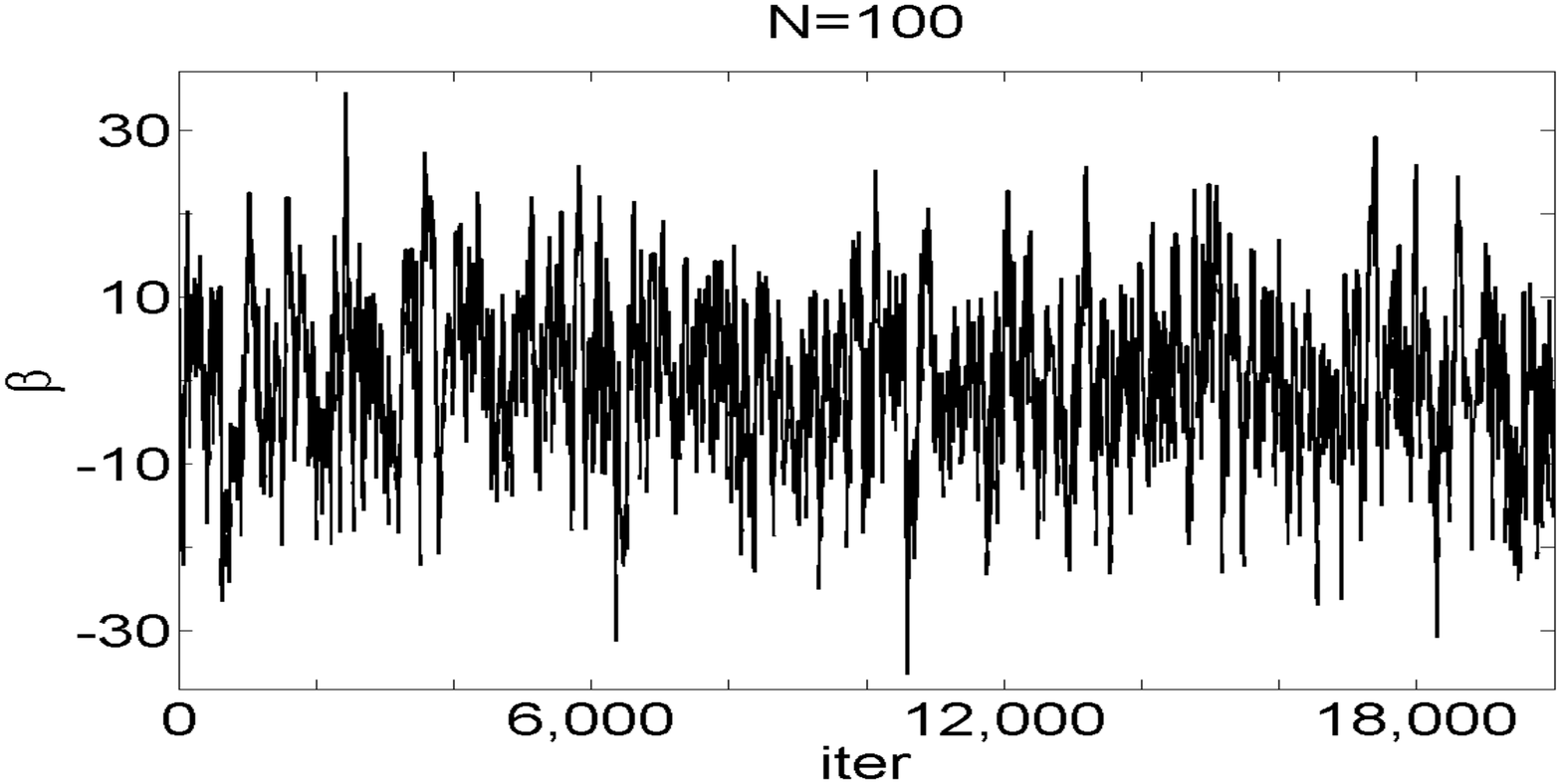}}
\end{center}
\end{figure}

\begin{figure}[H]
\begin{center}
\scalebox{0.2}{\includegraphics{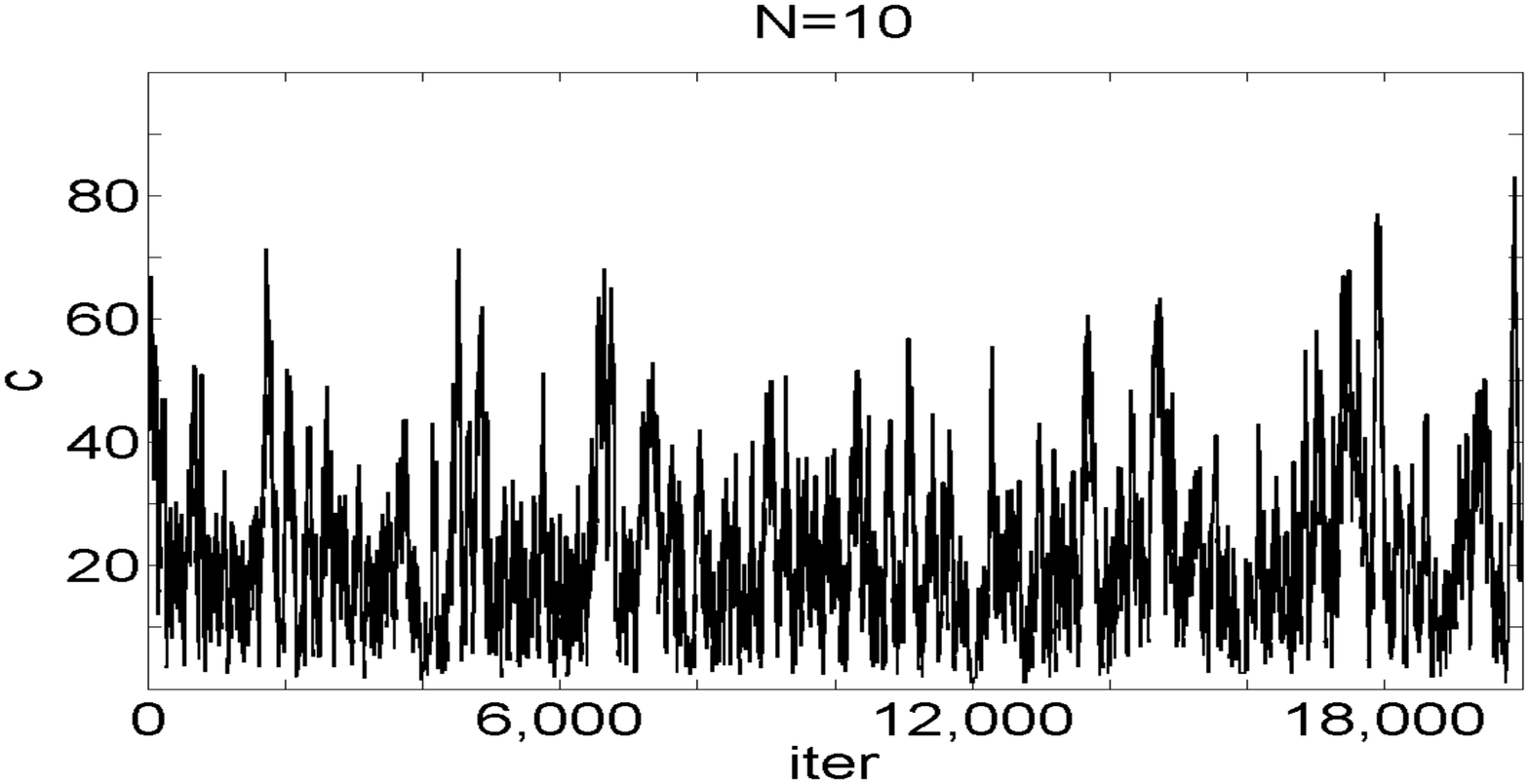}}\scalebox{0.2}{\includegraphics{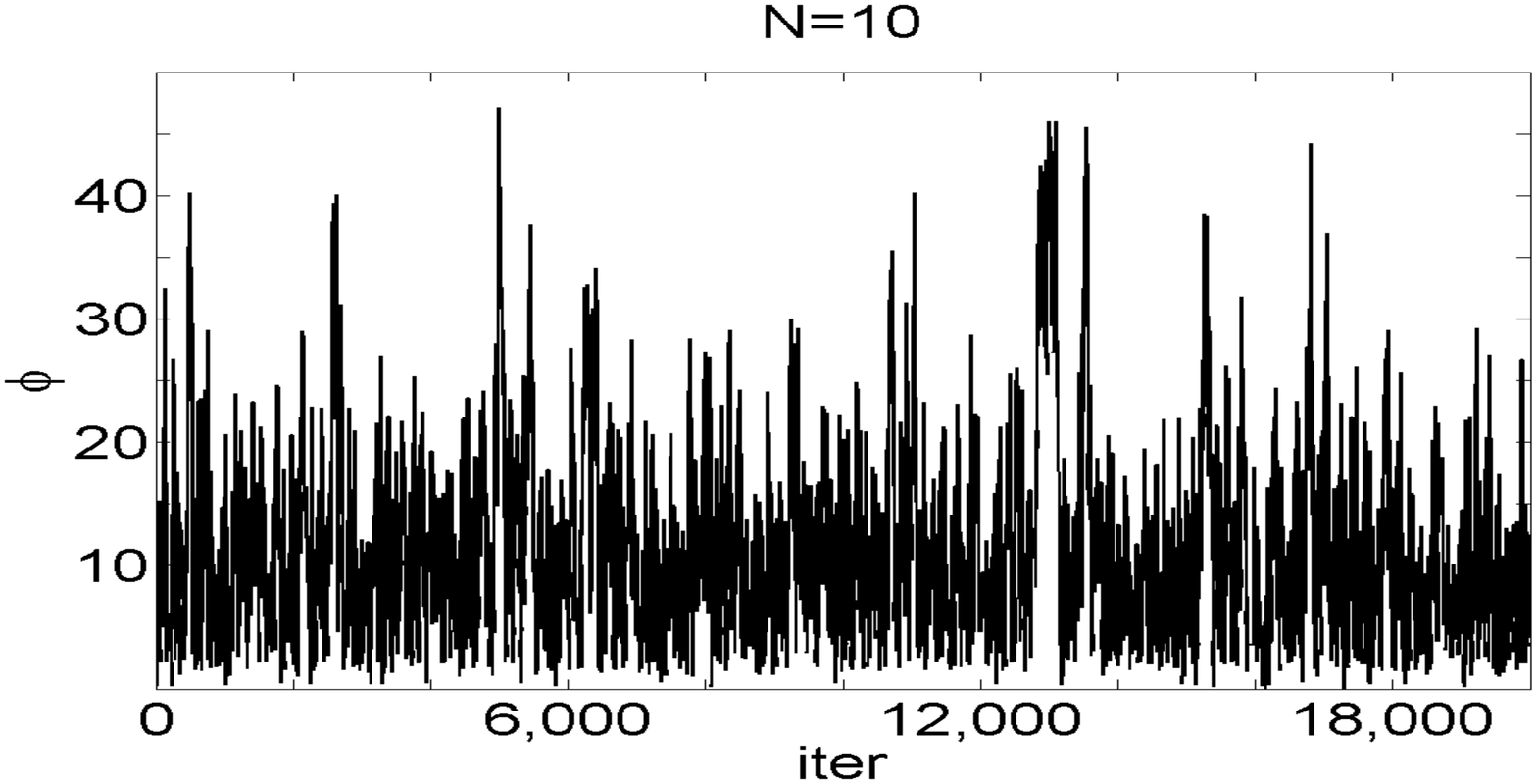}}\scalebox{0.2}{\includegraphics{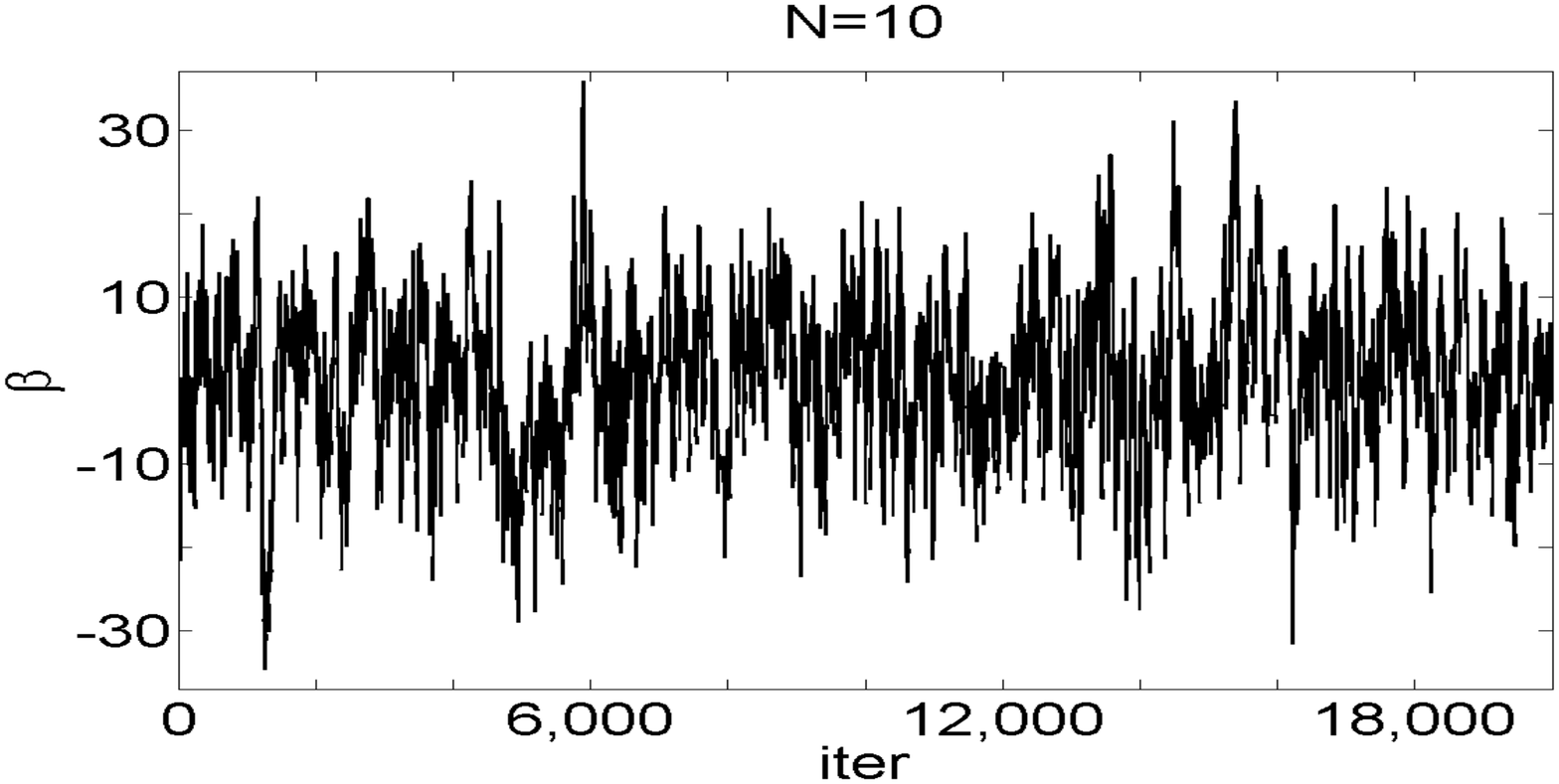}}
\caption{{\footnotesize Trace plot of each parameter across iterations for a PMMH algorithm using the SMC algorithm in Section \ref{sec:new_smc}. Each row displays the samples with different $N$}. Here $\xi_3=1.2$.}\label{fig:traceplot_col_new}
\end{center}
\end{figure}

\section{Summary}\label{sec:summ}

In this article we have investigated the alive particle filter; we developed and analyzed new particle estimates and derived new and principled MCMC algorithms.
There are several extensions to the work in this article. Firstly, we have presented and analyzed the most standard particle filter; one can investigate more
intricate filters commensurate with the current state of the art. Secondly, our theoretical results appear to hold under much weaker conditions than adopted (Section \ref{sec:linear_gaussian}); one
could extend the results in this direction. Thirdly, we have presented the most basic PMCMC algorithm; one can extend to particle Gibbs methods and beyond.
Finally, one can also use the SMC theory in this article to interact with that of MCMC theory to investigate the performance of our PMCMC procedures.

\subsection*{Acknowledgements}
The first author was supported by an MOE Singpore grant.
We thank Gareth Peters for useful conversations on this work.

\appendix

\section{Technical Results for the Predictor}\label{app:tech_pred}

The main result of this Section is below. Note we use the convention $\Phi_1(\eta_0^{T_0})(\varphi)=\eta_1(\varphi)$ and recall $\mathscr{F}_n$ is the filtration
generated by the particle system up-to time $n$.

\begin{cor}\label{cor:cond_lp}
Assume ($M_1$). Then for any $p\in[1,4]$ there exists a $C_p<\infty$ such that for any $n\geq 1$, $N\geq 2$ and $\varphi\in\mathcal{B}_b(\mathsf{E}_n)$
$$
\mathbb{E}[|[\eta_n^{T_n} - \Phi_n(\eta_{n-1}^{T_{n-1}}) ](\varphi)|^p|\mathscr{F}_{n-1}]^{1/p} \leq \frac{C_p\|\varphi\|}{\sqrt{N-1}} \quad \mathbb{P}-a.s..
$$
\end{cor}

\begin{proof}
The case $n=1$ follows directly from Lemma \ref{lem:tech_lem_neg1} and ($M_1$), so we consider $n\geq 2$.
For $n\geq 2$,  by Lemma \ref{lem:tech_res} $\mathbb{E}[\eta_n^{T_n}(\varphi)|\mathscr{F}_{n-1}] = \Phi_n(\eta_{n-1}^{T_{n-1}})(\varphi)$, so conditional
upon $\mathscr{F}_{n-1}$ we are in the setting of Lemma \ref{lem:tech_lem_neg1}. By ($M_1$) we can verify that
$\Phi_n(\eta_{n-1}^{T_{n-1}})(\mathsf{B}_{n})\wedge\Phi_n(\eta_{n-1}^{T_{n-1}})(\mathsf{B}_{n}^c)\geq c$ for
some deterministic constant $1>c>0$ and hence application of Lemma \ref{lem:tech_lem_neg1} proves the result.
\end{proof}

\subsection{Additional Technical Results}

In the following section let $(\mathsf{F},\mathscr{F})$ be a measurable space and $X^1,X^2,\dots$ be i.i.d.~random variables on $\mathsf{F}$ associated to $\nu\in\mathcal{P}(\mathsf{F})$.
Let $B\in\mathscr{F}$ be such that
$$
\nu(B) \wedge \nu(B^c) \geq c
$$
for some $1>c>0$. Let $N\geq 2$ and define
$$
T:= \inf\{n\geq 1: \sum_{i=1}^n \mathbb{I}_B(X^i) \geq N\}.
$$
Note that $T$ is a negative Binomial random variable, with parameters $N$ and success probability $\nu(B)$.
We will consider some $\mathbb{L}_p-$properties of 
$$
\frac{1}{T-1}\sum_{i=1}^{T-1} \varphi(X^i)
$$
with $\varphi\in\mathcal{B}_b(\mathsf{F})$. Expectations are written as $\mathbb{E}$.
Note that one can follow the proof of Lemma \ref{lem:tech_res} to show that
$$
\mathbb{E}\bigg[\frac{1}{T-1}\sum_{i=1}^{T-1} \varphi(X^i)\bigg] = \nu(\varphi).
$$
We then have the following technical results which are used in the main text.

\begin{lem}\label{lem:tech_lem_neg1}
For any $p\in[1,4]$  there exist a $C_p<\infty$ such that for any, $N\geq 2$, and $\varphi\in\mathcal{B}_b(\mathsf{F})$ 
$$
\mathbb{E}\bigg[\Big| 
\frac{1}{T-1}\sum_{i=1}^{T-1} [\varphi(X^i)-\nu(\varphi)]\Big|^p\bigg]^{1/p} \leq \frac{C_p\|\varphi\|_{\infty}}{\sqrt{N-1}}.
$$
\end{lem}

\begin{proof}
Throughout $C_p$ is a finite positive constant (that only depends upon $p$) whose value may change from line to line.
WLOG we will assume that $\nu(\varphi)=0$. By the Minkowski inequality
\begin{eqnarray}
\mathbb{E}\bigg[\Big| 
\frac{1}{T-1}\sum_{i=1}^{T-1} \varphi(X^i)\Big|^p\bigg]^{1/p} & \leq &
\mathbb{E}\bigg[\Big| 
\frac{1}{T-1}\sum_{i=1}^{T-1} \varphi(X^i)-\frac{N-1}{T-1}\frac{\nu(\mathbb{I}_B\varphi)}{\nu(B)} - \frac{T-N}{T-1}\frac{\nu(\mathbb{I}_{B^c}\varphi)}{\nu(B^c)}
\Big|^p\bigg]^{1/p} + \nonumber\\ & &
\mathbb{E}\bigg[\Big|\frac{N-1}{T-1}\frac{\nu(\mathbb{I}_B\varphi)}{\nu(B)} + \frac{T-N}{T-1}\frac{\nu(\mathbb{I}_{B^c}\varphi)}{\nu(B^c)}\Big|^p\bigg]^{1/p}
\label{eq:minkow_app_neg1}.
\end{eqnarray}
Lemma \ref{lem:tech_lem_neg2} will control the second term on the R.H.S.~so we focus on the first term on the R.H.S.. 

We have
$$
\mathbb{E}\bigg[\Big|\frac{1}{T-1}\sum_{i=1}^{T-1} \varphi(X^i)-\frac{N-1}{T-1}\frac{\nu(\mathbb{I}_B\varphi)}{\nu(B)} + \frac{T-N}{T-1}\frac{\nu(\mathbb{I}_{B^c}\varphi)}{\nu(B^c)}
\Big|^p\bigg]^{1/p} = 
$$
$$
\mathbb{E}\bigg[\Big|\frac{1}{T-1}\sum_{i=1}^{T-1} \mathbb{I}_B(X^i)\Big[\varphi(X^i)-\frac{\nu(\mathbb{I}_B\varphi)}{\nu(B)}\Big] + 
\frac{1}{T-1}\sum_{i=1}^{T-1} \mathbb{I}_{B^c}(X^i)\Big[\varphi(X^i)-\frac{\nu(\mathbb{I}_{B^c}\varphi)}{\nu(B^c)}\Big]
\Big|^p\bigg]^{1/p}.
$$
Now conditioning upon $T$ (so that $N-1$ samples lie in $B$ and $T-N$ lie in $B^c$ and we subtract the conditional expectations of the (conditionally) independent random variables)
 and applying an appropriately modified version of the M-Z inequality (e.g.~Chapter 7 of Del Moral (2004)) we have
$$
\mathbb{E}\bigg[\Big|\frac{1}{T-1}\sum_{i=1}^{T-1} \varphi(X^i)-\frac{N-1}{T-1}\frac{\nu(\mathbb{I}_B\varphi)}{\nu(B)} + \frac{T-N}{T-1}\frac{\nu(\mathbb{I}_{B^c}\varphi)}{\nu(B^c)}
\Big|^p\bigg]^{1/p} \leq
$$
$$
C_p \mathbb{E}\bigg[
\frac{1}{(T-1)^{p/2+1}}\Big(
(T-1)(\|\varphi\|_{\infty}(1+\frac{1}{\nu(B)}))^p +
(T-N) (\|\varphi\|_{\infty}(1+\frac{1}{\nu(B^c)}))^p
\Big)
\bigg]^{1/p}
$$
where we are using the conditional distribution of $X^1,\dots,X^{T-1}$, given $T$. Setting $\bar{c} = 
(\|\varphi\|_{\infty}(1+\frac{1}{\nu(B)}))^p\vee (\|\varphi\|_{\infty}(1+\frac{1}{\nu(B^c)}))^p$ we
have that
$$
\mathbb{E}\bigg[\Big|\frac{1}{T-1}\sum_{i=1}^{T-1} \varphi(X^i)-\frac{N-1}{T-1}\frac{\nu(\mathbb{I}_B\varphi)}{\nu(B)} + \frac{T-N}{T-1}\frac{\nu(\mathbb{I}_{B^c}\varphi)}{\nu(B^c)}
\Big|^p\bigg]^{1/p} \leq
C_p \bar{c}^{1/p}\mathbb{E}\bigg[\frac{1}{(T-1)^{p/2}}\Big(1+\frac{T-N}{T-1}\Big)\bigg]^{1/p}
$$
and then noting  $1/(T-1)\leq 1/(N-1)$, $(1+\frac{T-N}{T-1})\leq 2$ and using  $\nu(B) \wedge \nu(B^c) \geq c$ for dealing with $\bar{c}$ we have proved that
$$
\mathbb{E}\bigg[\Big|\frac{1}{T-1}\sum_{i=1}^{T-1} \varphi(X^i)-\frac{N-1}{T-1}\frac{\nu(\mathbb{I}_B\varphi)}{\nu(B)} + \frac{T-N}{T-1}\frac{\nu(\mathbb{I}_{B^c}\varphi)}{\nu(B^c)}
\Big|^p\bigg]^{1/p} \leq \frac{C_p\|\varphi\|_{\infty}}{\sqrt{N-1}}.
$$
Returning to \eqref{eq:minkow_app_neg1} and using Lemma \ref{lem:tech_lem_neg2}, along with above result allows us to complete the proof.
\end{proof}

\begin{lem}\label{lem:tech_lem_neg2}
For any $p\in[1,4]$  there exist a $C_p<\infty$ such that for any, $N\geq 2$, and $\varphi\in\mathcal{B}_b(\mathsf{F})$ 
$$
\mathbb{E}\bigg[\Big|\frac{N-1}{T-1}\frac{\nu(\mathbb{I}_B\varphi)}{\nu(B)} + \frac{T-N}{T-1}\frac{\nu(\mathbb{I}_{B^c}\varphi)}{\nu(B^c)} -\nu(\varphi)\Big|^p\bigg]^{1/p} \leq 
\frac{C_p\|\varphi\|_{\infty}}{\sqrt{N}}.
$$
\end{lem}

\begin{proof}
Throughout $C_p$ is a finite positive constant (that only depends upon $p$) whose value may change from line to line.
WLOG we will assume that $\nu(\varphi)=0$, so that $\nu(\mathbb{I}_B\varphi) = -\nu(\mathbb{I}_{B^c}\varphi)$.
Then we have that
\begin{eqnarray*}
\zeta(N,T) & := & \frac{N-1}{T-1}\frac{\nu(\mathbb{I}_B\varphi)}{\nu(B)} + \frac{T-N}{T-1}\frac{\nu(\mathbb{I}_{B^c}\varphi)}{\nu(B^c)}\\
& = & \frac{\nu(\mathbb{I}_B\varphi)N}{\nu(B)(1-\nu(B))}\Big[\frac{N-1 + \nu(B) - T\nu(B)}{N(T-1)}\Big].
\end{eqnarray*}
Now, via Minkowski
\begin{eqnarray}
\mathbb{E}[|\zeta(N,T) |^p]^{1/p} & \leq & \frac{|\nu(\mathbb{I}_B\varphi)|N}{\nu(B)(1-\nu(B))}
\Big\{\mathbb{E}\bigg[\Big(\frac{1}{(T-1)^p}\Big)\Big|1-\frac{T\nu(B)}{N}\Big|^p\bigg]^{1/p} +
\mathbb{E}\bigg[\Big|\frac{\nu(B)-1}{N(T-1)}\Big|^p\bigg]^{1/p}
\Big\}
\label{eq:minkow_app_neg}
\end{eqnarray}
As $1/(T-1)\leq 1/(N-1)$, we will focus on controlling the term
$$
\mathbb{E}\bigg[\Big|1-\frac{T\nu(B)}{N}\Big|^p\bigg]^{1/p}.
$$
If $Y^1,Y^2,\dots$ are independent $\mathcal{G}eo(\nu(B))$ random variables then
$$
\mathbb{E}\bigg[\Big|1-\frac{T\nu(B)}{N}\Big|^p\bigg]^{1/p} = \nu(B)\mathbb{E}\bigg[\Big|\frac{1}{N}\sum_{i=1}^N Y^i - \frac{1}{\nu(B)}\Big|^p\bigg]^{1/p}
$$
and applying an appropriately modified version of the M-Z inequality (e.g.~Chapter 7 of Del Moral (2004)) we have
$$
\mathbb{E}\bigg[\Big|1-\frac{T\nu(B)}{N}\Big|^p\bigg]^{1/p} \leq \frac{C_p}{\sqrt{N}} \Big[\frac{(1-\nu(B)(1-\nu(B) + \nu(B)^2))}{\nu(B)^4}\Big]^{1/p}
$$
where we have used the fourth central moment of a Geometric random variable and $C_p$ is a constant that only depends upon $p$ (that is independent of $\nu$ or $B$).
Returning to \eqref{eq:minkow_app_neg} and noting $\nu(B) \wedge \nu(B^c) \geq c$, we have shown that
$$
\mathbb{E}[|\zeta(N,T) |^p]^{1/p} \leq C\|\varphi\|_{\infty}\frac{N}{N-1}[\frac{C_p}{\sqrt{N}} + \frac{1}{N}]
$$
from which we can easily conclude.
\end{proof}

\section{Technical Results for the CLT}\label{app:tech_clt}

Recall convergence in probability is written $\rightarrow_{\mathbb{P}}$ and $N$ is going to $\infty$. In addition,
that the convention $\Phi_1(\eta_0^{T_0})(\varphi)=\eta_1(\varphi)$ is used and again recall $\mathscr{F}_n$ is the filtration
generated by the particle system up-to time $n$.

\begin{lem}\label{lem:clt_main_random_to_deterministic}
Assume ($M_1$). Then for any $n\geq 1$, $\varphi\in\mathcal{B}_b(\mathsf{E}_n)$ we have:
$$
\sqrt{T_n-1} [\eta_{n}^{T_n} -\Phi_n(\eta_{n-1}^{T_{n-1}}) ](\varphi) -  
\sqrt{(N-1)\eta_n(G_n)} [\eta_{n}^{N-1} -\Phi_n(\eta_{n-1}^{T_{n-1}}) ](\varphi) \rightarrow_{\mathbb{P}} 0.
$$
\end{lem}

\begin{proof}
We give the proof for any $n\geq 2$; the case $n=1$ follows a similar proof with only notational modifications.
Throughout the proof $0<C<\infty$ is a deterministic constant independent of $n$ and $N$ whose value may change from line to line.
Our proof follows a similar construction to that found in pp.~369 of Billingsley (1995). To that end, we have
$$
\sqrt{T_n-1} [\eta_{n}^{T_n} -\Phi_n(\eta_{n-1}^{T_{n-1}}) ](\varphi) -  
\sqrt{(N-1)\eta_n(G_n)} [\eta_{n}^{N-1} -\Phi_n(\eta_{n-1}^{T_{n-1}}) ](\varphi) = 
$$
\begin{equation}
\sqrt{T_n-1} [\eta_{n}^{T_n} -\Phi_n(\eta_{n-1}^{T_{n-1}}) ](\varphi)
- (T_n-1)\sqrt{\frac{\eta_n(G_n)}{(N-1)}}
[\eta_{n}^{T_n} -\Phi_n(\eta_{n-1}^{T_{n-1}}) ](\varphi)
+
\label{eq:main_tech_clt_lem1}
\end{equation}
\begin{equation}
(T_n-1)\sqrt{\frac{\eta_n(G_n)}{(N-1)}}
[\eta_{n}^{T_n} -\Phi_n(\eta_{n-1}^{T_{n-1}}) ](\varphi) - 
\sqrt{(N-1)\eta_n(G_n)} [\eta_{n}^{N-1} -\Phi_n(\eta_{n-1}^{T_{n-1}}) ](\varphi).
\label{eq:main_tech_clt_lem2}
\end{equation}
In Lemma \ref{lem:clt_first_conv}, we have shown that \eqref{eq:main_tech_clt_lem1} converges in probability to zero; hence, we focus upon
\eqref{eq:main_tech_clt_lem2}.

To shorten the subsequent notations, we set
\begin{eqnarray*}
S_n^{T_n}(\varphi) & = & (T_n-1)[\eta_{n}^{T_n} -\Phi_n(\eta_{n-1}^{T_{n-1}}) ](\varphi)\\
S_n^{N-1}(\varphi) & = & (N-1)[\eta_{n}^{N-1} -\Phi_n(\eta_{n-1}^{T_{n-1}}) ](\varphi).
\end{eqnarray*}
Then let $1>\varepsilon>0$ be given, and consider:
$$
\mathbb{P}\bigg(\Big|S_n^{T_n}(\varphi)-S_n^{N-1}(\varphi)\Big|\geq \varepsilon\sqrt{\frac{N-1}{\eta_n(G_n)}}\bigg)
$$ 
$$
\leq
\mathbb{P}\bigg(\Big|T_n-1-\frac{N-1}{\eta_n(G_n)}\Big|\geq \varepsilon^3(N-1)\bigg) + 
\mathbb{P}\bigg(\max_{|k-\frac{N-1}{\eta_n(G_n)}|\leq\varepsilon^3(N-1)}\Big|S_n^{k}(\varphi)-S_n^{N-1}(\varphi)\Big|\geq \varepsilon(N-1)^{1/2}\bigg)
$$
$$
\leq \mathbb{P}\bigg(\Big|T_n-1-\frac{N-1}{\eta_n(G_n)}\Big|\geq \varepsilon^3(N-1)\bigg) + 
2\mathbb{P}\Big(\max_{1\leq k\leq\varepsilon^3(N-1)}|S_n^{k}(\varphi)|\geq\varepsilon(N-1)^{1/2}\Big).
$$
Now for the latter probability, one can condition upon $\mathscr{F}_{n-1}$ and apply Kolmogorov's inequality noting that the conditional
variance of $\varphi(x) - \Phi_n(\eta_{n-1}^{T_{n-1}})(\varphi)$ is deterministically upper-bounded by $\textrm{Osc}(\varphi)^2$. Hence, we have that 
$$
\mathbb{P}\bigg(\Big|S_n^{T_n}(\varphi)-S_n^{N-1}(\varphi)\Big|\geq \varepsilon\sqrt{\frac{N-1}{\eta_n(G_n)}}\bigg)
$$
$$
\leq 
\mathbb{P}\bigg(\Big|T_n-1-\frac{N-1}{\eta_n(G_n)}\Big|\geq \varepsilon^3(N-1)\bigg) + 
2C\varepsilon.
$$
Noting Lemma \ref{lem:conv_t_n} and that we can make $\varepsilon$ arbitrarily small, the proof is completed.
\end{proof}

\begin{lem}\label{lem:clt_first_conv}
Assume ($M_1$). Then for any $n\geq 1$, $\varphi\in\mathcal{B}_b(\mathsf{E}_n)$ we have:
$$
\sqrt{T_n-1} [\eta_{n}^{T_n} -\Phi_n(\eta_{n-1}^{T_{n-1}}) ](\varphi)
- (T_n-1)\sqrt{\frac{\eta_n(G_n)}{(N-1)}}
[\eta_{n}^{T_n} -\Phi_n(\eta_{n-1}^{T_{n-1}}) ](\varphi)
\rightarrow_{\mathbb{P}} 0.
$$
\end{lem}

\begin{proof}
We give the proof for any $n\geq 2$; the case $n=1$ follows a similar proof with only notational modifications.
Throughout the proof $0<C<\infty$ is a deterministic constant independent of $n$ and $N$ whose value may change from line to line.
We will prove that 
$$
\sqrt{T_n-1} [\eta_{n}^{T_n} -\Phi_n(\eta_{n-1}^{T_{n-1}}) ](\varphi)
- (T_n-1)\sqrt{\frac{\eta_n(G_n)}{(N-1)}}
[\eta_{n}^{T_n} -\Phi_n(\eta_{n-1}^{T_{n-1}}) ](\varphi)
$$
 will go to zero in $\mathbb{L}_1$. To that end, we rewrite this expression as
$$
A(N) := [\eta_{n}^{T_n} -\Phi_n(\eta_{n-1}^{T_{n-1}}) ](\varphi) \bigg[\frac{T_n-1}{\sqrt{N-1}}\sqrt{\eta_n(G_n)}\bigg] \bigg[\sqrt{\frac{(N-1)}{(T_n-1)\eta_n(G_n)}}-1\bigg].
$$
To simplify the subsequent notations, we define
$$
B(N) := \bigg[\sqrt{\frac{(N-1)}{(T_n-1)\eta_n(G_n)}}-1\bigg].
$$
Then a simple application of H\"older's inequality gives
$$
\mathbb{E}[|A(N)|] \leq \mathbb{E}\bigg[\Big|[\eta_{n}^{T_n} -\Phi_n(\eta_{n-1}^{T_{n-1}}) ](\varphi) \bigg[\frac{T_n-1}{\sqrt{N-1}}\sqrt{\eta_n(G_n)}\bigg]\Big|^{3/2}\bigg]^{2/3}\mathbb{E}[|B(N)|^3]^{1/3}.
$$
We will show that:
\begin{enumerate}
\item{$\mathbb{E}\bigg[\Big|[\eta_{n}^{T_n} -\Phi_n(\eta_{n-1}^{T_{n-1}}) ](\varphi) \bigg[\frac{T_n-1}{\sqrt{N-1}}\sqrt{\eta_n(G_n)}\bigg]\Big|^{3/2}\bigg]^{2/3}$ is upper-bounded
by a finite deterministic constant $C$ that is independent of $N$.}
\item{$\lim_{N\rightarrow\infty}\mathbb{E}[|B(N)|^3]^{1/3}=0$.}
\end{enumerate}
This will conclude the proof.

\textbf{Proof of 1}.  We have, by another application of H\"older, that 
$$
\mathbb{E}\bigg[\Big|[\eta_{n}^{T_n} -\Phi_n(\eta_{n-1}^{T_{n-1}}) ](\varphi) \bigg[\frac{T_n-1}{\sqrt{N-1}}\sqrt{\eta_n(G_n)}\bigg]\Big|^{3/2}\bigg]^{2/3} \leq
\mathbb{E}[|[\eta_{n}^{T_n} -\Phi_n(\eta_{n-1}^{T_{n-1}}) ](\varphi)|^{3}]^{1/3}\mathbb{E}\bigg[\Big|\bigg[\frac{T_n-1}{\sqrt{N-1}}\sqrt{\eta_n(G_n)}\Big|^{3}\bigg]^{1/3}
$$
Application of Corollary \ref{cor:cond_lp} gives
\begin{equation}
\mathbb{E}\bigg[\Big|[\eta_{n}^{T_n} -\Phi_n(\eta_{n-1}^{T_{n-1}}) ](\varphi) \bigg[\frac{T_n-1}{\sqrt{N-1}}\sqrt{\eta_n(G_n)}\bigg]\Big|^{3/2}\bigg]^{2/3} \leq
\frac{C_3\|\varphi\|_{\infty}}{\sqrt{N-1}}\mathbb{E}\bigg[\Big|\bigg[\frac{T_n-1}{\sqrt{N-1}}\sqrt{\eta_n(G_n)}\Big|^{3}\bigg]^{1/3}.
\label{eq:tech_lem_clt_st_det1}
\end{equation}
Now turning to the expectation on the R.H.S.~of the inequality, we have
$$
\mathbb{E}\bigg[\Big|\bigg[\frac{T_n-1}{\sqrt{N-1}}\sqrt{\eta_n(G_n)}\Big|^{3}\bigg] = 
\frac{\eta_n(G_n)^{3/2} }{(N-1)^{3/2}}\mathbb{E}[T_n^3 - 3 T_n^2 + 3T_n -1]
$$
Using the fact that, via ($M_1$) 
\begin{equation}
\Phi_n(\eta_n^{T_{n-1}})(\mathsf{B}_{n})\wedge \Phi_n(\eta_n^{T_{n-1}})(\mathsf{B}_{n}^c)\geq c\label{eq:prob_cont}
\end{equation}
for deterministic $c$
and standard properties on raw moments of negative binomial random variables, it follows that 
$$
\mathbb{E}[T_n^3 - 3 T_n^2 + 3T_n -1] \leq C N^3.
$$
Thus, we can show that
$$
\mathbb{E}\bigg[\Big|\bigg[\frac{T_n-1}{\sqrt{N-1}}\sqrt{\eta_n(G_n)}\Big|^{3}\bigg]^{1/3} \leq C\eta_n(G_n)^{1/2}N^{1/2}.
$$
Returning to \eqref{eq:tech_lem_clt_st_det1}, we have shown that 
$$
\mathbb{E}\bigg[\Big|[\eta_{n}^{T_n} -\Phi_n(\eta_{n-1}^{T_{n-1}}) ](\varphi) \bigg[\frac{T_n-1}{\sqrt{N-1}}\sqrt{\eta_n(G_n)}\bigg]\Big|^{3/2}\bigg]^{2/3} \leq
C\eta_n(G_n)^{1/2}
$$
which completes the proof of 1.

\textbf{Proof of 2}. By Lemma \ref{lem:conv_t_n} and the continuous mapping theorem, we have that $B(N)\rightarrow_{\mathbb{P}} 0$. Thus if we can show that for some $\delta>0$,
$\sup_{N\geq 1}\mathbb{E}[|B(N)|^{3(1+\delta)}] <+\infty$ this will allow us to conclude. For simplicity of calculation, we set $\delta = 1/3$. Then, using the fact that $1/(T_n-1)\leq 1/(N-1)$ we have
$$
\mathbb{E}[|B(N)|^4] \leq \mathbb{E}\bigg[\Big|1-\sqrt{\frac{\eta_n(G_n))(T_n-1)}{N-1}}\Big|^4\bigg].
$$
On expanding the brackets and removing the negative terms, the expectation on the R.H.S. is upper-bounded by
$$
1 + \frac{6\eta_n(G_n)}{N-1}\mathbb{E}[T_n-1] + \frac{\eta_n(G_n)^2}{(N-1)^2}\mathbb{E}[T_n^2 -2T_n +1].
$$
Using the conditional negative binomial property of $T_n$ this expression is equal to
$$
1 + \frac{6\eta_n(G_n)}{N-1} \Big\{\mathbb{E}\Big[\frac{N}{\Phi_n(\eta_{n-1}^{T_{n-1}})(G_n)}\Big] -1\Big\} + \frac{\eta_n(G_n)^2}{(N-1)^2}
\Big\{\mathbb{E}\Big[\frac{N(1-\Phi_n(\eta_{n-1}^{T_{n-1}})(G_n))}{\Phi_n(\eta_{n-1}^{T_{n-1}})(G_n)^2}  + \frac{N^2}{\Phi_n(\eta_{n-1}^{T_{n-1}})(G_n)^2} - \frac{2N}{\Phi_n(\eta_{n-1}^{T_{n-1}})(G_n)}+1\Big]
\Big\}.
$$
Applying \eqref{eq:prob_cont} we easily show that this latter expression is, uniformly in $N$, upper-bounded by a constant $C$. That is, we have shown
that $\sup_{N\geq 1}\mathbb{E}[|B(N)|^{3(1+\delta)}] <+\infty$, which completes the proof of 2. This completes the proof.
\end{proof}

\begin{lem}\label{lem:conv_t_n}
Assume ($M_1$). Then for any $n\geq 1$, we have:
$$
\frac{T_n}{N} -\frac{1}{\eta_n(G_n)} \rightarrow_{\mathbb{P}} 0.
$$
\end{lem}

\begin{proof}
We give the proof for any $n\geq 2$; the case $n=1$ follows a similar proof with only notational modifications.
In Theorem \ref{theo:time_uniform}, we have proved that $\Phi_n(\eta_{n-1}^{T_{n-1}})(\varphi)$ converges almost surely to 
$\eta_n(\varphi)$ for $\varphi\in\mathcal{B}_b(\mathsf{E})$. Thus we consider
$$
\mathbb{E}\bigg[\Big(\frac{T_n}{N} - \frac{1}{\Phi_n(\eta_{n-1}^{T_{n-1}})(G_n)} \Big)^2\bigg].
$$
Now conditionally upon $\mathscr{F}_{n-1}$, $T_n$ is a negative binomial random variable with success probability $\Phi_n(\eta_{n-1}^{T_{n-1}})(G_n)$, so writing
$X_{i,n}(N)$ as (conditionally) independent geometric random variables with the same success probability, we have
$$
\mathbb{E}\bigg[\Big(\frac{T_n}{N} - \frac{1}{\Phi_n(\eta_{n-1}^{T_{n-1}})(G_n)} \Big)^2\bigg] = 
\mathbb{E}\bigg[\mathbb{E}\bigg[\Big(\frac{1}{N}\sum_{i=1}^N [X_{i,n}(N) - \frac{1}{\Phi_n(\eta_{n-1}^{T_{n-1}})(G_n)}] \Big)^2\bigg|\mathscr{F}_{n-1}\bigg]\bigg].
$$
Applying the conditional version of the M-Z inequality on the R.H.S.~of the inequality, we have the upper-bound:
$$
\frac{C}{N}\mathbb{E}\bigg[\frac{1-\Phi_n(\eta_{n-1}^{T_{n-1}})(G_n)}{\Phi_n(\eta_{n-1}^{T_{n-1}})(G_n)}\bigg]
$$
recalling that $\Phi_n(\eta_{n-1}^{T_{n-1}})(G_n)\geq c$ for some deterministic constant $c$ we conclude that
$$
\mathbb{E}\bigg[\Big(\frac{T_n}{N} - \frac{1}{\Phi_n(\eta_{n-1}^{T_{n-1}})(G_n)} \Big)^2\bigg] \leq \frac{C}{N}.
$$
The proof is completed on recalling that $\Phi_n(\eta_{n-1}^{T_{n-1}})(G_n)$ converges almost surely to 
$\eta_n(G_n)$.
\end{proof}

\section{Technical Results for the Normalizing Constant}\label{app:tech_nc}

\begin{lem}\label{lem:tech_res}
We have for any $n\geq 1$, $N\geq 2$ and $\varphi\in\mathbb{B}_b(\mathsf{E}_n)$, that
$$
\mathbb{E}[\eta_n^{T_n}(\varphi)|\mathscr{F}_{n-1}] = \Phi_n(\eta_{n-1}^{T_{n-1}})(\varphi)
$$
where $\Phi_1(\eta_{-1}^{T_{-1}})(\varphi)=M_1(\varphi)$.
\end{lem}

\begin{proof}
We have, for any $n\geq 1$, $N\geq 2$ that $T_n|\mathscr{F}_{n-1}$ is a Negative Binomial random variable with parameters $N-1$ and success probability 
$\Phi_n(\eta_{n-1}^{T_{n-1}})(\mathsf{B}_{n})=\Phi_n(\eta_{n-1}^{T_{n-1}})(G_{n})$ and note that from Neuts \& Zacks (1967) and Zacks (1980)
\begin{equation}
\mathbb{E}\Big[\frac{N-1}{T_{n}-1}\Big|\mathscr{F}_{n-1}\Big] = \Phi_n(\eta_{n-1}^{T_{n-1}})(\mathsf{B}_{n})\label{eq:neg_binom_ineq}.
\end{equation}

Now, 
\begin{eqnarray*}
\mathbb{E}[\eta_n^{T_n}(\varphi)|\mathscr{F}_{n-1}] & = & 
\mathbb{E}\Big[\frac{1}{T_n-1}\sum_{i=1}^{T_n-1} \varphi(X_n^i)\Big|\mathscr{F}_{n-1}\Big]\\
& = & \mathbb{E}\Big[\Big(\frac{1}{T_n-1}\Big)\Big\{(N-1)\frac{\Phi_n(\eta_{n-1}^{T_{n-1}})(\varphi\mathbb{I}_{\mathsf{B}_{n}})}{\Phi_n(\eta_{n-1}^{T_{n-1}})(\mathsf{B}_{n})}
+
(T_n-N)\frac{\Phi_n(\eta_{n-1}^{T_{n-1}})(\varphi\mathbb{I}_{\mathsf{B}_{n}^c})}{\Phi_n(\eta_{n-1}^{T_{n-1}})(\mathsf{B}_{n}^c)}
\Big\}\Big|\mathscr{F}_{n-1}\Big]\\
& = & 
\mathbb{E}\Big[\frac{N-1}{T_n-1}\frac{\Phi_n(\eta_{n-1}^{T_{n-1}})(\varphi\mathbb{I}_{\mathsf{B}_{n}})}{\Phi_n(\eta_{n-1}^{T_{n-1}})(\mathsf{B}_{n})} +
\Big(1-\frac{N-1}{T_n-1}\Big)\frac{\Phi_n(\eta_{n-1}^{T_{n-1}})(\varphi\mathbb{I}_{\mathsf{B}_{n}^c})}{\Phi_n(\eta_{n-1}^{T_{n-1}})(\mathsf{B}_{n}^c)}
\Big\}\Big|\mathscr{F}_{n-1}\Big]
\end{eqnarray*}
where we have used the fact that there are $N-1$ particles that are `alive' and $T_n-N$ that will die and used the conditional distribution of the samples given $T_n$. Now by \eqref{eq:neg_binom_ineq}, it follows then that
$$
\mathbb{E}[\eta_n^{T_n}(\varphi)|\mathscr{F}_{n-1}] = 
\Phi_n(\eta_{n-1}^{T_{n-1}})(\varphi\mathbb{I}_{\mathsf{B}_{n}}) +
\Phi_n(\eta_{n-1}^{T_{n-1}})(\varphi\mathbb{I}_{\mathsf{B}_{n}^c}) = 
\Phi_n(\eta_{n-1}^{T_{n-1}})(\varphi)
$$
which concludes the proof.
\end{proof}

\begin{lem}\label{lem:tech_lem}
We have for any $n\geq 2$, $N\geq 3$, $\varphi\in\mathbb{B}_b(\mathsf{E}_n^2)$:
$$
\mathbb{E}[(\eta_n^{T_n})^{\odot 2}(\varphi)|\mathscr{F}_{n-1}] = \Phi_n(\eta_{n-1}^{T_{n-1}})^{\otimes 2}(\varphi).
$$
\end{lem}

\begin{proof}
We have:
$$
\mathbb{E}[(\eta_n^{T_n})^{\odot 2}(\varphi)|\mathscr{F}_{n-1}] =
$$
\begin{equation}
\mathbb{E}\Big[\frac{(N-1)(N-2)}{(T_n-1)(T_n-2)}\Big|\mathscr{F}_{n-1}\Big] \frac{\Phi_n(\eta_{n-1}^{T_{n-1}})^{\otimes 2}(\mathbb{I}_{\mathsf{B}_{n}^2}\varphi)}
{\Phi_n(\eta_{n-1}^{T_{n-1}})^{\otimes 2}(\mathsf{B}_{n}^2)} +
\label{eq:two_succ}
\end{equation}
\begin{equation}
\mathbb{E}\Big[
2(N-1)\Big(\frac{1}{T_n-1} -
\frac{N-2}{(T_n-1)(T_n-2)}
\Big)
\Big|\mathscr{F}_{n-1}\Big]
\frac{\Phi_n(\eta_{n-1}^{T_{n-1}})^{\otimes 2}(\mathbb{I}_{\mathsf{B}_{n}\times \mathsf{B}_{n}^c}\varphi)}
{\Phi_n(\eta_{n-1}^{T_{n-1}})(\mathsf{B}_{n})\Phi_n(\eta_{n-1}^{T_{n-1}})(\mathsf{B}_{n}^c)} +
\label{eq:one_succ}
\end{equation}
\begin{equation}
\mathbb{E}\Big[
\Big(1-\frac{N-1}{T_{n}-1}-\frac{N-1}{T_{n}-2}+-\frac{(N-1)^2}{(T_{n}-1)(T_{n}-2)}\Big)
\Big|\mathscr{F}_{n-1}\Big]
\frac{\Phi_n(\eta_{n-1}^{T_{n-1}})^{\otimes 2}(\mathbb{I}_{(\mathsf{B}_{n}^c)^2}\varphi)}
{\Phi_n(\eta_{n-1}^{T_{n-1}})^{\otimes 2}((\mathsf{B}_{n}^c)^2)}.
\label{eq:no_succ}
\end{equation}
The three terms on the R.H.S.~arise due to the $(N-1)(N-2)$ different pairs of particles which land in $\mathsf{B}_{n}^2$ \eqref{eq:two_succ}, the
$2(N-1)(T_n-N)$ pairs of different particles which land in $\mathsf{B}_n\times \mathsf{B}_n^c$ \eqref{eq:one_succ} and the 
$(T_n-N)(T_n-N-1)$ different pairs of particles which land in $(\mathsf{B}_n^c)^2$ \eqref{eq:no_succ}; the factors of $\Phi_n(\eta_{n-1}^{T_{n-1}})$ arise
from the conditional distributions of the particles given $T_n$ (recalling that conditional on $\mathscr{F}_{n-1}$, $T_n$ is a negative binomial random variables parameters $N$
and $\Phi_n(\eta_{n-1}^{T_{n-1}})(B_{\epsilon}(y_n))$).

Now for \eqref{eq:two_succ}, we have from  Neuts \& Zacks (1967) and Zacks (1980) that
$$
\mathbb{E}\Big[\frac{(N-1)(N-2)}{(T_n-1)(T_n-2)}\Big|\mathscr{F}_{n-1}\Big] = \Phi_n(\eta_{n-1}^{T_{n-1}})^{\otimes 2}(\mathsf{B}_{n}^2).
$$ 
so that \eqref{eq:two_succ} becomes
$$
\Phi_n(\eta_{n-1}^{T_{n-1}})^{\otimes 2}(\mathbb{I}_{\mathsf{B}_{n}^2}\varphi).
$$
Recalling \eqref{eq:neg_binom_ineq} and using the above result, 
\eqref{eq:one_succ} becomes
$$
2\Phi_n(\eta_{n-1}^{T_{n-1}})^{\otimes 2}(\mathbb{I}_{\mathsf{B}_{n}\times \mathsf{B}_{n}^c}\varphi).
$$
Finally, noting that for any $t\neq 1,2$ $1/(t-2)=1/(t-1) + 1/[(t-1)(t-2)]$, and thus using the above results that
$$
\mathbb{E}\Big[\frac{N-1}{T_n-2}\Big|\mathscr{F}_{n-1}\Big] = \Phi_n(\eta_{n-1}^{T_{n-1}})^{\otimes 2}(\mathsf{B}_{n})
+\frac{\Phi_n(\eta_{n-1}^{T_{n-1}})^{\otimes 2}(\mathsf{B}_{n})^2}{N-1}
$$ 
it follows that \eqref{eq:no_succ} is equal to
$$
\Phi_n(\eta_{n-1}^{T_{n-1}})^{\otimes 2}(\mathbb{I}_{(\mathsf{B}_{n}^c)^2}\varphi).
$$
Hence we have shown
\begin{eqnarray*}
\mathbb{E}[(\eta_n^{T_n})^{\odot 2}(\varphi)|\mathscr{F}_{n-1}] & = & \Phi_n(\eta_{n-1}^{T_{n-1}})^{\otimes 2}(\mathbb{I}_{\mathsf{B}_{n}^2}\varphi)
+ 2\Phi_n(\eta_{n-1}^{T_{n-1}})^{\otimes 2}(\mathbb{I}_{\mathsf{B}_{n}\times \mathsf{B}_{n}^c}\varphi)
+ \Phi_n(\eta_{n-1}^{T_{n-1}})^{\otimes 2}(\mathbb{I}_{(\mathsf{B}_{n}^c)^2}\varphi)\\ & = &
\Phi_n(\eta_{n-1}^{T_{n-1}})^{\otimes 2}(\varphi).
\end{eqnarray*}
\end{proof}

\end{document}